\documentclass[11pt]{article}
\usepackage{latexsym}
\usepackage{tabularx,booktabs,multirow,delarray,array}
\usepackage{graphicx,amssymb,amsmath,amssymb,mathrsfs}
\usepackage[linesnumbered, vlined, ruled]{algorithm2e}

\newcommand{\VD}{\mbox{$V\!D$}}

\newcommand{\VM}{\mbox{$V\!M$}}
\newcommand{\HM}{\mbox{$H\!M$}}
\newcommand{\HStrip}{\mbox{$H\!Strip$}}

\newcommand{\Vor}{\mbox{$V\!D$}}
\newcommand{\Tri}{\mbox{$T\!r\!i$}}

\newcommand{\SPM}{\mbox{$S\!P\!M$}}

\def\calP{\mathcal{P}}

\def\calM{\mathcal{M}}

\def\spp{$L_1$-SP}
\def\spq{$L_1$-SPM}

\def\lgvd{$L_1$-GVD}

\def\calR{\mathcal{R}}
\def\calF{\mathcal{F}}
\def\bay{bay(\overline{cd})}

\def\canal{canal(x,y)}
\def\c{core}
\newcommand{\PC}{\gamma}
\def\gvd{G\!V\!D}

\def\lemmaspace{\vspace*{-0.00in}}
\def\sectionspace{\vspace*{-0.00in}}
\def\subsectionspace{\vspace*{-0.00in}}

\begin{document}

\baselineskip=14.0pt

\title{
\vspace*{-0.55in} Computing $L_1$
Shortest Paths among Polygonal Obstacles in the Plane\thanks{This
research was supported in part by NSF under Grant CCF-0916606.}}

\author{
Danny Z. Chen\thanks{Department of Computer Science and Engineering,
University of Notre Dame, Notre Dame, IN 46556, USA.
E-mail: {\tt \{dchen, hwang6\}@nd.edu}.
}
\hspace*{0.3in} Haitao Wang\footnotemark[2] \thanks{Corresponding
author.
}}

\date{}

\maketitle

\thispagestyle{empty}

\newtheorem{lemma}{Lemma}
\newtheorem{theorem}{Theorem}
\newtheorem{corollary}{Corollary}
\newtheorem{fact}{Fact}
\newtheorem{definition}{Definition}
\newtheorem{observation}{Observation}
\newtheorem{condition}{Condition}
\newtheorem{property}{Property}
\newtheorem{claim}{Claim}
\newenvironment{proof}{\noindent {\textbf{Proof:}}\rm}{\hfill $\Box$
\rm}

\pagestyle{plain}
\pagenumbering{arabic}
\setcounter{page}{1}

\begin{abstract}
Given a point $s$ and a set of $h$ pairwise disjoint polygonal
obstacles of totally $n$ vertices in the plane, we present a new
algorithm for building an $L_1$ shortest path map of size $O(n)$
in $O(T)$ time and $O(n)$ space such that for any query point $t$,
the length of the $L_1$ shortest obstacle-avoiding path from $s$ to
$t$ can be reported in $O(\log n)$ time and the actual shortest path can be
found in additional time proportional to the number of edges of the
path, where $T$ is the time for triangulating the free space.
It is currently known that
$T=O(n+h\log^{1+\epsilon}h)$ for an arbitrarily small constant $\epsilon>0$.
If the triangulation can be done optimally (i.e., $T=O(n+h\log h)$),
then our algorithm is optimal. Previously, the best algorithm computes 
such an $L_1$ shortest path map in $O(n\log n)$ time and $O(n)$ space.  
Our techniques can
be extended to obtain improved results for other related problems,
e.g., computing the $L_1$ geodesic Voronoi diagram for a set of point
sites in a polygonal domain,
finding shortest paths with fixed orientations, finding
approximate Euclidean shortest paths, etc.
\end{abstract}


\section{Introduction}
\label{sec:intro}

Computing obstacle-avoiding shortest paths in the plane is a
fundamental problem in computational geometry and has many
applications. The Euclidean version that measures the path length
by the Euclidean distance has been well studied (e.g., see
\cite{ref:ChenCo11,ref:ChenCo11Curved,ref:GhoshAn91,ref:HershbergerAn99,ref:InkuluA10,ref:KapoorEf88,ref:KapoorAn97,ref:MitchellSh96,ref:RohnertSh86,ref:StorerSh94}).
In this paper, we consider the $L_1$ version, defined as follows.
Given a point $s$ and a set of $h$ pairwise disjoint polygonal
obstacles, $\calP=\{P_1,P_2,\ldots,P_h\}$, of totally $n$
vertices in the plane, where $s$ is considered as a special point obstacle, the
plane minus the interior of the obstacles is called the {\em free
space} of $\calP$. Two obstacles are pairwise {\em disjoint} if they
do not intersect in their interior. The {\em $L_1$ shortest path
map problem}, denoted by \spq, is to compute a single-source
shortest path map (SPM for short) with $s$ as the {\em source point} such
that for any query point $t$,
an $L_1$ shortest obstacle-avoiding path from $s$ to $t$ can be
obtained efficiently. Note that such a path can consist of any polygonal segments
but the length of each segment of the path is measured by the $L_1$ metric.

We say that an SPM has {\em standard query performances} if for any
query point $t$, the length of the $L_1$ shortest obstacle-avoiding
path from $s$ to $t$ can be reported in $O(\log n)$ time and an
actual shortest path can be found in additional time proportional to
the number of edges (or turns) of the path.

If the input also includes another point $t$ and the problem
only asks for one single $L_1$ shortest path from $s$ to $t$, then we call this
problem version the {\em $L_1$ shortest path problem}, denoted by \spp.

A closely related problem version solvable by our approach is to
find shortest rectilinear paths. A {\em rectilinear path} is a path
each of whose edges is parallel to a coordinate axis and its length
is measured by the Euclidean distances or $L_1$ distances of its
segments (they are the same for rectilinear paths). Rectilinear
shortest paths are used
widely in VLSI design and network wire-routing applications. As shown in
\cite{ref:ClarksonRe87,ref:LarsonFi81,ref:MitchellAn89,ref:MitchellL192},
it is easy to convert an arbitrary polygonal path to a rectilinear
path with the same $L_1$ length. Thus, in this paper, we focus on
computing polygonal paths measured by the $L_1$ distance.

\subsection{Previous Work}

The \spp\ problem has been studied extensively (e.g., see
\cite{ref:ChenSh00,ref:ClarksonRe87,ref:ClarksonRe88,ref:LarsonFi81,ref:MitchellAn89,ref:MitchellL192,ref:WidmayerOn91}).
In general, there are two approaches for solving this problem:
Constructing a sparse ``path preserving" graph (analogous to a
visibility graph), and the continuous Dijkstra paradigm.
Clarkson, Kapoor, and Vaidya \cite{ref:ClarksonRe87} constructed a
graph of $O(n\log n)$ nodes and $O(n\log n)$ edges such that a
shortest path can be found in the graph in $O(n\log^2 n)$ time;
subsequently, they gave an algorithm of $O(n\log^{1.5}n)$ time and
$O(n\log^{1.5}n)$ space \cite{ref:ClarksonRe88}. Based on some
observations, Chen, Klenk, and Tu \cite{ref:ChenSh00} showed that
the problem was solvable in $O(n\log^{1.5}n)$ time and $O(n\log n)$
space. By applying the
continuous Dijkstra paradigm, Mitchell
\cite{ref:MitchellAn89,ref:MitchellL192} solved the problem in
$O(n\log n)$ time and $O(n)$ space.
An $O(n+h\log h)$ time lower bound can be established for
solving \spp\ (e.g., based on the results in
\cite{ref:deRezendeRe85}). Hence, Mitchell's algorithm is worst-case
optimal. Recently, by using a corridor structure and building a
smaller path preserving graph, Inkulu and Kapoor
\cite{ref:InkuluPl09} solved the \spp\ problem
in $O(n+h\log^{1.5}n)$ time and $O(n+h\log^{1.5}h)$ space.

For the query version of the problem, i.e., \spq, Mitchell's
algorithm \cite{ref:MitchellAn89,ref:MitchellL192} builds an SPM of
size $O(n)$ in $O(n\log n)$ time and $O(n)$ space with the standard query
performances.

In addition, for the {\em convex case} where all
polygonal obstacles in $\calP$ are convex, to our best knowledge,
we are not aware of any previous better results than those mentioned above.


\sectionspace
\subsection{Our Results}

We present an algorithm for \spq\ that builds an SPM of size
$O(n)$ in $O(T)$ time and $O(n)$ space with the standard query
performances, where $T$ always refers to the time for triangulating the free space
of $\calP$ in the paper. 
It is obvious to see that given an SPM, we can always
add $h-1$ line segments in the
free space to connect the obstacles in $\calP$ together to obtain a
single simple polygon and then triangulate the free space, in totally $O(n)$ time
\cite{ref:Bar-YehudaTr94,ref:ChazelleTr91}. It is currently known that
$T=\Omega(n+h\log h)$ and $T=O(n+h\log^{1+\epsilon}h)$
\cite{ref:Bar-YehudaTr94}, where $\epsilon$ is
an arbitrarily small positive constant. Therefore, we
essentially solve \spq\ in $\Theta(T)$ time. In other words, our
result shows that 
building an SPM is equivalent to triangulating the free space
of $\calP$ in terms of the running time.

Our approach uses Mitchell's algorithm
\cite{ref:MitchellAn89,ref:MitchellL192} as a procedure and further explores
the corridor structure of $\calP$ \cite{ref:KapoorAn97}. One
interesting observation we found is that to find an $L_1$ shortest
path among convex obstacles, it is sufficient to consider only the at
most four extreme vertices (along the horizontal and vertical
directions) of each obstacle (these vertices define a {\em core} for each
obstacle). Mitchell's algorithm is then applied to these cores, which
takes only $O(h\log h)$ time.
More work needs to be
done for computing an SPM. For example, one
key result we have is that we give an $O(n'+m')$ time algorithm for a 
special case of
constructing the $L_1$ geodesic Voronoi diagram in a simple polygon of
$n'$ vertices for $m'$ weighted point sites, where the sites all lie outside
the polygon and influence the polygon through one (open) edge (see
Fig.~\ref{fig:geoVoi}).
We are not aware of any specific previous work on this problem, although an
$O((n'+m')\log (n'+m'))$ time solution may be obtained by standard
techniques. Our linear time algorithm, which is
clearly optimal, may be interesting in its own right.

For the convex case where all obstacles in $\calP$ are
convex, we can find a shortest $s$-$t$ path in $O(n+h\log h)$ time and
$O(n)$ space since the triangulation can be done in $O(n+h\log h)$
time (e.g., by the approaches in
\cite{ref:Bar-YehudaTr94,ref:HertelFa85}); this is optimal.
A by-product of our
techniques, which may be a little ``surprising", is that in $O(n+h\log h)$ time
and $O(n)$ space, we can build
an SPM of size $O(h)$ (instead of $O(n)$) such that the shortest path
{\em length} queries are answered in $O(\log h)$ time
each (instead of $O(\log n)$ time).

\begin{figure}[t]
\begin{minipage}[t]{\linewidth}
\begin{center}
\includegraphics[totalheight=1.5in]{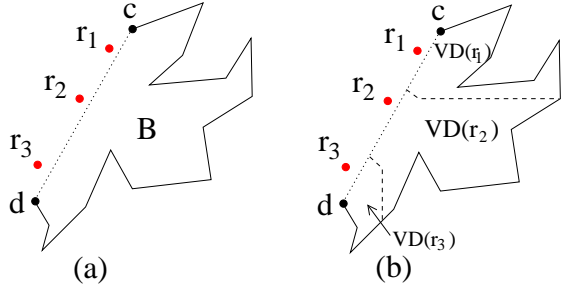}
\caption{\footnotesize (a) Three weighted sites (in red) and
a simple polygon $B$ with an open edge $\overline{cd}$. The goal is to compute
the $L_1$ geodesic Voronoi diagram in $B$ with respect to the three
sites which influence $B$ only through the edge $\overline{cd}$. (b)
Illustrating a possible solution: $B$ is partitioned into
three Voronoi regions $\Vor(r_i)$ for each $r_i$, $1\leq i\leq 3$.}
\label{fig:geoVoi}
\end{center}
\end{minipage}
\end{figure}

\subsection{Applications}
\label{sec:application}

Our techniques can be extended to solve other problems.

The {\em $L_1$ geodesic Voronoi diagram problem}, denoted by \lgvd, is defined 
as follows.  Given an obstacle set $\calP$ and a set
of $m$ point sites in the free space, compute the geodesic Voronoi
diagram for the $m$ point sites under the $L_1$ distance metric
among the obstacles in $\calP$.  Mitchell \cite{ref:MitchellAn89,ref:MitchellL192},
solves the \lgvd\ problem in $O((n+m)\log(n+m))$ time.
Our approach can compute it in $O(T'+n+(m+h)\log(m+h))$ time, where $T'$
is the time for triangulating the free space along with the $m$ point
sites. It is known that $T'=O(n+(m+h)\log^{1+\epsilon}(m+h))$
\cite{ref:Bar-YehudaTr94} or alternatively we can obtain 
$T'=O(n+h\log^{1+\epsilon}h+m\log n)$. Note that
when applying our algorithm to a single simple polygon $P$ of $n$ vertices,
the $L_1$ geodesic Voronoi diagram for $m$ point sites in $P$ can be
obtained in $O(n+m\log^{1+\epsilon} m)$ or $O(n+m(\log n+\log m))$
time. In comparison, the Euclidean version of the one simple polygon
case was solved in $O((n+m)\log (n+m))$ time \cite{ref:PapadopoulouA98}.

We also give better results
for the shortest path problem in ``fixed orientation metrics"
\cite{ref:MitchellAn89,ref:MitchellL192,ref:WidmayerOn87}, for which
a sought path is allowed to follow only a given set of orientations.
For a number $c$ of given orientations,
Mitchell's algorithm \cite{ref:MitchellAn89,ref:MitchellL192} finds
such a shortest path in $O(cn\log n)$ time and $O(cn)$ space, and
our algorithm takes $O(n+ h\log^{1+\epsilon}h+ c^2h\log ch)$ time and
$O(n+c^2h)$ space.  In addition, our approach also leads to an
$O(n+h\log^{1+\epsilon}h +(1/\delta)h\log \frac{h}{\sqrt{\delta}})$ time algorithm for
computing a $\delta$-optimal Euclidean shortest path among polygonal
obstacles for any constant $\delta>0$. For this problem, Mitchell's
algorithm \cite{ref:MitchellAn89,ref:MitchellL192} takes
$O((\sqrt{1/\delta})n\log n)$ time, and Clarkson's algorithm
\cite{ref:ClarksonAp87} runs in $O((1/\delta) n\log n)$ time.

\subsectionspace
\section{An Overview of Our Approaches}

In this section, we give an overview of our approaches as well as the
organization of this paper.
Denote by $\calF$ the free space of $\calP$.
We begin with our algorithm for the convex case, which is
a key procedure for solving the general problem.

We first discuss the \spp\ problem.
In the convex case, each obstacle in $\calP=\{P_1,P_2,\ldots,P_{h}\}$
is convex. For each $P_i\in \calP$, we compute its {\em
core}, denoted by $core(P_i)$, which is a simple polygon by
connecting the topmost, leftmost, bottommost, and rightmost points
of $P_i$.  Let $core(\calP)$ be the set of all $h$ cores of $\calP$.
For any point $t$ in the free space $\calF$,
we show that given any shortest $s$-$t$ path avoiding all
cores in $core(\calP)$, we can find in $O(n)$ time a shortest
$s$-$t$ path avoiding all obstacles in $\calP$ with the same $L_1$
length. Based on this observation, our algorithm
has two main steps: (1) Apply Mitchell's algorithm
\cite{ref:MitchellAn89,ref:MitchellL192} on $core(\calP)$ to compute
a shortest $s$-$t$ path $\pi_{\c}(s,t)$ avoiding the cores in
$core(\calP)$, which takes $O(h\log h)$ time since each core in
$core(\calP)$ has at most four vertices; (2) based on
$\pi_{\c}(s,t)$, compute a shortest $s$-$t$ path avoiding all
obstacles in $\calP$ in $O(n)$ time. This algorithm takes overall
$O(n+h\log h)$ time and $O(n)$ space.

To build an SPM in $\calF$ (with respect to the source point $s$),
similarly, we first apply Mitchell's algorithm
on $core(\calP)$ to compute an SPM of $O(h)$ size in the free space with respect to
all cores, which can be done in $O(n+h\log h)$ time and $O(n)$ space.
Based on the above SPM, in additional $O(n)$ time,
we are able to compute an SPM in $\calF$. Our results for the convex
case are given in Section \ref{sec:convexcase}.

For the general problem where the obstacles in $\calP$ are not
necessarily convex, based on a triangulation of the free space $\calF$,
we first compute a {\em corridor structure} \cite{ref:KapoorAn97},
which consists of $O(h)$ corridors and $O(h)$ junction triangles.
Each corridor possibly has a {\em corridor path}. As in
\cite{ref:KapoorAn97}, the corridor
structure can be used to partition the plane into a set $\calP'$
of $O(h)$ pairwise disjoint convex polygons of totally $O(n)$
vertices such that a shortest $s$-$t$ path in $\calF$ is
a shortest $s$-$t$ path avoiding the convex polygons in $\calP'$ and
possibly containing some corridor paths. All corridor paths
are contained in the polygons of $\calP'$. Thus, in addition to the
corridor paths, finding a shortest path is reduced to an instance
of the convex case. By incorporating the corridor path information
into Mitchell's continuous Dijkstra paradigm
\cite{ref:MitchellAn89,ref:MitchellL192}, our algorithm for
the convex case can be modified to find a shortest path in $O(T)$ time.
The above algorithm is presented in Section \ref{sec:general}.

Sections \ref{sec:spm}, \ref{sec:bay}, and \ref{sec:canal} are
together devoted to compute an SPM in $\calF$ (Section \ref{sec:spm}
outlines the algorithm). 
We use the corridor structure to partition
$\calF$ into the {\em ocean} $\calM$, {\em bays}, and {\em canals}. While the
ocean $\calM$ may be multiply connected, every bay or canal is a
simple polygon. Each bay has a single common boundary edge with
$\calM$ and each canal has two common boundary edges with $\calM$.
But two bays or two canals, or a bay and a canal do not share any
boundary edge. A common boundary edge of a bay (or canal) with
$\calM$ is called a {\em gate}. Thus each bay has one gate and each
canal has two gates. Further, the ocean $\calM$ is exactly the free
space with respect to the convex polygonal set $\calP'$.
By modifying our algorithm for the convex case, we can compute an SPM
in $\calM$ in $O(T)$ time. This part is discussed in Section
\ref{sec:spm}.

Denote by $\SPM(\calM)$ the SPM in $\calM$.
To obtain an SPM in $\calF$, we need to ``expand" $\SPM(\calM)$ into all bays and canals through their gates.
Here, a {\em challenging subproblem} is to solve efficiently a special case of the
(additively) weighted $L_1$ geodesic Voronoi diagram problem on a simple polygon
$B$: The weighted point sites all lie outside $B$ and influence $B$
through one (open) edge (e.g., see Fig.~\ref{fig:geoVoi}). The
subproblem models the procedure of expanding $\SPM(\calM)$ into a bay, where
the polygon $B$ is the bay, the point sites are obstacle vertices
in $\calM$, the weight of each site is the length of its shortest path
to the source point $s$, and the edge of the polygon (e.g., $\overline{cd}$
in Fig.~\ref{fig:geoVoi}) is the gate of the bay.
As discussed before, we give a linear time solution for this
subproblem in Section \ref{sec:bay}.
Note that although our presentation for solving the subproblem is long
and technically complicated, the algorithm itself is
simple and easy to implement; our effort is mostly for simplifying
the algorithm and showing its correctness.

Expanding $\SPM(\calM)$ into canals, which is discussed in
Section \ref{sec:canal}, is also done in linear time by using our
solution for the above subproblem as a main procedure. In summary,
given $\SPM(\calM)$, computing an SPM for the entire free space
$\calF$ takes additional $O(n)$ time.

We discuss a little more about the above challenging subproblem.
The problem may not look ``challenging"
at all as it can be solved by many existing techniques. For example,
one may attempt to use the continuous Dijkstra approach
\cite{ref:MitchellAn89,ref:MitchellL192} to let the ``wavelet" enter
into the bays/canals. However, that would lead to an $O((n'+m')\log
(n'+m'))$ time solution for the subproblem since it takes
logarithmic time to process each event, where $n'$ is the number of
vertices of $B$ and $m'$ is the number of weighted sites, and
consequently it
would take an overall $O(n\log n)$ time for building an SPM in
$\calF$. One may also want to use a sweeping algorithm
\cite{ref:FortuneA87}, which would also
lead to an $O((n'+m')\log (n'+m'))$ time solution since again it
takes logarithmic time to process each event. In addition, the
divide-and-conquer approach \cite{ref:ShamosCl75} would also take
$O((n'+m')\log (n'+m'))$ time since the merge procedure takes linear
time. Our algorithm for the subproblem, which can be viewed as an
incremental approach, takes $O(n'+m')$ time. Incremental approaches
have been widely used in geometric algorithms, and normally they can
result in good randomized algorithms. Incremental approaches have
also been used for constructing Voronoi diagrams, which usually
take quadratic time. Our result demonstrates that incremental
approaches are able to yield optimal deterministic solutions for
building Voronoi diagrams, and the success of it hinges on
discovering many geometric properties of the problem. We should point
out that our techniques for solving the challenging subproblem are quite
independent of other parts of the paper.

In Section \ref{sec:fixed}, we generalize our techniques to solve some
related problems discussed in Section \ref{sec:application}.
Section \ref{sec:con} concludes the paper.

As in \cite{ref:MitchellAn89,ref:MitchellL192}, for simplicity of
discussion, we assume that the free space $\calF$ is connected and the
point $t$ is always in $\calF$ (thus, a feasible $s$-$t$ path always
exists), and no two obstacle vertices lie on the same horizontal or
vertical line. In the rest of this paper, unless otherwise stated, a
shortest path always refers to an $L_1$ shortest path and a length
is always in the $L_1$ metric.


\sectionspace
\section{Shortest Paths among Convex Obstacles}
\label{sec:convexcase}

In this section, we give our algorithms for the convex case, which are
also used for the general case in later sections.
Let $\calP'=\{P_1',P_2'\ldots,P_{h}'\}$ be a set of $h$ pairwise
disjoint convex polygonal obstacles of totally $n$ vertices. With respect to the source point $s$, our algorithm builds an SPM of $O(n)$ size with standard query performances in $O(n+h\log h)$ time and $O(n)$ space.

\sectionspace
\subsection{Notation and Observations}

For each convex polygon $P_i'\in \calP'$, we define its {\em core}, denoted by
$\c(P_i')$, as the simple polygon by connecting the leftmost,
topmost, rightmost, and bottommost vertices of $P_i'$ with line
segments (see Fig.~\ref{fig:core}). Note that $\c(P_i')$ is
contained in $P_i'$ and has at most four edges. Let $\c(\calP')$ be
the set of the cores of all obstacles in $\calP'$. Consider a point $t$ in the free space $\calF$. A key observation (to be proved)
is that a shortest $s$-$t$ path avoiding the cores in $\c(\calP')$
corresponds to a shortest $s$-$t$ path avoiding the obstacles in
$\calP'$ with the same $L_1$ length.
Note that a path avoiding the cores in $\c(\calP')$ may intersect
the interior of some obstacles in $\calP'$.

To prove the above key observation, we first define some concepts.
Consider an obstacle $P_i'$ and $\c(P_i')$. For each edge
$\overline{ab}$ of $\c(P_i')$ with vertices $a$ and $b$, if
$\overline{ab}$ is not an edge of $P_i'$, then it divides $P_i'$
into two polygons, one of them containing $\c(P_i')$; we call the
one that does not contain $\c(P_i')$ an {\em ear} of $P'_i$ {\em
based on} $\overline{ab}$, denoted by $ear(\overline{ab})$ (see
Fig.~\ref{fig:core}). If $\overline{ab}$ is also an edge of $P_i$,
then $ear(\overline{ab})$ is not defined. Note that
$ear(\overline{ab})$ has only one edge bounding $\c(P_i')$, i.e.,
$\overline{ab}$, which we call its {\em core edge}. The other edges
of $ear(\overline{ab})$ are on the boundary of $P_i'$, which we call
{\em obstacle edges}. There are two paths between $a$ and $b$ along
the boundary of $ear(\overline{ab})$: One path is the core edge
$\overline{ab}$ and the other consists of all its obstacle edges. We
call the latter path the {\em obstacle path} of the ear. A line
segment is {\em positive-sloped} (resp., {\em negative-sloped}) if
its slope is positive (resp., negative). An ear is {\em
positive-sloped} (resp., {\em negative-sloped}) if its core edge is
positive-sloped (resp., negative-sloped). Note that by our
assumption no two obstacle vertices lie on the same horizontal or
vertical line, and thus no ear has a horizontal or vertical core
edge. A point $p$ is {\em higher} (resp., {\em lower}) than another
point $q$ if the $y$-coordinate of $p$ is no smaller (resp., no
larger) than that of $q$. The next observation is self-evident.

\begin{figure}[t]
\begin{minipage}[t]{0.48\linewidth}
\begin{center}
\includegraphics[totalheight=1.2in]{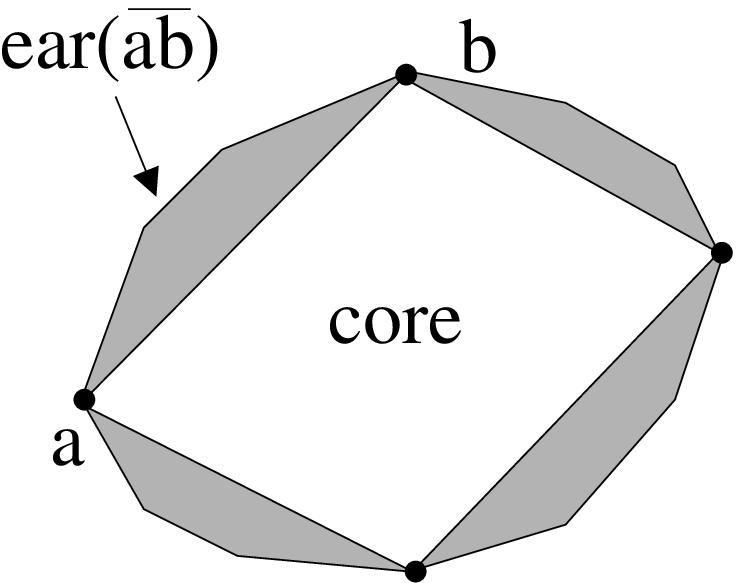}
\caption{\footnotesize Illustrating the core and ears of a convex
obstacle; $ear(\overline{ab})$ is indicated.} \label{fig:core}
\end{center}
\end{minipage}
\hspace*{0.04in}
\begin{minipage}[t]{0.5\linewidth}
\begin{center}
\includegraphics[totalheight=1.2in]{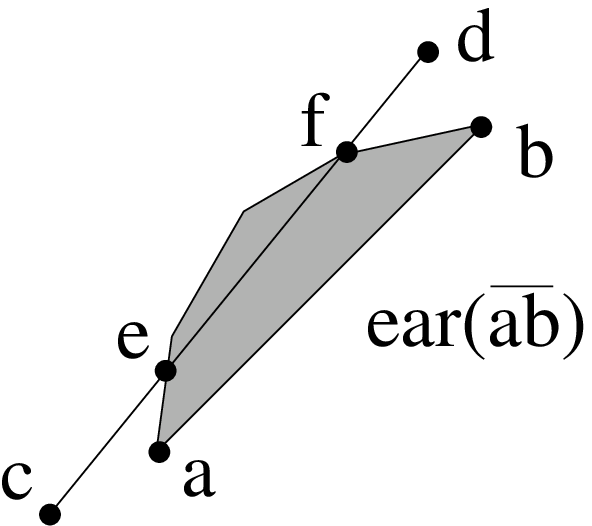}
\caption{\footnotesize The line segment $\overline{cd}$ penetrates
$ear(\overline{ab})$; $\overline{cd}$ intersects the obstacle path
of $ear(\overline{ab})$ at $e$ and $f$.} \label{fig:travelthrough}
\end{center}
\end{minipage}
\vspace*{-0.15in}
\end{figure}

\begin{observation}\label{obser:10}
For any ear, its obstacle path is monotone in both the $x$- and
$y$-coordinates. Specifically, consider an ear $ear(\overline{ab})$
and suppose the vertex $a$ is lower than the vertex $b$. If
$ear(\overline{ab})$ is positive-sloped, then the obstacle path from
$a$ to $b$ is monotonically increasing in both the $x$- and
$y$-coordinates; if it is negative-sloped, then the obstacle path
from $a$ to $b$ is monotonically decreasing in the $x$-coordinates
and monotonically increasing in the $y$-coordinates.
\end{observation}

For an ear $ear(\overline{ab})$ and a line segment $\overline{cd}$,
we say that $\overline{cd}$ {\em penetrates} $ear(\overline{ab})$ if
the following hold (see Fig.~\ref{fig:travelthrough}): (1)
$\overline{cd}$ intersects the interior of $ear(\overline{ab})$, (2)
neither $c$ nor $d$ is in the interior of $ear(\overline{ab})$, and
(3) $\overline{cd}$ does not intersect the core edge $\overline{ab}$
at its interior. The next lemma will be useful later.

\begin{lemma}\label{lem:10}
Suppose a line segment $\overline{cd}$ penetrates an ear
$ear(\overline{ab})$. If $\overline{cd}$ is positive-sloped (resp.,
negative-sloped), then $ear(\overline{ab})$ is also positive-sloped
(resp., negative-sloped).
\end{lemma}
\begin{proof}
We only prove the case when $\overline{cd}$ is positive-sloped since
the other case is similar.

Assume to the contrary that $ear(\overline{ab})$ is negative-sloped.
Without loss of generality (WLOG), we assume $a$ is lower than $b$. By
Observation \ref{obser:10}, the obstacle path of
$ear(\overline{ab})$ from $a$ to $b$ is monotonically decreasing in
the $x$-coordinates. Thus, the rightmost point and leftmost point of
$ear(\overline{ab})$ are $a$ and $b$, respectively. Note that
$ear(\overline{ab})$ is contained in the region between the two
vertical lines passing through $a$ and $b$. Since $\overline{cd}$ is
positive-sloped and $\overline{ab}$ is negative-sloped, if
$\overline{cd}$ intersects an interior point of
$ear(\overline{ab})$, then $\overline{cd}$ must cross
$\overline{ab}$ at an interior point. But since $\overline{cd}$
penetrates $ear(\overline{ab})$, $\overline{cd}$ cannot intersect
any interior point of $\overline{ab}$. Hence, we have a
contradiction. The lemma thus follows.
\end{proof}

Clearly, if $\overline{cd}$ penetrates the ear $ear(\overline{ab})$,
then $\overline{cd}$ intersects the boundary of $ear(\overline{ab})$
at two points and both points lie on the obstacle path of
$ear(\overline{ab})$ (e.g., see Fig.~\ref{fig:travelthrough}).

\begin{lemma}\label{lem:20}
Suppose a line segment $\overline{cd}$ penetrates an ear
$ear(\overline{ab})$. Let $e$ and $f$ be the two points on the
obstacle path of $ear(\overline{ab})$ that $\overline{cd}$
intersects. Then the $L_1$ length of the line segment
$\overline{ef}$ is equal to that of the portion of the obstacle path
of $ear(\overline{ab})$ between $e$ and $f$ (see
Fig.~\ref{fig:travelthrough}).
\end{lemma}
\begin{proof}
WLOG, suppose $\overline{cd}$ is
positive-sloped and $e$ is lower than $f$. By Lemma \ref{lem:10},
$ear(\overline{ab})$ is also positive-sloped. The segment
$\overline{ef}$ from $e$ to $f$ is monotonically increasing in both
the $x$- and $y$-coordinates. Denote by $\widehat{ef}$ the portion
of the obstacle path of $ear(\overline{ab})$ between $e$ and $f$.
Since $ear(\overline{ab})$ is positive-sloped, by Observation
\ref{obser:10}, the portion $\widehat{ef}$ from $e$ to $f$ is
monotonically increasing in both the $x$- and $y$-coordinates.
Therefore, the $L_1$ lengths of $\overline{ef}$ and $\widehat{ef}$
are equal. The lemma thus follows.
\end{proof}

If $\overline{cd}$ penetrates $ear(\overline{ab})$, then by Lemma
\ref{lem:20}, we can obtain another path from $c$ to $d$ by
replacing $\overline{ef}$ with the portion of the obstacle path of
$ear(\overline{ab})$ between $e$ and $f$ such that the new path has
the same $L_1$ length as $\overline{cd}$ and the new path does not
intersect the interior of $ear(\overline{ab})$.

The results in the following lemma have been proved in
\cite{ref:MitchellAn89,ref:MitchellL192}.

\begin{lemma}{\em \cite{ref:MitchellAn89,ref:MitchellL192}}\label{lem:30}
There exists a shortest $s$-$t$ path in the free space such that if
the path makes a turn at a point $p$, then $p$ is an obstacle
vertex.
\end{lemma}

We call a shortest path that satisfies the property in Lemma
\ref{lem:30} a {\em vertex-preferred shortest path}. Mitchell's
algorithm \cite{ref:MitchellAn89,ref:MitchellL192} can find a
vertex-preferred shortest $s$-$t$ path. Denote by $\Tri(\calP')$ a
triangulation of the free space and the space inside all obstacles.
Note that the free space can be triangulated in $O(n+h\log h)$ time
\cite{ref:Bar-YehudaTr94,ref:HertelFa85} and the space inside all
obstacles can be triangulated in totally $O(n)$ time
\cite{ref:ChazelleTr91}. Hence, $\Tri(\calP')$ can be computed in
$O(n+h\log h)$ time. The next lemma gives our key observation.

\lemmaspace
\begin{lemma}\label{lem:40}
Given a vertex-preferred shortest $s$-$t$ path that avoids the
polygons in $\c(\calP')$, we can find in $O(n)$ time a shortest
$s$-$t$ path with the same $L_1$ length that avoids the obstacles in
$\calP'$.
\end{lemma}
\lemmaspace
\begin{proof}
Consider a vertex-preferred shortest $s$-$t$ path for $\c(\calP')$,
denoted by $\pi_{\c}(s,t)$. Suppose it makes turns at
$p_1,p_2,\ldots,p_k$, ordered from $s$ to $t$ along the path, and
each $p_i$ is a vertex of a core in $\c(\calP')$. Let $p_0=s$ and
$p_{k+1}=t$. Then for each $i = 0,1,\ldots, k$, the portion of
$\pi_{\c}(s,t)$ from $p_i$ to $p_{i+1}$ is the line segment
$\overline{p_ip_{i+1}}$, which does not intersect the interior of
any core in $\c(\calP')$. Below, we first show that we can find a
path from $p_i$ to $p_{i+1}$ such that it avoids the obstacles in
$\calP'$ and has the same $L_1$ length as $\overline{p_ip_{i+1}}$.

If $\overline{p_ip_{i+1}}$ does not intersect the interior of any
obstacle in $\calP'$, then we are done with $\overline{p_ip_{i+1}}$.
Otherwise, because $\overline{p_ip_{i+1}}$ avoids $\c(\calP')$, it
can intersects only the interior of some ears. Consider any such ear
$ear(\overline{ab})$. Below, we prove that $\overline{p_ip_{i+1}}$
penetrates $ear(\overline{ab})$.

First, we already know that $\overline{p_ip_{i+1}}$ intersects the
interior of $ear(\overline{ab})$. Second, it is obvious that neither
$p_i$ nor $p_{i+1}$ is in the interior of $ear(\overline{ab})$. It
remains to show that $\overline{p_ip_{i+1}}$ cannot intersect the
core edge $\overline{ab}$ of $ear(\overline{ab})$ at the interior of
$\overline{ab}$. Denote by $A'\in\calP'$ the obstacle that contains
$ear(\overline{ab})$. The interior of $\overline{ab}$ is in the
interior of $A'$. Since $\overline{p_ip_{i+1}}$ does not intersect
the interior of $A'$, $\overline{p_ip_{i+1}}$ cannot intersect
$\overline{ab}$ at its interior. Therefore, $\overline{p_ip_{i+1}}$
penetrates $ear(\overline{ab})$.

Recall that we have assumed that no two obstacle vertices lie on the
same horizontal or vertical line. Since both $p_i$ and $p_{i+1}$ are
obstacle vertices, the segment $\overline{p_ip_{i+1}}$ is either
positive-sloped or negative-sloped. WLOG,
assume $\overline{p_ip_{i+1}}$ is positive-sloped. By Lemma
\ref{lem:10}, $ear(\overline{ab})$ is also positive-sloped. Let $e$
and $f$ denote the two intersection points between
$\overline{p_ip_{i+1}}$ and the obstacle path of
$ear(\overline{ab})$, and $\widehat{ef}$ denote the portion of the
obstacle path of $ear(\overline{ab})$ between $e$ and $f$. By Lemma
\ref{lem:20}, we can replace the line segment $\overline{ef}$
($\subseteq \overline{p_ip_{i+1}}$) by $\widehat{ef}$ to obtain a
new path from $p_i$ to $p_{i+1}$ such that the new path has the same
$L_1$ length as $\overline{p_ip_{i+1}}$. Further, as a portion of
the obstacle path of $ear(\overline{ab})$, $\widehat{ef}$ is a
boundary portion of the obstacle $A'$ that contains
$ear(\overline{ab})$, and thus $\widehat{ef}$ does not intersect the
interior of any obstacle in $\calP'$.

By processing each ear whose interior is intersected by
$\overline{p_ip_{i+1}}$ as above, we find a new path from $p_i$ to
$p_{i+1}$ such that the path has the same $L_1$ length as
$\overline{p_ip_{i+1}}$ and the path does not intersect the interior
of any obstacle in $\calP'$.

By processing each segment $\overline{p_ip_{i+1}}$ in
$\pi_{\c}(s,t)$ as above for $i=0,1,\ldots,k$, we obtain another
$s$-$t$ path $\pi(s,t)$ such that the $L_1$ length of $\pi(s,t)$ is
equal to that of $\pi_{\c}(s,t)$ and $\pi(s,t)$ avoids all obstacles
in $\calP'$. Below, we show that $\pi(s,t)$ is a shortest $s$-$t$
path avoiding the obstacles in $\calP'$.

Since each core in $\c(\calP')$ is contained in an obstacle in
$\calP'$, the length of a shortest $s$-$t$ path avoiding
$\c(\calP')$ cannot be longer than that of a shortest $s$-$t$ path
avoiding $\calP'$. Because the length of $\pi(s,t)$ is equal to that
of $\pi_{\c}(s,t)$ and $\pi_{\c}(s,t)$ is a shortest $s$-$t$ path
avoiding $\c(\calP')$, $\pi(s,t)$ is a shortest $s$-$t$ path
avoiding $\calP'$.

Note that the above discussion also provides a way to construct
$\pi(s,t)$, which can be easily done in $O(n)$ time with the help of
the triangulation $\Tri(\calP')$. The lemma thus follows.
\end{proof}

Since each core in $\c(\calP')$ is contained in an obstacle in
$\calP'$, the corollary below follows from Lemma \ref{lem:40}
immediately.

\begin{corollary}
A shortest $s$-$t$ path avoiding the obstacles in $\calP'$ is a
shortest $s$-$t$ path avoiding the cores in $\c(\calP')$.
\end{corollary}

\subsection{Computing a Single Shortest Path}

Based on Lemma \ref{lem:40}, our algorithm for finding a single shortest $s$-$t$ path works as follows: (1)
Apply Mitchell's algorithm \cite{ref:MitchellAn89,ref:MitchellL192}
on $\c(\calP')$ to find a vertex-preferred shortest $s$-$t$ path
avoiding the cores in $\c(\calP')$; (2) by Lemma \ref{lem:40}, find
a shortest $s$-$t$ path that avoids the obstacles in $\calP'$. The
first step takes $O(h\log h)$ time and $O(h)$ space since the cores
in $\c(\calP')$ have totally $O(h)$ vertices. The second step takes
$O(n)$ time and $O(n)$ space.

\begin{theorem}\label{theo:10}
Given a set of $h$ pairwise disjoint convex polygonal obstacles of
totally $n$ vertices in the plane, we can find an $L_1$ shortest
path between two points in the free space in $O(n+h\log h)$ time and
$O(n)$ space.
\end{theorem}

\subsection{Computing the Shortest Path Map}

In this subsection, we compute the SPM for $\calP'$.
Mitchell's algorithm \cite{ref:MitchellAn89,ref:MitchellL192} can
compute an $O(n)$ size SPM with the standard query performances
in $O(n\log n)$ time and $O(n)$ space.

By applying Mitchell's algorithm
\cite{ref:MitchellAn89,ref:MitchellL192} on the core set
$\c(\calP')$, we can compute an $O(h)$ size SPM in $O(h\log h)$ time and $O(h)$ space, denoted by $\SPM(\c(\calP'),s)$.
With a planar point location data structure
\cite{ref:EdelsbrunnerOp86,ref:KirkpatrickOp83}, for any query point
$t$ in the free space $\calF$, the length of a shortest
$s$-$t$ path avoiding $\c(\calP')$ can be reported in $O(\log h)$
time, which is also the length of a shortest $s$-$t$ path avoiding
$\calP'$ by Lemma \ref{lem:40}. We thus have the following result.

\lemmaspace
\begin{theorem}\label{theo:20}
Given a set of $h$ pairwise disjoint convex polygonal obstacles of
totally $n$ vertices in the plane, in $O(n+h\log h)$ time and $O(n)$ space, we can
construct a shortest path map of size $O(h)$ with respect to a
source point $s$, such that the length of an $L_1$ shortest path
between $s$ and any query point in the free space can be reported in
$O(\log h)$ time.
\end{theorem}
\lemmaspace

The result in Theorem \ref{theo:20} is superior to Mitchell's
algorithm \cite{ref:MitchellAn89,ref:MitchellL192} in three aspects,
i.e., the preprocessing time, the SPM size, and the length query
time.  However, with the SPM for Theorem \ref{theo:20}, an actual
shortest path avoiding $\calP'$ between $s$ and a query point $t$
cannot be reported in additional time proportional to the number of
turns of the path, although we can use this SPM to report an actual
shortest path $\pi_{\c}(s,t)$ between $s$ and $t$ avoiding
$\c(\calP')$ in additional time proportional to the number of turns
of $\pi_{\c}(s,t)$ and then find an actual shortest path avoiding
$\calP'$ between $s$ and $t$ in another $O(n)$ time using
$\pi_{\c}(s,t)$ by Lemma \ref{lem:40}.

To process queries on actual shortest paths avoiding $\calP'$
efficiently, in Lemma \ref{lem:50} below, using
$\SPM(\c(\calP'),s)$, we compute an SPM for $\calP'$, denoted by
$\SPM(\calM)$, of $O(n)$ size, which has the standard query performances, i.e.,  answers a shortest path length query in $O(\log n)$ time and reports an actual path in
additional time proportional to the number of turns of the path.

\lemmaspace
\begin{lemma}\label{lem:50}
Given the shortest path map $\SPM(\c(\calP'),s)$ for the core set
$\c(\calP')$, we can compute a shortest path map $\SPM(\calM)$
for the obstacle set $\calP'$ in $O(n)$ time (with the help of the
triangulation $\Tri(\calP')$).
\end{lemma}
\lemmaspace

\begin{proof}
Note that the polygons in $\calP'$ are pairwise disjoint in their
interior. For simplicity of discussion in this proof, we assume that
any two different polygons in $\calP'$ have disjoint interior as
well as disjoint boundaries.

Consider a cell $C_{core}(r)$ with the root $r$ in
$\SPM(\c(\calP'),s)$. Recall that $r$ is always a vertex of a core
in $\c(\calP')$ and all points in $C_{core}(r)$ are visible to $r$
with respect to $\c(\calP')$
\cite{ref:MitchellAn89,ref:MitchellL192}. In other words, for any
point $p$ in the cell $C_{core}(r)$, the line segment
$\overline{rp}$ is contained in $C_{core}(r)$, and further, there
exists a shortest $s$-$p$ path avoiding $\c(\calP')$ that contains
$\overline{rp}$.

Denote by $\calF(\calP')$ (resp., $\calF(\c(\calP'))$) the free
space with respect to $\calP'$ (resp., $\c(\calP')$). Note that the
cell $C_{core}(r)$ is a simple polygon in $\calF(\c(\calP'))$. We
assume that $C_{core}(r)$ contains some points in $\calF(\calP')$
since otherwise we do not need to consider $C_{core}(r)$.

The cell $C_{core}(r)$ may intersect some ears. In other words,
certain space in $C_{core}(r)$ may be occupied by some ears. Let
$C(r)$ be the subregion of $C_{core}(r)$ by removing from
$C_{core}(r)$ the space occupied by all ears except their obstacle
paths. Thus $C(r)$ lies in $\calF(\calP')$. However, for each point
$p\in C(r)$, $p$ may not be visible to $r$ with respect to $\calP'$.
Our task here is to further decompose $C(r)$ into a set of {\em SPM
regions} such that each such region has a root visible to all points
in the region with respect to $\calP'$; further, we need to make
sure that each point $q$ in an SPM region has a shortest path in
$\calF(\calP')$ from $s$ that contains the line segment connecting
$q$ and the root of the region. For this, we first show that $C(r)$
is a connected region.

To show that $C(r)$ is connected, it suffices to show that for
any point $p\in C(r)$, there is a path in $C(r)$ that connects $r$
and $p$. Consider an arbitrary point $p\in C(r)$. Since $p\in
C_{core}(r)$, $\overline{rp}$ is in $C_{core}(r)$ and there is a
shortest path in $\calF(\c(\calP'))$ from $s$ to $p$ that contains
$\overline{rp}$. If the segment $\overline{rp}$ does not intersect
the interior of any ear, then we are done since $\overline{rp}$ is
in $C(r)$. If $\overline{rp}$ intersects the interior of some ears,
then let $ear(\overline{ab})$ be one of such ears. By the proof of
Lemma \ref{lem:40}, $\overline{rp}$ penetrates $ear(\overline{ab})$.
Let $e$ and $f$ be the two points on the obstacle path of
$ear(\overline{ab})$ that $\overline{rp}$ intersects, and
$\widehat{ef}$ be the portion of the obstacle path between $e$ and
$f$. Note that if $\overline{rp}$ is horizontal or vertical, then it
cannot penetrate $ear(\overline{ab})$ due to the monotonicity of its
obstacle path by Observation \ref{obser:10}. WLOG,
assume $\overline{rp}$ is positive-sloped. Then by Lemma
\ref{lem:20}, $ear(\overline{ab})$ is also positive-sloped. Recall
that $e$ and $f$ lie on $\overline{rp}$. WLOG,
assume $r$ is higher than $p$ and $f$ is higher than $e$. Then the
segment $\overline{ef}$ from $e$ to $f$ is monotonically increasing
in both the $x$- and $y$-coordinates. By Observation \ref{obser:10},
the obstacle path portion $\widehat{ef}$ from $e$ to $f$ is also
monotonically increasing in both the $x$- and $y$-coordinates. As in
the proof of Lemma \ref{lem:40}, for any point $q\in \widehat{ef}$,
there is a shortest path in  $\calF(\c(\calP'))$ from $s$ to $q$
that contains $\overline{rf}$ and the portion of $\widehat{ef}$
between $f$ and $q$. Since $\overline{ef}$ is on $\overline{rp}$
contained in the cell $C_{\c}(r)$, by the properties of the shortest
path map $\SPM(\calM)$ \cite{ref:MitchellAn89,ref:MitchellL192},
$\widehat{ef}$ is also contained in the cell $C_{\c}(r)$. Thus,
$\widehat{ef}$ is also contained in $C(r)$. If we process each ear
whose interior intersects $\overline{rp}$ as above, we find a path
in $C(r)$ that connects $r$ and $p$; further, this path has the same
$L_1$ length as $\overline{rp}$. Hence, $C(r)$ is a connected
region.

Next, we claim that for any point $p\in C(r)$, there is a shortest
path in $\calF(\calP')$ from $s$ to $p$ that contains $r$. Indeed,
since $p\in C_{core}(r)$, there is a shortest path in
$\calF(\c(\calP'))$ from $s$ to $p$ that contains $\overline{rp}$;
let $\pi_{core}(s,r)$ be the portion of this path between $s$ and
$r$. On one hand, we have shown above that there is a path from $r$
to $p$ in $C(r)$ with the same $L_1$ length as $\overline{rp}$. On
the other hand, by Lemma \ref{lem:40}, there exists a path in
$\calF(\calP')$ from $s$ to $r$ with the same length as
$\pi_{core}(s,r)$. Hence, a concatenation of these two paths results
in a shortest path from $s$ to $p$ in $\calF(\calP')$ that contains
$r$. Our claim thus follows.

The above claim and its proof also imply that decomposing $C(r)$
into a set of SPM regions is equivalent to computing an SPM in
$C(r)$ with the vertex $r$ as the source point, which we denote by
$\SPM(C(r))$. Since $C(r)$ is a connected region and $C_{core}(r)$
is a simple polygon, we claim that $C(r)$ is a (possibly degenerate)
simple polygon.  This is because for any ear $E$ that intersects
$C_{core}(r)$, the portion $E\cap C_{core}(r)$ lies on the boundary
of the simple polygon $C_{core}(r)$; thus, removing $E$ except its
obstacle path from $C_{core}(r)$ (to form $C(r)$) changes only the
boundary shape of $C_{core}(r)$ but does not change the nature of a
simple polygonal region (from $C_{core}(r)$ to $C(r)$). Based on the
fact that $C(r)$ is a (possibly degenerate) simple polygon,
$\SPM(C(r))$ can be easily computed in linear time in terms of the
number of edges of $C(r)$. For example, since the Euclidean shortest
path between any two points in a simple polygon is also an $L_1$
shortest path between the two points \cite{ref:HershbergerCo94}, an
SPM in a simple polygon with respect to the Euclidean distance is
also one with respect to the $L_1$ distance. Therefore, we can use a
corresponding shortest path algorithm for the Euclidean case (e.g.,
\cite{ref:GuibasLi87}) to compute each $\SPM(C(r))$ in our problem.

Note that our discussion above also implies that given
$\SPM(\c(\calP'),s)$, for each cell $C_{\c}(r)$ with a root $r$, we
can compute the corresponding $\SPM(C(r))$ separately. Clearly, the
$\SPM(C(r))$'s corresponding to all cells in $\SPM(\c(\calP'),s)$
constitute a shortest path map $\SPM(\calM)$ for $\calP'$.

Due to the planarity of the cell regions involved, the total number
of edges of all $C(r)$'s is $O(n)$.  Given a triangulation
$\Tri(\calP')$, all regions $C(r)$ can be obtained in totally $O(n)$
time. Computing all $\SPM(C(r))$'s also takes totally $O(n)$ time.
Thus, $\SPM(\calM)$ can be constructed in $O(n)$ time. The lemma
thus follows.
\end{proof}

Theorem \ref{theo:20} and Lemma \ref{lem:50} together lead to the
following result.

\lemmaspace
\begin{theorem}\label{theo:30}
Given a set of $h$ pairwise disjoint convex polygonal obstacles of
totally $n$ vertices in the plane, in $O(n+h\log h)$ time and $O(n)$ space, we can
construct a shortest path map of size $O(n)$ with respect to a
source point $s$, such that given any query point $t$ in the free
space, the length of an $L_1$ shortest $s$-$t$ path can be reported
in $O(\log h)$ time and an actual path can be found in $O(\log n+k)$
time where $k$ is the number of turns of the path.
\end{theorem}

\section{Shortest Paths among General Polygonal Obstacles}
\label{sec:general}

In this section, we consider the general case, i.e., the obstacles in
$\calP$ are not necessarily convex. In the following, in Section
\ref{subsec:pre}, we review the corridor structure
\cite{ref:KapoorAn97}, and introduce the {\em
ocean} $\calM$. In Section
\ref{subsec:single}, we present the algorithm for computing a single
shortest path and the similar idea also computes an SPM for $\calM$,
i.e., $\SPM(\calM)$. In Section \ref{sec:spm}, we outline our
algorithm for computing an SPM in the entire free space $\calF$.

\sectionspace
\subsection{Preliminaries}
\label{subsec:pre}

For simplicity of discussion, we assume that all obstacles are
contained in a large rectangle $\calR$ (see
Fig.~\ref{fig:triangulation}). Let $\calF$ be the free space
inside $\calR$. Let $t$ be an arbitrary point in $\calF$.

\begin{figure}[t]
\begin{minipage}[t]{0.53\linewidth}
\begin{center}
\includegraphics[totalheight=1.4in]{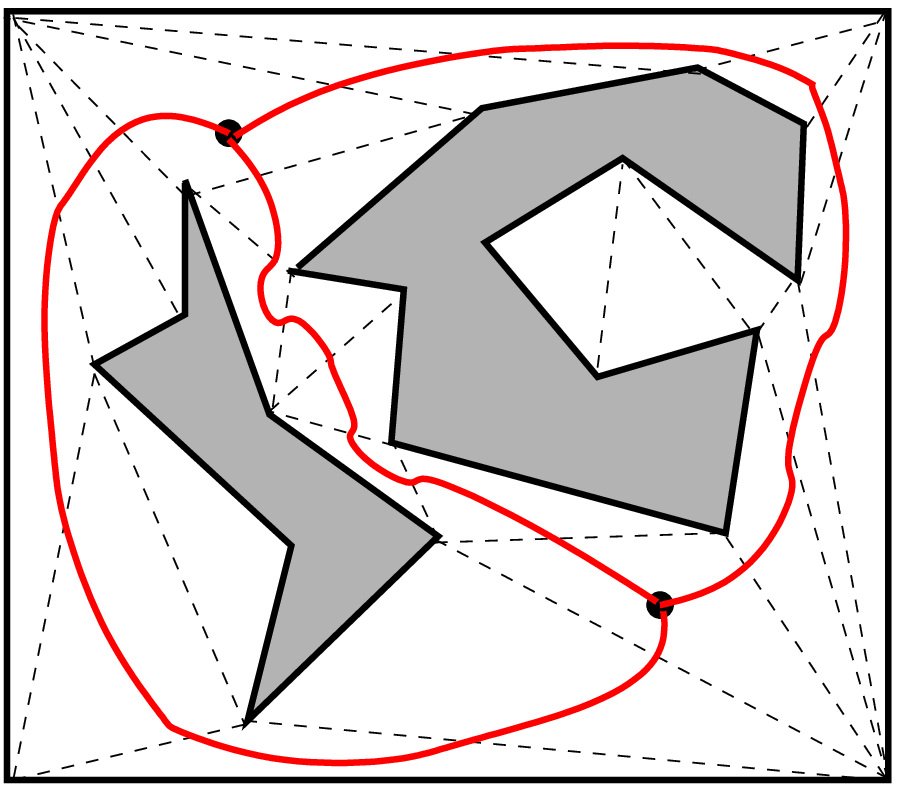}
\caption{\footnotesize Illustrating a triangulation of the free
space among two obstacles and the corridors (with red solid curves).
There are two junction triangles indicated by the large dots inside
them, connected by three solid (red) curves. Removing the two
junction triangles results in three corridors.}
\label{fig:triangulation}
\end{center}
\end{minipage}
\hspace*{0.02in}
\begin{minipage}[t]{0.45\linewidth}
\begin{center}
\includegraphics[totalheight=1.4in]{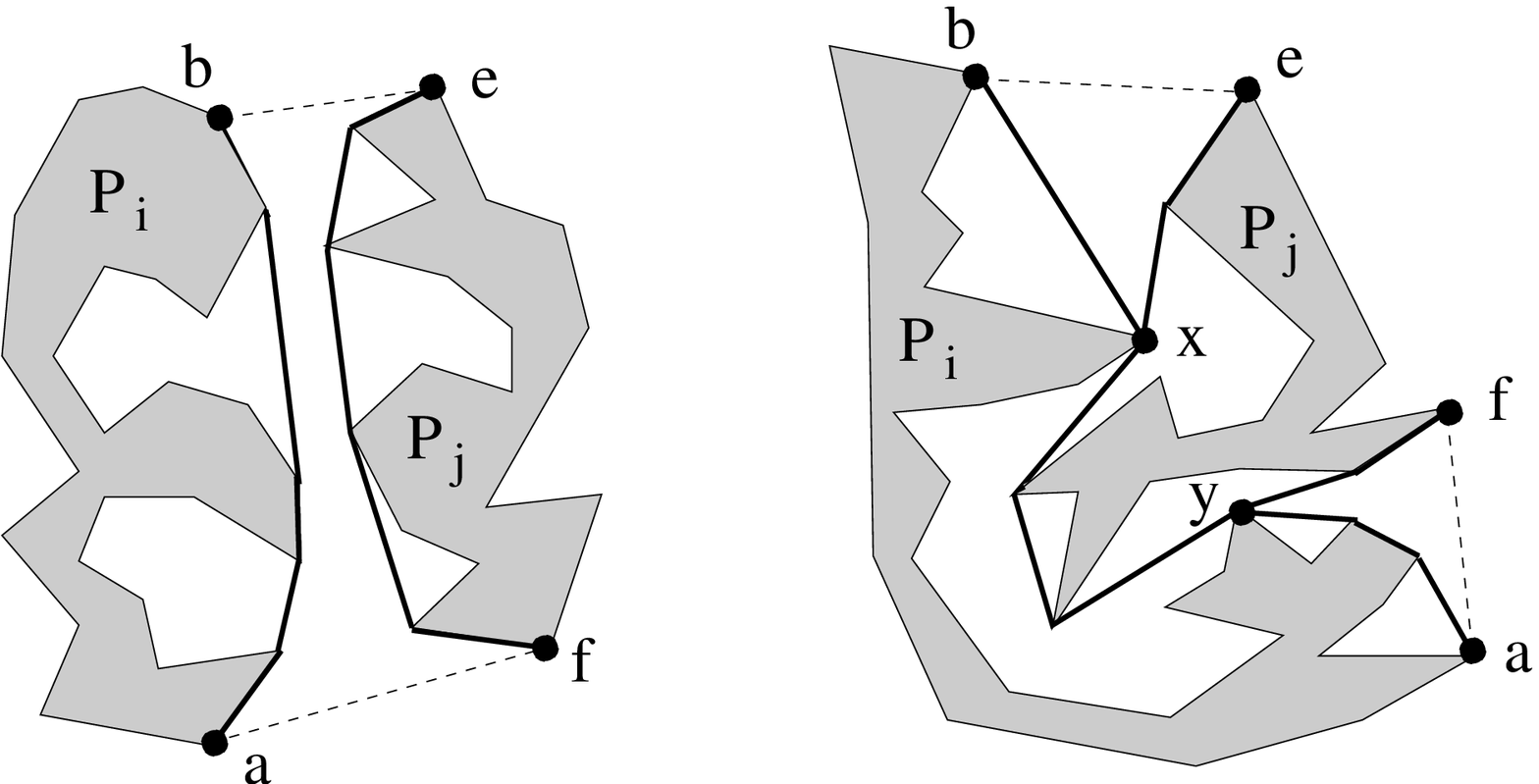}
\caption{\footnotesize Illustrating an open hourglass (left) and a
closed hourglass (right) with a corridor path linking the apices $x$
and $y$ of the two funnels. The dashed segments are diagonals. The
paths $\pi(a,b)$ and $\pi(e,f)$ are shown with thick solid curves.}
\label{fig:corridor}
\end{center}
\end{minipage}
\end{figure}

We first review the corridor structure \cite{ref:KapoorAn97}.
Denote by $\Tri(\calF)$ a triangulation of $\calF$.
Let $G(\calF)$ denote the (planar) dual graph of $\Tri(\calF)$,
i.e., each node of $G(\calF)$ corresponds to a triangle in
$\Tri(\calF)$ and each edge connects two nodes of $G(\calF)$
corresponding to two triangles sharing a diagonal of $\Tri(\calF)$.
The degree of each node in $G(\calF)$ is at most three.  As in
\cite{ref:KapoorAn97}, at least one node dual to a triangle incident
to each of $s$ and $t$ is of degree three. Based on $G(\calF)$, we
compute a planar 3-regular graph, denoted by $G^3$ (the degree of
each node in $G^3$ is three), possibly with loops and multi-edges,
as follows. First, we remove every degree-one node from $G(\calF)$
along with its incident edge; repeat this process until no
degree-one node exists. Second, remove every degree-two node from
$G(\calF)$ and replace its two incident edges by a single edge; repeat this process until no degree-two node exists. The
resulting graph is $G^3$ (e.g., see Fig.~\ref{fig:triangulation}).
The resulting graph $G^3$ has $O(h)$ faces,
$O(h)$ nodes, and $O(h)$ edges \cite{ref:KapoorAn97}. Each node of
$G^3$ corresponds to a triangle in $\Tri(\calF)$, which is called a
{\em junction triangle} (e.g., see Fig.~\ref{fig:triangulation}).
The removal of all junction triangles from $G^3$ results in $O(h)$
{\em corridors}, each of which corresponds to one edge of $G^3$.

The boundary of a corridor $C$ consists of four parts (see
Fig.~\ref{fig:corridor}): (1) A boundary portion of an obstacle
$P_i\in \calP$, from a point $a$ to a point $b$; (2) a diagonal of a
junction triangle from $b$ to a boundary point $e$ on an obstacle
$P_j\in \calP$ ($P_i=P_j$ is possible); (3) a boundary portion of
the obstacle $P_j$ from $e$ to a point $f$; (4) a diagonal of a
junction triangle from $f$ to $a$. The two diagonals $\overline{be}$
and $\overline{af}$ are called the {\em doors} of $C$. The corridor
$C$ is a simple polygon. Let
$\pi(a,b)$ (resp., $\pi(e,f)$) denote the shortest path from $a$ to $b$
(resp., $e$ to $f$) inside $C$. The region $H_C$ bounded by
$\pi(a,b), \pi(e,f)$, and the two diagonals $\overline{be}$ and
$\overline{fa}$ is called an {\em hourglass}, which is {\em open} if
$\pi(a,b)\cap \pi(e,f)=\emptyset$ and {\em closed} otherwise (see
Fig.~\ref{fig:corridor}). If $H_C$ is open, then both $\pi(a,b)$ and
$\pi(e,f)$ are convex chains and are called the {\em sides} of
$H_C$; otherwise, $H_C$ consists of two ``funnels" and a path
$\pi_C=\pi(a,b)\cap \pi(e,f)$ joining the two apices of the two
funnels, called the {\em corridor path} of $C$. The two funnel
apices connected by the corridor path are called the {\em corridor
path terminals}. Each funnel side is also convex. We compute the
hourglass for each corridor. After the triangulation, computing the
hourglasses for all corridors takes totally $O(n)$ time.

Let $Q$ be the union of all junction triangles and hourglasses. Then
$Q$ consists of $O(h)$ junction triangles, open hourglasses,
funnels, and corridor paths. As shown in \cite{ref:InkuluPl09},
there exists a shortest $s$-$t$ path $\pi(s,t)$ avoiding the
obstacles in $\calP$ which is contained in $Q$. Consider a corridor
$C$. If $\pi(s,t)$ contains an interior point of $C$, then the path
$\pi(s,t)$ must intersect both doors of $C$; further, if the
hourglass $H_C$ of $C$ is closed, then we claim that we can make the
corridor path of $C$ entirely contained in $\pi(s,t)$. Suppose
$\pi(s,t)$ intersects the two doors of $C$, say, at two points $p$
and $q$ respectively. Then since $C$ is a simple polygon, a
Euclidean shortest path between $p$ and $q$ inside $C$, denoted by
$\pi_E(p,q)$, is also an $L_1$ shortest path in $C$
\cite{ref:HershbergerCo94}. Note that $\pi_E(p,q)$ must contain the
corridor path of $C$. If we replace the portion of $\pi(s,t)$
between $p$ and $q$ by $\pi_E(p,q)$, then we obtain a new $L_1$
shortest $s$-$t$ path that contains the corridor path $\pi_C$. For
simplicity, we still use $\pi(s,t)$ to denote the new path. In other
words, $\pi(s,t)$ has the property that if $\pi(s,t)$ intersects
both doors of $C$ and the hourglass $H_C$ is closed, then the
corridor path of $C$ is contained in $\pi(s,t)$.

Let $\calM$ be $Q$ minus the corridor paths.  We call $\calM$ the {\em ocean}.
Clearly, $\calM\subseteq\calF$. The boundary of
$\calM$ consists of $O(h)$ reflex vertices and $O(h)$ convex chains,
implying that the complementary region $\calR\setminus \calM$ consists
of a set of polygons of totally $O(h)$ reflex vertices and $O(h)$
convex chains.  As shown in \cite{ref:KapoorAn97}, the region
$\calR\setminus \calM$ can be partitioned into a set $\calP'$ of $O(h)$
convex polygons of totally $O(n)$ vertices (e.g., by extending an
angle-bisecting segment inward from each reflex vertex). The ocean
$\calM$ is exactly the free space with respect to the convex polygons
in $\calP'$. In
addition, for each corridor path, no portion of it lies in $\calM$.
Further, the shortest path
$\pi(s,t)$ is a shortest $s$-$t$ path avoiding all convex polygons
in $\calP'$ and possibly utilizing some corridor paths. The set
$\calP'$ can be easily obtained in $O(n+h\log h)$ time. Therefore, as
in \cite{ref:KapoorAn97},
other than the corridor paths, we reduce our original \spp\ problem
to the convex case.

\sectionspace
\subsection{Finding a Single Shortest Path and Computing an SPM for 
$\calM$}
\label{subsec:single}

With the convex polygon set $\calP'$, to find a shortest $s$-$t$
path in $\calF$, if there is no corridor path, then we can simply apply our
algorithm for the convex case
in Section \ref{sec:convexcase}. Otherwise, the situation
is more complicated because the corridor paths can give possible
``shortcuts" for the sought $s$-$t$ path, and we must take these
possible ``shortcuts" into consideration while running the
continuous Dijkstra paradigm \cite{ref:MitchellAn89,ref:MitchellL192}.
The details are given below.

First, we compute the core set $\c(\calP')$ of $\calP'$. However,
the way we construct $\c(\calP')$ here is slightly different from
Section \ref{sec:convexcase}. For each convex polygon
$A'\in\calP'$, in addition to its leftmost, topmost, rightmost, and
bottommost vertices, if a vertex $v$ of $A'$ is a corridor path
terminal, then $v$ is also kept as a vertex of the core $\c(A')$. In
other words, $\c(A')$ is a simple (convex) polygon whose vertex set
consists of the leftmost, topmost, rightmost, and bottommost
vertices of $A'$ and all corridor path terminals on $A'$.
Since there are $O(h)$ terminal vertices, the cores in $\c(\calP')$
still have totally $O(h)$ vertices and edges. Further, the core set
thus defined still has the properties discussed in Section
\ref{sec:convexcase} for computing shortest $L_1$ paths, e.g.,
Observation \ref{obser:10} and Lemmas \ref{lem:10}, \ref{lem:20},
and \ref{lem:40}. Hence, by using our scheme in Section
\ref{sec:convexcase}, we can first find a shortest $s$-$t$ path
avoiding the cores in $\c(\calP')$ in $O(h\log h)$ time by applying
Mitchell's algorithm \cite{ref:MitchellAn89,ref:MitchellL192}, and
then obtain a shortest $s$-$t$ path avoiding $\calP'$ in $O(n)$ time
by Lemma \ref{lem:40}. But, the path thus computed may not be a true
shortest path in $\calF$ since the corridor paths are not utilized.
To find a true shortest path in $\calF$, we need to modify the
continuous Dijkstra paradigm when applying it to
$\c(\calP')$, as follows.


In Mitchell's algorithm \cite{ref:MitchellAn89,ref:MitchellL192},
when an obstacle vertex $v$ is hit by the wavefront for the first
time, it will be ``permanently labeled" with a value $d(v)$, which
is the length of a shortest path from $s$ to $v$ in the free space.
The wavefront consists of many ``wavelets" (each wavelet is a line
segment of slope $1$ or $-1$). The algorithm maintains a priority
queue (called ``event queue"), and each element in the queue is a
wavelet associated with an ``event point" and an ``event distance",
which means that the wavelet will hit the event point at the event
distance. The algorithm repeatedly takes (and removes) an element
from the event queue with the smallest event distance, and processes
the event. After an event is processed, some new events may be added
to the event queue. The algorithm stops when the point $t$ is hit by
the wavefront for the first time.

To handle the corridor paths in our problem, consider a corridor
path $\pi_C$ with $x$ and $y$ as its terminals and let $l$ be the
length of $\pi_C$. Recall that $x$ and $y$ are vertices of a core in
$\c(\calP')$. Consider the moment when the vertex $x$ is permanently
labeled with the distance $d(x)$. Suppose the wavefront that first
hits $x$ is from the funnel whose apex is $x$. Then according to our
discussions above, the only way that the wavelet of the wavefront at $x$ can affect
a shortest $s$-$t$ path is through the corridor path $\pi_C$. If $y$
is not yet permanently labeled, then $y$ has not been hit by the
wavefront. We initiate a ``pseudo-wavelet" that originates from $x$
with the event point $y$ and event distance $d(x)+l$, meaning that
$y$ will be hit by this pseudo-wavelet at the distance $d(x)+l$. We
add the pseudo-wavelet to the event queue.
If $y$ has been permanently labeled, then the wavefront has already
hit $y$ and is currently moving along the corridor path $\pi_C$ from
$y$ to $x$. Thus, the wavelet through $x$ will meet the wavelet
through $y$ somewhere on the path $\pi_C$, and these two wavelets
will ``die" there and never affect the free space outside the
corridor. Thus, if $y$ has been permanently labeled, then we do not
need to do anything on $y$. In addition, at the moment when the
vertex $x$ is permanently labeled, if the wavefront that first hits $x$ is
from the corridor path $\pi_C$ (i.e., through $y$), then the
wavelet at $x$ will keep going to the funnel of $x$ through $x$;
therefore, we process this event on $x$ as usual (i.e., as in
\cite{ref:MitchellAn89,ref:MitchellL192}), by initiating new wavelets
that originate from $x$.

For a corridor path $\pi_C$ with two terminals $x$ and $y$,
when $x$ is permanently labeled, if the wavefront that first hits $x$ is
not from the corridor path $\pi_C$, then we call $x$ a {\em wavefront incoming}
terminal; otherwise, $x$ is a {\em wavefront outgoing} terminal.
According to our discussion above, at least one of $x$ and $y$ must be a wavefront
incoming terminal. In fact, both $x$ and $y$ can be wavefront
incoming terminals, in which case the wavefronts passing through
$x$ and $y$ ``die" inside the corridor.

Intuitively, the above
treatment of corridor path terminals makes corridor paths act as
possible ``shortcuts" when we propagate the wavefront. The rest of
the algorithm proceeds in the same way as in
\cite{ref:MitchellAn89,ref:MitchellL192} (e.g., processing the
segment dragging queries). The algorithm stops when the wavefront
first hits the point $t$, at which moment a shortest $s$-$t$ path in
$\calF$ has been found.

Since there are $O(h)$ corridor paths, with the above modifications
to Mitchell's algorithm as applied to $\c(\calP')$, its running time
is still $O(h\log h)$. Indeed, comparing with the original
continuous Dijkstra scheme \cite{ref:MitchellAn89,ref:MitchellL192}
(as applied to $\c(\calP')$),
there are $O(h)$ additional events on the corridor path terminals,
i.e., events corresponding to those pseudo-wavelets.
To handle these additional events, we may, for example, as
preprocessing, for each corridor path, associate with each its
corridor path terminal $x$ the other terminal $y$ as well as the
corridor path length $l$. Thus, during the algorithm, when we process
the event point at $x$, we can find $y$ and $l$ immediately. In this
way, each additional event is handled in $O(1)$ time in addition to
adding a new event for it to the event queue. Hence, processing all
events still takes $O(h\log h)$ time. Note that the shortest $s$-$t$
path thus computed may penetrate some ears of $\calP'$. As in Lemma
\ref{lem:40}, we can obtain a shortest $s$-$t$ path in the free
space $\calF$ in additional $O(n)$ time. Since applying
Mitchell's algorithm on $core(\calP')$ takes $O(h)$ space, the space
used in our entire algorithm is $O(n)$.

In summary, we have the following result.

\lemmaspace
\begin{theorem}\label{theo:50}
Given a set of $h$ pairwise disjoint polygonal obstacles of totally
$n$ vertices in the plane, we can find an $L_1$ shortest path
between two points in the free space in $O(n+h\log^{1+\epsilon} h)$
time (or $O(n+h\log h)$ time if a triangulation of the free space is
given) and $O(n)$ space.
\end{theorem}
\lemmaspace

As Mitchell's algorithm \cite{ref:MitchellAn89,ref:MitchellL192}, the
above algorithm also computes a shortest path map on the free space of
the convex polygons in $\calP'$, i.e., $\SPM(\calM)$.
We should point out that
because of the $O(h)$ corridor paths, $\SPM(\calM)$ is
different from a ``normal" SPM in the following aspect. Consider a
corridor path $\pi_C$ with two terminals $x$ and $y$. Suppose $x$ is a
wavefront incoming terminal and $y$ is a wavefront outgoing
terminal. Then this means that the algorithm determines a shortest path from $s$ to
$y$ which goes through $x$. Corresponding to the corridor path $\pi_C$, we may
put a ``pseudo-cell" in
$\SPM(\calM)$ with $x$ as the root such that $y$ is the only point
in this ``pseudo-cell", and we also associate with the pseudo-cell the
corridor path $\pi_C$, which indicates that there is a
shortest $s$-$y$ path that consists of a shortest $s$-$x$ path and the
corridor path $\pi_C$. If $x$ and $y$ are both wavefront incoming
terminals, then we need not do anything for this corridor path.
Clearly, since there are $O(h)$ corridor paths, the above procedure of
building pseudo-cells affects neither the space bound nor the
time bound for constructing $\SPM(\calM)$. Therefore, the
$\SPM(\calM)$ of size $O(n)$ can be computed in $O(T)$ time and $O(n)$
space, where $T$ is the time for triangulating $\calF$.
Based on $\SPM(\calM)$, in Section \ref{sec:spm},
we will compute an SPM on the entire free space $\calF$ in additional
$O(n)$ time.

\sectionspace
\subsection{Computing a Shortest Path Map}
\label{sec:spm}

Based on $\SPM(\calM)$, in Section \ref{sec:spm}, together with Sections
\ref{sec:bay} and \ref{sec:canal}, we will compute in additional
$O(n)$ time an SPM on the entire free space $\calF$ with respect to
the source point $s$, denoted by $\SPM(\calF)$, which has the
standard query performances, i.e.,
for any query point $t$, it reports the length of a shortest $s$-$t$ path in
$O(\log n)$ time and the actual path in additional
time proportional to the number of turns of the path.

As discussed in \cite{ref:MitchellAn89,ref:MitchellL192}, $\SPM(\calF)$ may not be
unique. We show that an $\SPM(\calF)$ of size $O(n)$ can be
computed in $O(n+h\log^{1+\epsilon} h)$ time (or $O(n+h\log h)$ time
if a triangulation of the free space is given). Our techniques for
constructing $\SPM(\calF)$ are
independent of those in the earlier sections of this paper, and are
also different from those in the previous work (e.g.,
\cite{ref:MitchellAn89,ref:MitchellL192}).

This section introduces the new concepts, {\em bays} and {\em canals},
and outlines the algorithm, while the details are given
in Sections \ref{sec:bay} and \ref{sec:canal}. One key subproblem we
need to solve efficiently is the special weighted $L_1$ geodesic
Voronoi diagram problem, i.e., the challenging subproblem
illustrated in Fig.~\ref{fig:geoVoi}. A linear time
algorithm is given in Section \ref{sec:bay} for it. Section
\ref{sec:canal} deals with another subproblem, where
the algorithm in Section \ref{sec:bay} is used as a procedure.

\sectionspace
\subsubsection{Bays and Canals}

Recall that $\calM\subseteq\calF$. To compute $\SPM(\calF)$, since
we already have $\SPM(\calM)$, we only need to compute the portion of
$\SPM(\calF)$ in the space $\calF\setminus\calM$. We first examine
the space $\calF\setminus\calM$, which we partition into two type of
regions, {\em bays} and {\em canals}, defined as follows.

\begin{figure}[t]
\begin{minipage}[t]{\linewidth}
\begin{center}
\includegraphics[totalheight=1.4in]{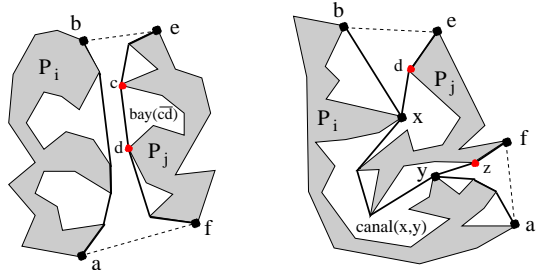}
\caption{\footnotesize Illustrating a bay $\bay$ in an open hourglass (left) and a
canal $\canal$ in a closed hourglass (right) with a corridor path linking
the apices $x$ and $y$ of its two funnels.
}
\label{fig:baycanal}
\end{center}
\end{minipage}
\vspace*{-0.15in}
\end{figure}

Consider an hourglass $H_C$ of a corridor $C$. We first discuss
the case when $H_C$ is open (see Fig.~\ref{fig:baycanal}). $H_C$
has two sides. Let $S_1(H_C)$ be an arbitrary side of $H_C$.
The obstacle vertices on $S_1(H_C)$ all lie on the same
obstacle, say $P\in\calP$. Let $c$ and $d$ be any two adjacent
vertices on $S_1(H_C)$ such that the line segment $\overline{cd}$ is
not an edge of $P$ (see the left figure in
Fig.~\ref{fig:baycanal}, with $P=P_j$). The region enclosed
by $\overline{cd}$ and a boundary portion of $P$ between $c$ and
$d$ is called the {\em bay} of $P$ and $\overline{cd}$, denoted by
$bay(\overline{cd})$, which is a simple polygon. We call
$\overline{cd}$ the {\em bay gate}.

If the hourglass $H_C$ is closed, then let $x$ and $y$ be the two apices of
its two funnels. Consider two adjacent vertices $c$ and $d$
on a side of a funnel such that the line segment $\overline{cd}$ is
not an obstacle edge. If neither $c$ nor $d$ is a funnel apex, then
$c$ and $d$ must both lie on the same obstacle and the
segment $\overline{cd}$ also defines a bay with that obstacle as above. However, if either
$c$ or $d$ is a funnel apex, say, $x=c$,
then $x$ and $d$ may lie on different obstacles.
If they both lie on the same obstacle, then they also define a
bay; otherwise, we call $\overline{xd}$ the {\em canal gate} at $x$
(see Fig.~\ref{fig:baycanal}).
Similarly, there is also a canal gate at the funnel apex $y$,
say $\overline{yz}$. Let $P_i$ and $P_j$ be
the two obstacles defining the hourglass $H_C$. The region enclosed by
$P_i$, $P_j$, and the two canal gates $\overline{xd}$ and
$\overline{yz}$ that contains the
corridor path of $H_C$ is called the {\em canal} of $H_C$, denoted by
$canal(x,y)$, which is a simple polygon.

It is easy to see that $\calF\setminus \calM$ consists of all
bays and canals thus defined.

To build $\SPM(\calF)$, we need to compute the portion of $\SPM(\calF)$ in all
bays and canals since we already have $\SPM(\calM)$. As all bays and
canals are connected with $\calM$ through their gates, we need to
``expand" $\SPM(\calM)$ to all bays/canals through their gates. Henceforth,
when saying ``compute an SPM for a bay/canal," we mean ``expand
$\SPM(\calM)$ into that bay/canal", and vice versa.
Computing an SPM for a bay is a key (i.e., the challenging
subproblem). Computing an SPM
for a canal uses the algorithm for a bay as a main procedure.

\subsubsection{Expanding $\SPM(\calM)$ into Bays and Canals}

We discuss the bays first. Consider a bay $\bay$.
If its gate $\overline{cd}$ is in a single cell $C(r)$ of $\SPM(\calM)$
with $r$ as the root, then each point in $bay(\overline{cd})$
has a shortest path to $s$ via $r$. Thus, to construct an SPM for $\bay$,
it suffices to compute an SPM on $bay(\overline{cd})$ with respect
to the single point $r$. This can be easily done
in linear time (in terms of the number of vertices of
$bay(\overline{cd})$) since $bay(\overline{cd})$ is a simple
polygon\footnote{For example, since the Euclidean shortest
path between any two points in a simple polygon is also an $L_1$
shortest path \cite{ref:HershbergerCo94}, a Euclidean SPM in a
simple polygon is also an $L_1$ one. Thus, we can use a
corresponding shortest path algorithm for the Euclidean case (e.g.,
\cite{ref:GuibasLi87}) to compute an $L_1$ SPM in $\bay$ with
respect to $r$ in linear time.}. Note that although $r$ may not be a vertex of
$bay(\overline{cd})$, we can, for example, connect $r$ to both $c$ and
$d$ with two line segments (both $\overline{rc}$ and $\overline{rd}$
are in $C(r)$) to obtain a new simple polygon that contains
$bay(\overline{cd})$.

If the gate $\overline{cd}$ is not contained in a single cell of
$\SPM(\calM)$, then the situation is more complicated. In this case,
multiple vertices of $\SPM(\calM)$ may lie in the interior of $\overline{cd}$
(i.e., the intersections of the boundaries of the cells of $\SPM(\calM)$ with
$\overline{cd}$). This is actually the challenging subproblem
illustrated by Fig.~\ref{fig:geoVoi}.
We refer to the vertices of $\SPM(\calM)$ on $\overline{cd}$ (including
its endpoints $c$ and $d$) as the {\em $\SPM(\calM)$ vertices} and let $m'$
be their total number. Let $n'$ be the number of vertices of
$bay(\overline{cd})$. A straightforward approach for computing an
SPM for $bay(\overline{cd})$ is to use the continuous Dijkstra
paradigm \cite{ref:MitchellAn89,ref:MitchellL192} to let the
wavefront continue to move into $bay(\overline{cd})$. But, this
approach may take $O((n'+m')\log (m'+n'))$ time. Later in Section
\ref{sec:bay}, we derive an $O(n'+m')$ time algorithm, as stated below.

\lemmaspace
\begin{theorem}\label{theo:baytime}
For a bay of $n'$ vertices with $m'$ $\SPM(\calM)$ vertices
on its gate, a shortest path map of size $O(n'+m')$
for the bay can be computed in $O(n'+m')$ time.
\end{theorem}
\lemmaspace

Since a canal has two gates which are
also edges of $\calM$, multiple $\SPM(\calM)$ vertices may lie on both
its gates.
Later in Section \ref{sec:canal}, we show the following result.

\begin{theorem}\label{theo:canaltime}
For a canal of $n'$ vertices with totally $m'$ $\SPM(\calM)$ vertices
on its two gates, a shortest path map of size $O(n'+m')$
for the canal can be computed in $O(n'+m')$ time.
\end{theorem}

\subsubsection{Wrapping Things Up}

By Theorems \ref{theo:baytime} and $\ref{theo:canaltime}$, the
time bound for computing the shortest path maps for all bays and canals is
linear in terms of the total sum of the numbers of obstacle vertices of all
bays and canals, which is $O(n)$, and the total number of the
$\SPM(\calM)$ vertices on the gates of all bays and canals, which is
also $O(n)$ since the size of $\SPM(\calM)$ is $O(n)$.

We hence conclude that given $\SPM(\calM)$,
$\SPM(\calF)$ can be computed in additional $O(n)$ time. With a
linear size planar point location data structure
\cite{ref:EdelsbrunnerOp86,ref:KirkpatrickOp83}, we have the
following result.

\begin{theorem}\label{theo:60}
Given a set of $h$ pairwise disjoint polygonal obstacles of totally $n$
vertices and a source point $s$ in the plane, we can build a shortest path
map of size $O(n)$ with respect to $s$ in $O(n+h\log^{1+\epsilon} h)$
time (or $O(n+h\log h)$ time if a triangulation of the free space is
given) and $O(n)$ space,
such that for any query point $t$, the length of a shortest $s$-$t$ path
can be reported in $O(\log n)$ time and the actual
path can be found in additional $O(k)$ time, where $k$ is the number
of turns of the path.
\end{theorem}

\section{Computing a Shortest Path Map for a Bay}
\label{sec:bay}

Consider a bay $bay(\overline{cd})$ with the gate $\overline{cd}$ (see
Fig.~\ref{fig:baycanal}).
Let $\SPM(\bay)$ be the SPM for $\bay$ that we seek to compute.

For the case when the segment $\overline{cd}$ lies in a single cell $C(r)$ of
$\SPM(\calM)$ with the root $r$, we have already shown how to
construct $\SPM(\bay)$
in linear time (in terms of the number of vertices of $bay(\overline{cd})$).
If the gate $\overline{cd}$ is not contained in a single cell of
$\SPM(\calM)$, then let $m'$ be the number of $\SPM(\calM)$ vertices
on $\overline{cd}$, and $n'$ be the number of vertices of
$bay(\overline{cd})$. In this section, we give an
algorithm for computing $\SPM(\bay)$ in $O(n'+m')$ time.

Let $R$ be the set of roots of the cells of $\SPM(\calM)$ that
intersect with $\overline{cd}$. To obtain $\SPM(\bay)$, we can first
compute, for each $r\in R$, the {\em Voronoi region} $\VD(r)$ inside
$bay(\overline{cd})$ such that for any point $t\in \VD(r)$, there is
a shortest $s$-$t$ path via $r$; we then compute an SPM on $\VD(r)$
with respect to the single point $r$. Since every $\VD(r)$ is a simple
polygonal region in $bay(\overline{cd})$, the shortest path map
$\SPM(\VD(r),r)$ can be computed in linear time in terms of the
number of vertices of $\VD(r)$ (e.g.,
by using an algorithm in \cite{ref:GuibasLi87,ref:HershbergerCo94}).
Thus, the key is to
decompose $bay(\overline{cd})$ into Voronoi regions for the roots of
$R$, which is exactly the challenging subproblem illustrated by
Fig.~\ref{fig:geoVoi}. Denote by $\Vor(bay(\overline{cd}))$ this
Voronoi diagram decomposition of
$bay(\overline{cd})$. We aim to compute
$\Vor(bay(\overline{cd}))$ in $O(n'+m')$ time.

Without loss of generality (WLOG), assume that $\overline{cd}$ is positive-sloped,
$bay(\overline{cd})$ is on the right of $\overline{cd}$,
and the vertex $c$ is higher than $d$ (e.g., $\bay=B$ in
Fig.~\ref{fig:geoVoi}). Other cases can be handled similarly.
Let $R=\{r_1,r_2,\ldots,r_k\}$ be the set of
roots of the cells of $\SPM(\calM)$ that intersect with
$\overline{cd}$ in the order from $c$ to $d$ along $\overline{cd}$.
Note that $R$ may be a multi-set, i.e., two roots $r_i$ and $r_j$
with $i\neq j$ may refer to the same physical point; but this is not
important to our algorithm (e.g., we can view each $r_i$ as a
physical copy of the same root). Let $c=v_0, v_1, \ldots, v_k=d$ be
the $\SPM(\calM)$ vertices on $\overline{cd}$ ordered from $c$ to
$d$ (thus $m'=k+1$). Hence, for each $1\leq i\leq k$, the segment
$\overline{v_{i-1}v_{i}}$ is on the boundary of the cell $C(r_i)$ of
$\SPM(\calM)$. Note that each cell $C(r_i)$ is a star-shaped polygon, and for each $1\leq i\leq k-1$,
$v_i$ lies on the common boundary of $C(r_i)$ and $C(r_{i+1})$ (i.e., $v_i\in C(r_i)\cap C(r_{i+1})$).
To obtain $\Vor(bay(\overline{cd}))$, for each
$r_i\in R$, we need to compute the Voronoi region $\Vor(r_i)$.

Our algorithm can be viewed as an incremental one, i.e., it
considers the roots in $R$ one by one. It is commonly known that
incremental approaches can construct Voronoi diagrams in quadratic
time, or may give good randomized result. In contrast, our
algorithm is deterministic and takes only linear time. The success
of it hinges on that we can find an {\em order} of the roots
in $R$ such that by following this order to consider the roots in $R$
incrementally, we are able to compute $\Vor(bay(\overline{cd}))$ in
linear time. The order is nothing but that
of the indices of the roots in $R$ we have defined.
With this order, the algorithm is quite simple. 
However, it is quite challenging to argue its
correctness and achieve a linear time implementation. Our strategy
is to show that the algorithm implicitly maintains a number of {\em
invariants} that assure the correctness of the algorithm. For this
purpose, we give many observations (in Section
\ref{sec:observations}). Additionally, some interesting techniques
are also used to implement and simplify the algorithm. 

We first give an algorithm overview in Section \ref{sec:algoOverview}.

\subsectionspace
\subsection{Algorithm Sketch}
\label{sec:algoOverview}

To compute $\Vor(\bay)$, it turns out that we need to deal with the interactions
between some rays, each
of which belongs to the bisector of two roots in $R$. Every such ray is either
horizontal or vertical. Further, considering the roots in $R$ incrementally
is equivalent to considering the corresponding rays incrementally.
We process these
rays in a certain order (e.g., as to be proved, their origins somehow form
a staircase structure). For each ray considered,
if it is vertical, then it is easy (it eventually leads to a ray shooting
operation), and its processing does not introduce any new ray.
But, if it is horizontal,
then the situation is more complicated since its processing may
introduce many new horizontal rays and (at most) one vertical ray,
also in a certain order along a staircase structure
(in addition to causing a ray shooting operation).
A stack is used to store certain vertical rays that need to be further
processed.

The algorithm needs to perform ray shooting operations for some
vertical and horizontal rays. Although there are known data structures
for ray shooting queries
\cite{ref:ChazelleRa94,ref:ChazelleVi89,ref:GuibasLi87,ref:HershbergerA95},
they are not efficient enough for a linear time implementation of the
entire algorithm.
Based on observations, our approach
makes use of the horizontal visibility map and vertical visibility map
of $\bay$ \cite{ref:ChazelleTr91}. More
specifically, we prove that all vertical ray shootings are in a ``nice"
sorted order (called {\em
target-sorted}). With this property, all vertical ray shootings are
performed in totally linear time by using the vertical visibility
map of $\bay$. The horizontal visibility map is used to guide the overall
process of the algorithm. During the algorithm, we march into the bay and
the horizontal visibility map allows us to keep track of our current
position (i.e., in a trapezoid of the map that contains our
current position).
The horizontal visibility map also allows each horizontal ray shooting to be
done in $O(1)$ time.
In addition, in the preprocessing of the algorithm, we also need to perform
some other ray shootings (for rays of slope $-1$); our linear time solution
for this also hinges on the target-sorted property of such rays.

Our algorithm is conceptually simple. As mentioned above, the only
data structures we need are linked lists, a stack, and the horizontal and vertical
visibility maps.  Its correctness relies
on the fact that the algorithm implicitly maintains a set
of invariant properties in each
iteration. To prove the algorithmic correctness,
of course, we need to show that these invariant
properties hold iteratively. Specifically, in our discussion of the algorithm, after
each iteration we formally prove that the invariants are well
maintained. For this purpose, before presenting the algorithm in
Section \ref{sec:algorithm}, we first show a set of observations in
Section \ref{sec:observations}, which capture some essential
properties of this $L_1$ problem. These observations may be helpful
for solving other related problems as well.
However, the discussion of these
observations and the formal proofs that the invariant properties
are maintained by the algorithm somehow make the presentation
of this whole section lengthy, technically complicated, or even tedious,
for which we ask for the reader's patience.

\subsectionspace
\subsection{Observations}
\label{sec:observations}

In this subsection, we give a number of observations, most of which
help capture the behaviors of the bisectors for the
roots of $R$ in computing $\Vor(bay(\overline{cd}))$.
Although some of the observations individually might appear
simple, they are essential
and adding them up leads to an efficient algorithmic strategy
for computing $\Vor(\bay)$ (as presented in Section \ref{sec:algorithm}).  The
observations also allow our algorithm to perform some key operations (e.g.,
ray shootings) in a faster manner than using a standard approach
\cite{ref:ChazelleRa94,ref:ChazelleVi89,ref:GuibasLi87,ref:HershbergerA95}.

For a point $p$, denote by $x(p)$ its $x$-coordinate and by $y(p)$
its $y$-coordinate.
For two objects $O_1$ and $O_2$ in the plane, if $x(p_1)\leq
x(p_2)$ for any two points $p_1\in O_1$ and $p_2\in O_2$, then we say
$O_1$ is to the {\em left} or {\em west} of $O_2$,
or $O_2$ is to the {\em right} or {\em  east} of $O_1$; if $y(p_1)\leq
y(p_2)$ for any two points $p_1\in O_1$ and $p_2\in O_2$, then we say
$O_1$ is to the {\em south} of $O_2$ or $O_1$ is {\em below} $O_2$, or
$O_2$ is to the {\em north} of $O_1$ or $O_2$ is {\em above} $O_1$.
If $O_1$ is to the left of $O_2$ and is also below $O_2$, then we say
$O_1$ is to the {\em southwest} of $O_2$ or
$O_2$ is to the {\em northeast} of $O_1$. We define {\em southeast}
and {\em northwest} similarly.

\begin{figure}[t]
\begin{minipage}[t]{\linewidth}
\begin{center}
\includegraphics[totalheight=1.2in]{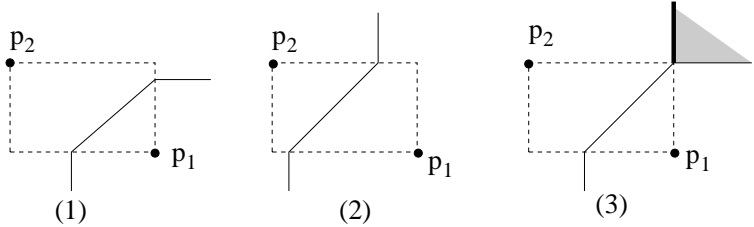}
\caption{\footnotesize Illustrating some cases of the bisector $B(p_1,p_2)$ of two
weighted points $p_1$ and $p_2$. In (3), an entire quadrant
(the shaded area) can be used as $B(p_1,p_2)$, but we choose $B(p_1,p_2)$
to be the vertical (solid thick) half-line.}
\label{fig:bisector}
\end{center}
\end{minipage}
\vspace*{-0.15in}
\end{figure}

In our problem, each root $r_i\in R$ can be viewed as an additively
weighted point whose weight is the $L_1$ length of a shortest
path from $s$ to $r_i$.  Thus, we need to consider the
possible shapes of the bisector of two weighted points.
For two weighted points $p_1$ and $p_2$ with weights $w_1$ and
$w_2$, respectively, their bisector $B(p_1,p_2)$ consists of all points
$q$ such that the $L_1$ length of the line segment $\overline{p_1q}$
plus $w_1$ is equal to the $L_1$ length of $\overline{p_2q}$ plus $w_2$.
Figure \ref{fig:bisector} shows some cases. Note
that the bisector can be an entire quadrant of the plane
(e.g., see Figure \ref{fig:bisector}(3)); in this case, as
in \cite{ref:MitchellAn89,ref:MitchellL192}, we choose a vertical
half-line as the bisector. For any pair of consecutive roots $r_{i-1}$ and
$r_i$ in $R$ for $2\leq i\leq k$, since the $\SPM(\calM)$ vertex
$v_{i-1}\in \overline{cd}$ is on the common boundary of $C(r_{i-1})$
and $C(r_i)$, $v_{i-1}$ lies on the bisector $B(r_{i-1},r_i)$ of $r_{i-1}$ and
$r_i$. For two points $p_1$ and $p_2$, denote by $Rec(p_1,p_2)$ the
rectangle with $p_1$ and $p_2$ as its two diagonal vertices.
The next observation is self-evident.

\lemmaspace
\begin{observation}\label{obser:20}
The bisector $B(p_1,p_2)$ consists of three portions: Two half-lines
and a line segment connecting them; the line segment has a slope $1$ or
$-1$ and is the intersection of $B(p_1,p_2)$ and
the rectangle $Rec(p_1,p_2)$, and each of the two
half-lines is perpendicular to an edge of $Rec(p_1,p_2)$ that
touches the half-line. Depending on the relative
positions and weights of $p_1$ and $p_2$, some portions of
$B(p_1,p_2)$ may degenerate and become
empty. $B(p_1,p_2)$ is monotone to both the
$x$- and $y$-axes. For any line $l$ containing a
portion of $B(p_1,p_2)$, $p_1$ and $p_2$ cannot lie strictly on the
same side of $l$.
\end{observation}
\lemmaspace

We call the {\em open} line segment of $B(p_1,p_2)$ strictly
inside $Rec(p_1,p_2)$ its
{\em middle segment}, denoted by $B_M(p_1,p_2)$, and
the two half-lines of $B(p_1,p_2)$ its two {\em rays},
each originating at a point on an edge of $Rec(p_1,p_2)$. Thus, the
origins of the two rays of $B(p_1,p_2)$ are the two endpoints of
$B_M(p_1,p_2)$.

Since each cell in an SPM is a star-shaped simple polygon, the
observation below is obvious.

\lemmaspace
\begin{observation}\label{obser:30}
Let $C(r)$ and $C(r')$ be two different cells in $\SPM(\calM)$ with roots
$r$ and $r'$. For any two points $p\in C(r)$ and $p'\in C(r')$, the
line segments $\overline{pr}$ and $\overline{p'r'}$ cannot cross
each other.
\end{observation}
\lemmaspace

The next lemma shows the possible relative positions of
two consecutive roots in $R$.

\lemmaspace
\begin{lemma} \label{lem:700}
For any two consecutive roots $r_{i-1}$ and $r_i$ in $R$ with $2\leq i\leq
k$, $r_i$ cannot be to the northeast of $r_{i-1}$, or equivalently,
$r_{i-1}$ cannot be to the southwest of $r_i$.
\end{lemma}
\lemmaspace
\begin{proof}
Since the $\SPM(\calM)$ vertex
$v_{i-1}\in \overline{cd}$
lies on the common boundary of the two
cells $C(r_{i-1})$ and $C(r_i)$, $v_{i-1}$ is on the bisector
$B(r_{i-1},r_i)$.

\begin{figure}[t]
\begin{minipage}[t]{\linewidth}
\begin{center}
\includegraphics[totalheight=1.0in]{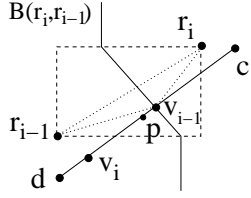}
\caption{\footnotesize An example of $r_i$ to the
northeast of $r_{i-1}$. The point $p\neq v_{i-1}$ is on $\overline{v_{i-1}v_i}$
and is infinitely close to $ v_{i-1}$. The line segment $\overline{r_ip}$ must
cross $\overline{r_{i-1}v_{i-1}}$.}
\label{fig:northeast}
\end{center}
\end{minipage}
\vspace*{-0.15in}
\end{figure}

Assume to the contrary that $r_i$ is to the northeast of $r_{i-1}$.
Note that $v_{i-1}$ may lie on either a half-line or the middle
segment of $B(r_{i-1},r_i)$. In either case, since $r_i$ is to the
northeast of $r_{i-1}$ and $\overline{cd}$ is positive-sloped,
according to Observation \ref{obser:20}, $v_{i-1}$ must be lower than $r_i$,
and $v_{i-1}$ must be to the right of $r_{i-1}$ (see
Fig.~\ref{fig:northeast}).

Since the segment $\overline{v_{i-1}v_i}$ is not a single point
and $v_i$ is to the left of $v_{i-1}$, we can find a point
$p\in\overline{v_{i-1}v_i}$ such that $p\neq v_{i-1}$ and $p$ is
infinitely close to $v_{i-1}$ (see Fig.~\ref{fig:northeast}). Since
$p\in\overline{v_{i-1}v_i}$ and $\overline{v_{i-1}v_i}\subseteq
C(r_i)$, we have $p\in C(r_i)$. Note that $v_{i-1}\in
C(r_{i-1})\cap C(r_i)$. Below we show that the two line segments
$\overline{r_ip}$ and $\overline{r_{i-1}v_{i-1}}$ must cross each
other, which contradicts with Observation \ref{obser:30}.

Since both $r_i$ and $r_{i-1}$ are obstacle vertices, by our
assumption, $r_i$ and $r_{i-1}$ do not lie on a horizontal
or vertical line. Hence $r_i$ is {\em strictly} to the northeast
of $r_{i-1}$. Note that no root in $R$ lies on $\overline{cd}$.
Since $v_{i-1}$ is lower than $r_i$ and is to the right
of $r_{i-1}$, the three points $v_{i-1}$, $r_i$, and $r_{i-1}$ do
not lie on the same line (see Fig.~\ref{fig:northeast}). In other words, the
triangle $\triangle r_iv_{i-1}r_{i-1}$ is a proper one. Further,
suppose $\rho(r_i,v_{i-1})$ (resp., $\rho(r_i,r_{i-1})$) is the ray
originating from $r_i$ and going through $v_{i-1}$ (resp., $r_{i-1}$);
then $\rho(r_i,r_{i-1})$ can be obtained by rotating
$\rho(r_i,v_{i-1})$ clockwise by an angle $\angle
v_{i-1}r_ir_{i-1}>0^{\circ}$. By the definition of the point $p$, during this
rotation, $p$ will be encountered by the rotating ray $\rho(r_i,v_{i-1})$ at an
angle $\angle v_{i-1}r_ip$ with $0^{\circ}<\angle v_{i-1}r_ip
<\angle v_{i-1}r_ir_{i-1}$, which
implies that $\overline{r_ip}$ crosses $\overline{r_{i-1}v_{i-1}}$.
The lemma thus follows.
\end{proof}

By Lemma \ref{lem:700},
there are three cases on the possible relative positions of $r_{i-1}$ with
respect to $r_i$, i.e., $r_{i-1}$ can be to the
southeast, northwest, or northeast of $r_{i}$.

\lemmaspace
\begin{lemma}\label{lem:relativepositions}
Consider any two consecutive roots $r_{i-1}$ and $r_i$ in $R$
with $2\leq i\leq k$.
\begin{enumerate}
\item
If $r_i$ is to the southeast of $r_{i-1}$, then $v_{i-1}$ is on a
ray of $B(r_{i-1},r_i)$ that is horizontally going east and $v_{i-1}$ is
to the right of $Rec(r_{i-1},r_i)$ (see Fig.~\ref{fig:relativepos}(1)).
\item
If $r_i$ is to the northwest of $r_{i-1}$, then $v_{i-1}$ is on a
ray of $B(r_{i-1},r_i)$ that is vertically going south and $v_{i-1}$ is
below $Rec(r_{i-1},r_i)$ (see Fig.~\ref{fig:relativepos}(2)).
\item
If $r_i$ is to the southwest of $r_{i-1}$, then $v_{i-1}$ is either on
the middle segment $B_M(r_{i-1},r_i)$, or on a ray of $B(r_{i-1},r_i)$
that is either horizontally going east or vertically going south (see
Fig.~\ref{fig:relativepos}(3)).
Further, if $v_{i-1}$ is on the ray horizontally going east, then
$v_{i-1}$ is to the right of $Rec(r_{i-1},r_i)$; if
$v_{i-1}$ is on the ray vertically going south, then
$v_{i-1}$ is below $Rec(r_{i-1},r_i)$.

\end{enumerate}
\end{lemma}
\lemmaspace
\begin{proof}
We first prove Part 1 of the lemma.
If $r_i$ is to the southeast of $r_{i-1}$ (see
Fig.~\ref{fig:relativepos}(1)), then the rectangle $Rec(r_{i-1},r_i)$
cannot intersect $\overline{cd}$. Thus, $v_{i-1}$ cannot be on
$B_M(r_{i-1},r_i)$, and $v_{i-1}$ must be on a ray of
$B(r_{i-1},r_i)$, denoted by $\rho$. By Observation \ref{obser:20},
the origin of $\rho$ is on an edge $\alpha$ of $Rec(r_{i-1},r_i)$ and is
perpendicular to the edge $\alpha$. Since $v_{i-1}\in \rho$ and $r_{i-1}$ is to
the northwest of $r_i$, $\alpha$ must be one of the two edges incident
to $r_i$, i.e., the bottom edge or the right edge of
$Rec(r_{i-1},r_i)$.
In addition, if $\alpha$ is the bottom edge of $Rec(r_{i-1},r_i)$, then $\rho$
must be vertically going south; further, since $r_i$ is to the
southeast of $r_{i-1}$, by a similar argument as that for the
proof of Lemma \ref{lem:700}, we can obtain a contradiction. Thus,
$\alpha$ is the right edge of $Rec(r_{i-1},r_i)$ and $\rho$
must be horizontally going east. In addition, it is easy to see that
$v_{i-1}$ must be to the right of $Rec(r_{i-1},r_i)$.
Part 1 of the lemma thus follows.

\begin{figure}[t]
\begin{minipage}[t]{\linewidth}
\begin{center}
\includegraphics[totalheight=1.4in]{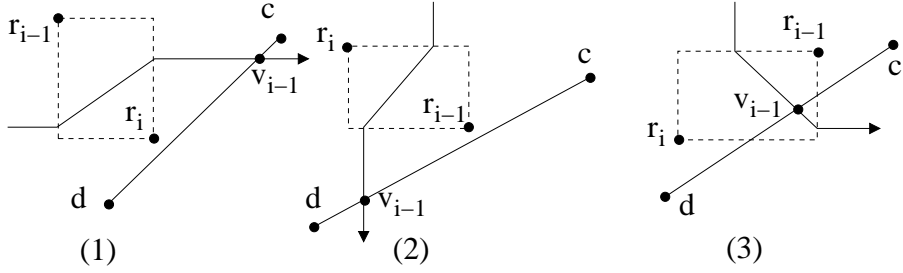}
\caption{\footnotesize Illustrating the three possible relative
positions of $r_{i-1}$ and $r_i$.}
\label{fig:relativepos}
\end{center}
\end{minipage}
\vspace*{-0.15in}
\end{figure}

Part 2 can be proved analogously as Part 1, and we omit it.

For Part 3, if $Rec(r_{i-1},r_i)$ intersects $\overline{cd}$, then it is
possible that $B_M(r_{i-1},r_i)$ intersects $\overline{cd}$ (at
$v_{i-1}$). If $B_M(r_{i-1},r_i)$ doest not intersect $\overline{cd}$,
then $v_{i-1}$ lies on a ray of $B(r_{i-1},r_i)$, denoted by $\rho$.
Again, the origin of $\rho$ is on either the right edge of $Rec(r_{i-1},r_i)$
or the bottom edge of $Rec(r_{i-1},r_i)$. In the former case,
$\rho$ is horizontally going east and $v_{i-1}$ is to the right of
$Rec(r_{i-1},r_i)$. In the latter case,
$\rho$ is vertically going south and $v_{i-1}$ is below $Rec(r_{i-1},r_i)$.
Part 3 thus follows.
\end{proof}

For any two consecutive roots $r_{i-1}$ and $r_i$ in $R$ with $2\leq i\leq
k$, if $v_{i-1}$ is on a ray $\rho$ of $B(r_{i-1},r_i)$, then we let $\rho_{i-1}$
be the ray originating at $v_{i-1}$ with the same direction as $\rho$.
If $v_{i-1}$ lies on the middle segment
of $B(r_{i-1},r_i)$, then by Lemma \ref{lem:relativepositions},
$r_{i-1}$ is to the northeast of $r_i$ and $\overline{cd}$ intersects
$Rec(r_{i-1},r_i)$; in this case, let $\rho_{i-1}$ be the ray of
$B(r_{i-1},r_i)$ that is below or to the right of $v_{i-1}$ and goes inside
$bay(\overline{cd})$. For a ray $\rho$, let $or(\rho)$ denote the
origin of $\rho$. Observation \ref{obser:40} below is obvious.

\lemmaspace
\begin{observation}\label{obser:40}
For any $2\leq i\leq k$, the ray $\rho_{i-1}$ is either horizontally
going east or vertically going south. If $v_{i-1}$ is on a ray of
$B(r_{i-1},r_i)$, then $or(\rho_{i-1})=v_{i-1}$; if $v_{i-1}$ is on
$B_M(r_{i-1},r_i)$, then $or(\rho_{i-1})$ is on either the right edge
or the bottom edge of $Rec(r_{i-1},r_i)$.
\end{observation}
\lemmaspace

\lemmaspace
\begin{lemma}\label{lem:1000}
Consider any two consecutive roots $r_{i-1}$ and $r_i$ in $R$ with $2\leq i\leq k$.
\begin{enumerate}
\item
If the ray $\rho_{i-1}$ is horizontal, then $r_{i-1}$ is above
$\rho_{i-1}$ and $r_i$ is below $\rho_{i-1}$.
\item
If $\rho_{i-1}$ is
vertical, then $r_{i-1}$ is to the right of
$\rho_{i-1}$ and $r_i$ is to the left of $\rho_{i-1}$.
\item The origin $or(\rho_{i-1})$ of $\rho_{i-1}$ is always below $r_{i-1}$
and to the right of $r_i$.
\end{enumerate}
\end{lemma}
\lemmaspace
\begin{proof}
There are three cases on the possible relative positions of $r_{i-1}$ and $r_i$.

\begin{itemize}
\item
If $r_{i-1}$ is to the northwest of $r_i$ (see
Fig.~\ref{fig:relativepos}(1)), then by the proof of
Lemma \ref{lem:relativepositions}, $\rho_{i-1}$ is horizontal and
is contained in the ray of $B(r_{i-1},r_i)$ whose origin is
on the right edge of $Rec(r_{i-1},r_i)$.  Since $r_{i-1}$ and $r_i$ are
two diagonal vertices of $Rec(r_{i-1},r_i)$, $\rho_{i-1}$ is above
$r_{i}$ and below $r_{i-1}$.

Further, the origin
$or(\rho_{i-1})$ is $v_{i-1}$, which is below $r_{i-1}$ and to the
right of $r_i$.
\item
If $r_{i-1}$ is to the southeast of $r_i$ (see
Fig.~\ref{fig:relativepos}(2)), then by the proof of
Lemma \ref{lem:relativepositions}, $\rho_{i-1}$ is vertical and
lies on the ray of $B(r_{i-1},r_i)$ whose origin is
on the bottom edge of $Rec(r_{i-1},r_i)$.
Since $r_{i-1}$ and $r_i$ are
two diagonal vertices of $Rec(r_{i-1},r_i)$, $\rho_{i-1}$ is to the
right of $r_i$ and to the left of $r_{i-1}$.
Further, the origin
$or(\rho_{i-1})$ is $v_{i-1}$, which is below $r_{i-1}$ and to the
right of $r_i$.

\item
If $r_{i-1}$ is to the northeast of $r_i$ (see
Fig.~\ref{fig:relativepos}(3)), then if $\rho_{i-1}$
is horizontal, then the proof is similar to the first case; otherwise, the
proof is similar to the second case.
\end{itemize}

The lemma thus follows.
\end{proof}

\lemmaspace
\begin{lemma}\label{lem:80}
For any $i$ with $3\leq i\leq k-1$, if $r_i$ is to the southwest of
$r_{i-1}$, then $v_{i-2}$ is to the right of the rectangle
$Rec(r_{i-1},r_i)$ and $v_{i}$ is below $Rec(r_{i-1},r_i)$.
\end{lemma}
\lemmaspace
\begin{proof}
Suppose $r_i$ is to the southwest of $r_{i-1}$. We only prove that
$v_{i-2}$ is to the right of the rectangle
$Rec(r_{i-1},r_i)$. The case that $v_{i}$ is below $Rec(r_{i-1},r_i)$
can be proved analogously.

Note that $v_{i-2}\in B(r_{i-2},r_{i-1})$.
We discuss the three possible relative positions of $r_{i-2}$ and
$r_{i-1}$. By Lemma \ref{lem:700}, $r_{i-2}$ may be to the
southeast, northwest, or northeast of $r_{i-1}$.
Since $r_i$ is to the southwest of $r_{i-1}$, to prove $v_{i-2}$ is to
the right of $Rec(r_{i-1},r_i)$, it suffices to show
that $v_{i-2}$ is to the right of $r_{i-1}$.

\begin{itemize}
\item
If $r_{i-2}$ is to the southeast of $r_{i-1}$, then by Lemma
\ref{lem:relativepositions}, $v_{i-2}$ is on the ray of
$B(r_{i-1},r_{i-2})$ vertically going south, i.e., $\rho_{i-2}$ is vertical.
By Lemma \ref{lem:1000}, $r_{i-1}$ is to the left of $\rho_{i-2}$.
Since $v_{i-2}\in \rho_{i-2}$, $v_{i-2}$ is to the right of $r_{i-1}$.

\item
If $r_{i-2}$ is to the northwest of $r_{i-1}$, then by Lemma
\ref{lem:relativepositions}, $v_{i-2}$ is to the right of
$Rec(r_{i-2},r_{i-1})$, and thus to the right of $Rec(r_{i-1},r_i)$.

\item
If $r_{i-2}$ is to the northeast of $r_{i-1}$, then the rectangle
$Rec(r_{i-2},r_{i-1})$ is to the northeast of $Rec(r_{i-1},r_i)$.
If $v_{i-2}$ is on $B_M(r_{i-2},r_{i-1})$, then since $v_{i-2}$ is inside
$Rec(r_{i-2},r_{i-1})$, $v_{i-2}$ is to the right of
$Rec(r_{i-1},r_i)$; otherwise, the proof is similar to the above two
cases.
\end{itemize}

The lemma thus follows.
\end{proof}

Recall that when sketching the algorithm in Section
\ref{sec:algoOverview}, we mentioned that the origins of the rays involved
somehow form a staircase structure. The next lemma states this important fact.

\lemmaspace
\begin{lemma}\label{lem:90}
For any $i$ with $2\leq i\leq k-1$, $or(\rho_{i-1})$ is to the northeast of
$or(\rho_i)$.
\end{lemma}
\lemmaspace
\begin{proof}
We first discuss a scenario that will be used later in this proof.
Consider any two consecutive roots $r_j$ and $r_{j+1}$ in $R$, $1\leq j\leq
k-1$, with $or(\rho_{j})\neq v_{j}$. Then based on our discussion above, it must be the case that
$r_{j+1}$ is to the southwest of $r_{j}$, $\overline{cd}$
intersects the rectangle $Rec(r_j,r_{j+1})$, and
$or(\rho_{j})$ is a point on an edge of $Rec(r_j,r_{j+1})$.
Let $z_{j}$ be the intersection
of $\overline{cd}$ and the right edge of $Rec(r_j,r_{j+1})$
(see Fig.~\ref{fig:originpos}).
The origin $or(\rho_{j})$ can be either
on the right edge or the bottom edge of $Rec(r_j,r_{j+1})$.
In either case, $or(\rho_{j})$ must be both below and to the left of
$z_{j}$, i.e., $z_{j}$ is to the northeast of
$or(\rho_{j})$.

\begin{figure}[t]
\begin{minipage}[t]{\linewidth}
\begin{center}
\includegraphics[totalheight=1.0in]{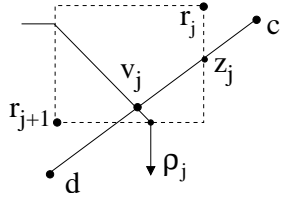}
\caption{\footnotesize Illustrating the case when $r_{j}$ is to the
northeast of $r_{j+1}$ and $or(\rho_{j})\neq v_{j}$.}
\label{fig:originpos}
\end{center}
\end{minipage}
\vspace*{-0.15in}
\end{figure}

Consider any $i$ with $2\leq i\leq k-1$.
To prove the lemma, depending on whether
$or(\rho_{i-1})=v_{i-1}$ and whether $or(\rho_{i})=v_i$,
there are four cases.

\begin{enumerate}
\item
If $or(\rho_{i-1})=v_{i-1}$ and $or(\rho_{i})=v_{i}$,
then since $v_{i-1}$ and $v_i$ are on $\overline{cd}$ in the order from
$c$ to $d$, $v_{i-1}$ is to the northeast of $v_i$, and thus
$or(\rho_{i-1})$ is to the northeast of $or(\rho_{i})$.

\item
If $or(\rho_{i-1})=v_{i-1}$ and $or(\rho_{i})\neq v_{i}$,
then by our discussion at the beginning of this proof,
$r_{i+1}$ is to the southwest of $r_{i}$, the rectangle
$Rec(r_i,r_{i+1})$ intersects $\overline{cd}$, and
the point $z_i$ is to the northeast of $or(\rho_{i})$. Further,
since $r_{i+1}$ is to the southwest of $r_{i}$,
by Lemma \ref{lem:80}, $v_{i-1}$ is to the right of
$Rec(r_i,r_{i+1})$ and thus to the right of $z_i$. Since $v_{i-1}$ is to the
right of $z_i$ and both
$v_{i-1}$ and $z_i$ are on $\overline{cd}$, $v_{i-1}$ is to the
northeast of $z_i$.  Therefore, $or(\rho_{i-1})$
($=v_{i-1}$) is to the northeast of $or(\rho_{i})$.

\item
If $or(\rho_{i-1})\neq v_{i-1}$ and $or(\rho_{i})= v_{i}$,
then the analysis is somewhat similar to the second case.

\item
If $or(\rho_{i-1})\neq v_{i-1}$ and $or(\rho_{i})\neq v_{i}$,
then $r_{i-1}$ is to the northeast of $r_i$ and $r_{i}$ is to the
northeast of $r_{i+1}$. Hence,
the rectangle $Rec(r_{i-1},r_i)$ is to the northeast
of $Rec(r_i,r_{i+1})$. Since $or(\rho_{i-1})$ is on
$Rec(r_{i-1},r_i)$ and $or(\rho_{i})$ is on $Rec(r_i,r_{i+1})$,
we also obtain that $or(\rho_{i-1})$ is to
the northeast of $or(\rho_{i})$.

\end{enumerate}

The lemma thus follows.
\end{proof}
\lemmaspace
\begin{lemma}\label{lem:100}
Consider any root $r_i\in R$ with $1\leq i\leq k$. For any ray
$\rho_{j}$, if $j\leq i-1$ and $\rho_{j}$ is vertical,
then $\rho_{j}$ is to the right of $r_i$; if $j\geq i$ and $\rho_{j}$ is
horizontal, then $\rho_{j}$ is below $r_i$.
\end{lemma}
\lemmaspace
\begin{proof}
WLOG, assume $i<k$.  Consider the ray $\rho_i$, which is on
$B(r_i,r_{i+1})$. By Lemma \ref{lem:1000}, the origin $or(\rho_i)$ is
below $r_i$. By Lemma \ref{lem:90}, for any ray $\rho_j$ with $j\geq
i$, $or(\rho_j)$ is below $or(\rho_i)$ and thus is below
$r_i$. Hence, if $\rho_j$ is horizontal, then
$\rho_j$ must be below $r_i$.

By an analogous analysis, we can show that if $j\leq i-1$ and $\rho_j$ is
vertical, then $\rho_j$ is to the right of $r_i$. We omit the details.
The lemma thus follows.
\end{proof}

Note that in any SPM, a common boundary of two adjacent cells $C(r)$ and
$C(r')$ is a subset of the bisector $B(r,r')$.

For any two consecutive roots $r_{i-1}$ and $r_i$ in $R$, $2\leq i\leq
k$, the vertex $v_{i-1}$ divides $B(r_{i-1},r_i)$ into two portions;
we denote by $B_{bay}(r_{i-1},r_i)$ the portion that goes inside $\bay$
following $v_{i-1}$.
A key to building $\Vor(bay(\overline{cd}))$ is to compute the
interactions among all $B_{bay}(r_{i-1},r_i)$'s, for $i=2,3,\ldots,k$, inside
$bay(\overline{cd})$. Note that if $v_{i-1}$ is on a ray of
$B(r_{i-1},r_i)$, then $B_{bay}(r_{i-1},r_i)$ is the ray $\rho_{i-1}$;
otherwise, $v_{i-1}$ is on $B_M(r_{i-1},r_i)$ (i.e., the middle segment
of $B(r_{i-1},r_i)$), and
$B_{bay}(r_{i-1},r_i)$ consists of a portion of $B_M(r_{i-1},r_i)$ in
$Rec(r_{i-1},r_i)$ (i.e., the line segment
$\overline{v_{i-1}or(\rho_{i-1})}$) and the ray $\rho_{i-1}$.
Lemma \ref{lem:800} below shows that the portion
of $B_M(r_{i-1},r_i)$ which is inside
$bay(\overline{cd})$ will appear in $\SPM(\calF)$ (and
thus in $\Vor(bay(\overline{cd}))$), implying
that we can simply keep it when computing
$\Vor(bay(\overline{cd}))$ and we only need to further deal with
the rays $\rho_{i}$ for $i=1,2,\ldots,k-1$.
Thus, dealing with the rays $\rho_{i}$ is the main issue
of our algorithm (as discussed in Section \ref{sec:algoOverview}).

\lemmaspace
\begin{lemma} \label{lem:800}
For any two consecutive roots $r_{i-1}$ and $r_i$ in $R$, $2\leq
i\leq k$, if $v_{i-1}$ lies on $B_M(r_{i-1},r_i)$, then the
portion of $B_M(r_{i-1},r_i)$ inside $bay(\overline{cd})$ will
appear in $\Vor(bay(\overline{cd}))$.
\end{lemma}
\lemmaspace
\begin{proof}
Consider two consecutive roots $r_{i-1}$ and $r_i$ in $R$, $2\leq
i\leq k$, with $v_{i-1}$ lying on $B_M(r_{i-1},r_i)$.

Denote by $B_M'$ the portion of $B_M(r_{i-1},r_i)$ inside
$bay(\overline{cd})$. Recall that $B_M(r_{i-1},r_i)$ is an open
segment that does not contain its endpoints and is strictly inside
$Rec(r_{i-1},r_i)$.  To prove the lemma, it suffices to show
that for any two roots $r_j$ and $r_h$ in $R$ with $\{r_j,r_h\}\neq
\{r_{i-1},r_i\}$, if a portion of $B(r_j,r_h)$ appears in
$\SPM(\calF)$, then that portion does not intersect $B_M'$.

By Lemma \ref{lem:700}, $r_i$ may be to the southeast, or northwest,
or southwest of $r_{i-1}$. Since $\overline{cd}$ is positive-sloped,
if $r_i$ is to the northwest or southeast of $r_{i-1}$, then
$\overline{cd}$ cannot intersect the rectangle $Rec(r_{i-1},r_i)$
and thus $v_{i-1}$ cannot lie on $B_M(r_{i-1},r_i)$.
Therefore, the only possible case is that $r_i$ is to the southwest of
$r_{i-1}$.

First, we assume $i-1\geq 2$ and consider the root $r_{i-2}$.
We discuss the possible relative positions of $r_{i-2}$ with respect
to $r_{i-1}$. Recall that the bisector portion
$B_{bay}(r_{i-2},r_{i-1})$ either is $\rho_{i-2}$ or consists of
$\overline{v_{i-2}or(\rho_{i-2})}$ and $\rho_{i-2}$. Note that in
either case, when moving along $B_{bay}(r_{i-2},r_{i-1})$ from
$v_{i-2}$, $B_{bay}(r_{i-2},r_{i-1})$ is monotonically increasing in the
$x$-coordinates. Hence, $v_{i-2}$ is a leftmost point of
$B_{bay}(r_{i-2},r_{i-1})$.
Since $r_i$ is to the southwest of
$r_{i-1}$, by Lemma \ref{lem:80}, $v_{i-2}$ is to the right of
$Rec(r_{i-1},r_i)$ and thus is strictly to the right of $B_M'$.
Hence, $B_{bay}(r_{i-2},r_{i-1})$ cannot intersect $B_M'$.


%
%
%

For any pair of consecutive roots $r_{j-1}$ and $r_j$ in $R$,
$2\leq j\leq i-2$, similarly, when moving from $v_{j-1}$ along
$B_{bay}(r_{j-1},r_j)$, $B_{bay}(r_{j-1},r_j)$ is monotonically increasing in the
$x$-coordinates. Since $v_{i-2}$ is strictly to the right of $B_M'$ and
$v_{j-1}$ is to the right of $v_{i-2}$, $B_{bay}(r_{j-1},r_j)$
cannot intersect $B_M'$.

Let $R_1=\{r_1,r_2,\ldots,r_{i-1}\}$ and
$R_2=\{r_{i},r_{i+1},\ldots,r_{k}\}$. (Note that since $R$ may be a
multi-set, $R_1$ and $R_2$ possibly contain the same physical
root, but this is not important to our analysis.)

For any two different pairs of consecutive roots $r_{j-1},r_j$ and
$r_{t-1},r_t$ with $2\leq j\leq i-1$ and $2\leq t\leq i-1$,
it is possible that $B_{bay}(r_{j-1},r_j)$ and $B_{bay}(r_{t-1},r_t)$
intersect in $\SPM(\calF)$; if that happens, then let $B'$ be the
resulting bisector.  It is not difficult to see that $B'$ must be going
in a direction between the original directions of $B_{bay}(r_{j-1},r_j)$ and
$B_{bay}(r_{t-1},r_t)$. Since neither $B_{bay}(r_{j-1},r_j)$ nor
$B_{bay}(r_{t-1},r_t)$ intersects $B_M'$, $B'$ cannot intersect $B_M'$.
We can further consider the possible intersection between $B'$ and
the bisector of another two roots in $R_1$ in the manner as above, and
show likewise that the new bisector thus resulted cannot intersect $B_M'$.

The above argument shows that for any two
roots $r_j$ and $r_t$ in $R_1$ such that a portion of $B(r_j,r_t)$ appears in
$\Vor(bay(\overline{cd}))$, that portion does not intersect
$B_M'$. By a similar argument, we can also show that for any two
roots $r_j$ and $r_t$ in $R_2$ such that a portion of $B(r_j,r_t)$ appears in
$\Vor(bay(\overline{cd}))$, that portion does not intersect $B_M'$.

It remains to show that for any two roots $r_j\in R_1$ and $r_t\in
R_2$ such that $\{r_j,r_t\}\neq \{r_{i-1},r_i\}$ and a portion of
$B(r_j,r_t)$ appears in $\Vor(bay(\overline{cd}))$,
that portion does not intersect $B_M'$. Note that the case of $B(r_j,r_t)$
(partially) appearing in $\Vor(bay(\overline{cd}))$ can occur only after
$B_{bay}(r_{i-1},r_i)$ is ``blocked" by an intersection between
$B_{bay}(r_{i-1},r_i)$ and the bisector of two roots in $R_1$ or two roots in $R_2$.
Since the bisector of any two roots in $R_1$ or
any two roots in $R_2$ cannot intersect $B_M'$,
the portion of $B(r_j,r_t)$ appearing in
$\Vor(bay(\overline{cd}))$ cannot intersect $B'_M$ either.

The lemma thus follows.
\end{proof}

The observations presented above help determine the behaviors of the bisectors for
the roots in $R$ (e.g., the properties of the rays
$\rho_1,\rho_2,\ldots,\rho_{k-1}$), which are crucial to constructing
$\Vor(\bay)$. They form a basis
for both showing the correctness and the efficiency of our algorithm
in Section \ref{sec:algorithm}. For example, Lemma \ref{lem:90} can
help conduct a set of ray shooting operations in linear time, and Lemma
\ref{lem:800} allows us to decompose the problem into certain
subproblems with good properties.

\subsection{The Algorithm for Computing $\Vor(\bay)$}
\label{sec:algorithm}

In this subsection, we present our algorithm for computing
$\Vor(bay(\overline{cd}))$, i.e., computing the Voronoi region
$\Vor(r)$ for each root $r\in R$.

As shown in \cite{ref:MitchellAn89,ref:MitchellL192}, a key property
of the problem in the $L_1$ metric is: There exists
an SPM such that each edge of the SPM is horizontal, or vertical, or
of a slope $1$ or $-1$.
As shown below, the curves
involved in specifying $\Vor(bay(\overline{cd}))$ consist of only
line segments of slopes $0$, $\infty$, and $-1$ (there is no $+1$,
which is due to the assumption that $\overline{cd}$ is
positive-sloped). A line (segment) is said to be {\em ($-1$)-sloped}
if its slope is $-1$. Our algorithm needs to perform some vertical,
horizontal, and ($-1$)-sloped ray shooting queries, whose total number
is $O(k)$.
By exploiting some properties of our
problem shown in Section \ref{sec:observations}, we conduct all ray
shootings in a global manner in totally $O(n'+k)$ time.

\begin{algorithm}
\caption{Computing a shortest path map for $\bay$}
\label{algo:10}
{\footnotesize
\SetAlgoVlined
\KwIn{$\bay$, $R=\{r_1,r_2,\ldots,r_k\}$, and $\SPM(\calM)$ vertices
$v_1,v_2,\ldots,v_{k-1}$.}
\KwOut{A shortest path map on $\bay$ with respect to the source point $s$.}
\BlankLine
\tcc{Preprocessing}
Compute the ray set $\Psi=\{\rho_1,\rho_2,\ldots,\rho_{k-1}\}$ \;
Compute the line segment $\overline{v_ior(\rho_i)}$ for each $1\leq
i\leq k-1$ if $v_i\neq or(\rho_i)$ \;
Compute the horizontal visibility map $\HM(\bay)$ and the vertical visibility map
$\VM(\bay)$ \;
Compute the trapezoid in $\HM(\bay)$ that contains $or(\rho_i)$ for each
$1\leq i\leq k-1$ \;
\tcc{The main algorithm}
$p^*\leftarrow c$, $S\leftarrow \emptyset$,
$Q\leftarrow \{\rho_1,\rho_2,\ldots,\rho_{k-1}\}$ \tcc*[l]{$Q$ is a
queue storing the rays.}
\While{$Q$ is not empty}
{
   Consider the first ray $\rho$ in $Q$ and remove it from $Q$ \tcc*[l]{Assume $\rho$ is on
   $B(r_j,r_i)$ with $i>j$.}
   \eIf{$\rho$ is vertical}
   {Push $\rho$ onto the top of $S$, and exit the current loop \;
   }
   ( \tcc*[h]{$\rho$ is horizontal.})
   {
    Compute the target point $tp(\rho)$ \;
    \eIf{$S$ is empty}
     {The Voronoi region $\Vor(r_j)$ is determined with
     $\overline{or(\rho)tp(\rho)}$ \;
      $p^*\leftarrow tp(\rho)$, and exit the current loop \;
     }
     ( \tcc*[h]{$S$ is not empty; assume $\rho' \subset B(r_t,r_j)$ with
     $j>t$ is the
     ray at the top of $S$.})  
     {Scan $\partial(p^*,tp(\rho))$ to compute the target points on
     $\partial(p^*,tp(\rho))$ of the rays in $S$ \;
      \eIf{$tp(\rho')$ is before $tp(\rho)$ (i.e., $tp(\rho')$ has
      been computed)}
      {
       Determine the Voronoi
       regions for the roots defining the rays in $S$ \;
       Pop all rays out of $S$ \;
       $p^*\leftarrow tp(\rho)$, and exit the current loop \;
      }
      ( \tcc*[h]{$tp(\rho')$ is not before $tp(\rho)$ (i.e.,
      $tp(\rho')$ has not been computed).})  
      {
       Determine the Voronoi region $\Vor(r_j)$ \;
       Let $p$ be the intersection of $\rho$ and $\rho'$, and $q$ be the
       intersection of the horizontal line through $r_{i}$ and the
       vertical line through $r_t$; let $p'$ be the other intersection of
       $B_M(r_t,r_i)$ and the boundary of $Rec(p,q)$ than $p$\;
       Move from $p$ along $\overline{pp'}$ in $\HM(\bay)$ until either $p'$ or $\partial$
       is encountered first\;
       \eIf{$\partial$ is encountered (say, at the point $z$)}
       {
        Scan $\partial(tp(\rho),z)$ to compute the target points on
        $\partial(tp(\rho),z)$ of the rays in $S$ \;
        Determine the Voronoi regions for the roots defining the rays in $S$ \;
        Pop all rays out of $S$ \;
        $p^*\leftarrow z$, and exit the current loop \;
       }
      ( \tcc*[h]{$\partial$ is not encountered.})  
       {
        Pop $\rho'$ out of $S$ \;
        \eIf{$p'$ is on the bottom edge of $Rec(p,q)$}
        {
         Push the ray originating at $p'$ and going south onto the top
         of $S$ \;
        }
        ( \tcc*[h]{$p'$ is on the right edge of $Rec(p,q)$.})  
        {
         Add the ray originating at $p'$ and going east to the front
         of $Q$ \;
        }
         $p^*\leftarrow tp(\rho)$, and exit the current loop \;
       }
      }
     }
   }
}
For each $r_i\in\Psi$, compute the SPM on the Voronoi region $\Vor(r_i)$ with
respect to $r_i$\;
}
\end{algorithm}

The pseudo-code of Algorithm \ref{algo:10} summarizes the entire algorithm.

Before describing the main algorithm, we
discuss some preprocessing work as well as
some basic algorithmic methods that will be used later in the main algorithm.

\subsectionspace
\subsubsection{Preliminaries and Preprocessing}
\label{prelim-and-prepproc}

By Lemma \ref{lem:800}, for any two consecutive roots $r_{i-1}$ and
$r_i$ in $R$, $2\leq i\leq k$, if the middle
segment $B_M(r_{i-1},r_i)$ of their bisector intersects $\overline{cd}$ (at
$v_{i-1}$), then we can ``separately" process the portion of
$B_M(r_{i-1},r_i)$ inside $Rec(r_{i-1},r_i)$, as follows. Let
$\partial$ be the boundary of $bay(\overline{cd})$ minus
$\overline{cd}$, i.e., $\partial$ consists of all edges of $\bay$
except $\overline{cd}$.

\begin{figure}[t]
\begin{minipage}[t]{\linewidth}
\begin{center}
\includegraphics[totalheight=1.5in]{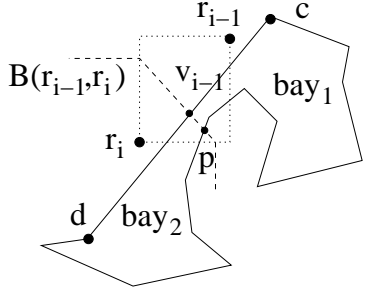}
\caption{\footnotesize Illustrating an example of $B_M(r_{i-1},r_i)$
intersecting both $\overline{cd}$ (at $v_{i-1}$) and $\partial$ (at
$p$).  The line segment $\overline{v_{i-1}p}$ divides
$\bay$ into $bay_1$ and $bay_2$.}
\label{fig:prebay}
\end{center}
\end{minipage}
\vspace*{-0.15in}
\end{figure}

Clearly, $v_{i-1}$ divides $B_M(r_{i-1},r_i)$ into two portions;
one portion does not contain any point in $\bay$ and the other
contains some points in $\bay$. Denote by $B_M'(r_{i-1},r_i)$
the portion that contains some points in $\bay$. Thus,
$B_M'(r_{i-1},r_i)$ is a line segment and $v_{i-1}$ is one of its
endpoints (and $or(\rho_{i-1})$ is the other endpoint).
We first determine whether $B_M'(r_{i-1},r_i)$ intersects
$\partial$, by performing a $-1$-sloped ray shooting operation.
Specifically, we shoot a ray $\rho$ originating at $v_{i-1}$
and passing through the other endpoint of
$B_M'(r_{i-1},r_i)$. If the length of the portion of $\rho$ between
$v_{i-1}$ and the first point $p$ on $\partial$ hit by $\rho$ is larger
than the length of $B_M'(r_{i-1},r_i)$, then $B_M'(r_{i-1},r_i)$ does not intersect
$\partial$, and we do nothing. Otherwise, $B_M'(r_{i-1},r_i)$ intersects
$\partial$ (at the point $p$).
By Lemma \ref{lem:800}, the line
segment $\overline{v_{i-1}p}$ appears in $\SPM(\calF)$.  Also,
$\overline{v_{i-1}p}$ partitions $bay(\overline{cd})$
into two simple polygons (see Fig.~\ref{fig:prebay}); one
polygon contains $\overline{cv_{i-1}}$ as an edge, which we denote
as $bay_1$, and we denote the other polygon as $bay_2$. Let
$R_1=\{r_1,r_2,\ldots,r_{i-1}\}$ and
$R_2=\{r_{i},r_{i+1},\ldots,r_{k}\}$.  (Note that since $R$ may be a
multi-set, $R_1$ and $R_2$ possibly refers to the same
physical root, but this is not important to our algorithm.) Since
$\overline{v_{i-1}p}$ is in $\Vor(\bay)$, it is not difficult to
see that for any point $q$ in $bay_1$, there is a root $r\in R_1$
such that a shortest path from $s$ to $q$ goes through $r$.
Similarly, for any point $q$ in $bay_2$, there is a root $r\in
R_2$ such that a shortest $s$-$q$ path goes through $r$.
This implies that we can divide the original problem of computing
$\Vor(bay(\overline{cd}))$ on $bay(\overline{cd})$ and $R$ into two
subproblems of computing $\Vor(bay_1)$ on $bay_1$ and $R_1$ and
computing $\Vor(bay_2)$ on $bay_2$ and $R_2$.

If we process each pair of consecutive roots in $R$ as above, then
the original problem may be divided into multiple subproblems, each
of which has the following property: For any pair of consecutive
roots $r_{i-1}$ and $r_i$ in the corresponding root subset of $R$, if
$B_M(r_{i-1},r_i)$ intersects $\overline{cd}$, then
$B'_M(r_{i-1},r_i)$ does not intersect $\partial$ and is contained in
the corresponding subpolygon of $bay(\overline{cd})$; further,
$B'_M(r_{i-1},r_i)$ is in $\Vor(\bay)$ and has been computed.

To perform the above process, a key is to derive an efficient
method for the $-1$-sloped ray shooting operations. For this,
we choose to check all pairs of consecutive roots in $R$ in the order of
$r_1,r_2,\ldots,r_k$. In this way, it is easy to see that the ray
shootings are conducted such that the origins of the rays are
sorted along $\overline{cd}$ from $c$ to $d$. This is summarized by
the next observation.

\lemmaspace
\begin{observation}\label{obser:50}
The preprocessing conducts $O(k)$ $-1$-sloped ray shooting operations
that are organized such that the origins of all rays are on $\overline{cd}$
ordered from $c$ to $d$.
\end{observation}
\lemmaspace

We show next that the ray
shootings for Observation \ref{obser:50} can be done in $O(n'+k)$ time.
Since the origins of all rays
in Observation \ref{obser:50} are sorted on $\overline{cd}$,
we can perform the ray shootings by computing the visible region of $bay(\overline{cd})$
 from $\overline{cd}$ along the direction of these rays.  This can be easily done by
a visibility algorithm on a simple polygon (e.g.,
\cite{ref:AtallahAn91,ref:JoeCo87,ref:LeeVi83}). Below, we
give a different algorithm for a more general problem; this more
general result is needed by the main algorithm.

Given a simple polygon $P$, the {\em horizontal visibility map} of
$P$ contains a horizontal line segment inside $P$ through each vertex of $P$,
extending as long as possible without properly crossing the
boundary of $P$ (such line segments are called the {\em diagonals};
see Fig.~\ref{fig:visMap}). The {\em vertical visibility map} with
vertical diagonals is defined similarly.  Each
region in a visibility map is a trapezoid (a triangle is a special
trapezoid). A visibility map of a simple polygon can be computed in
linear time \cite{ref:ChazelleTr91}.

\begin{figure}[t]
\begin{minipage}[t]{\linewidth}
\begin{center}
\includegraphics[totalheight=1.2in]{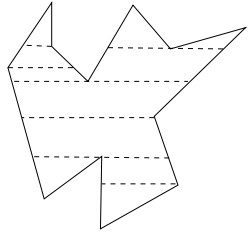}
\caption{\footnotesize Illustrating the horizontal visibility map of a simple
polygon.}
\label{fig:visMap}
\end{center}
\end{minipage}
\vspace*{-0.15in}
\end{figure}

For a ray $\rho$ with its origin in $bay(\overline{cd})$ (inside it
or on the boundary), the boundary point of $bay(\overline{cd})$ that
is not the origin $or(\rho)$ hit by $\rho$ first is called the {\em
target point} of $\rho$, denoted by $tp(\rho)$.
Recall that $\partial$ is the boundary of $\bay$ excluding the edge
$\overline{cd}$.
In the rest of this
paper, unless otherwise stated, a ray in our discussion always has
its origin in $\bay$ and its target point on $\partial$.

We say that $m$ parallel rays $\rho_1',\rho_2',\ldots,\rho_m'$
are {\em target-sorted} if we move from $c$ to $d$ (clockwise) on $\partial$,
we encounter the target points of these rays on $\partial$ in
the order of $tp(\rho'_1),tp(\rho'_2),\ldots,tp(\rho'_m)$.

Given a set of $m$ target-sorted parallel rays $\rho_1',\rho_2',\ldots,\rho_m'$
for $bay(\overline{cd})$ whose origins are in $\bay$ and whose target
points are on $\partial$, below we
present a {\em visibility map based approach} for computing their target
points in $O(n'+m)$ time (recall that $n'$ is the number of vertices of $\bay$).

WLOG, we assume that the rays are all horizontal.
We first compute the horizontal
visibility map of $bay(\overline{cd})$ in $O(n')$ time.
Then, starting from the vertex $c$, we scan $\partial$ and check each
edge $e$ of $\partial$ and the trapezoid $t(e)$ of the visibility map
bounded by $e$, to see whether the next
ray $\rho_i'$ (initially $i=1$) is in the trapezoid $t(e)$ and can hit the edge $e$.
Once the target point of the
ray $\rho_i'$ is found, we continue with the next ray $\rho_{i+1}'$.
Clearly, the time for computing all target points is $O(n'+m)$.
Thus, we have the following result.

\lemmaspace
\begin{lemma}\label{lem:rayshooting}
Given a set of $m$ target-sorted parallel rays for
$bay(\overline{cd})$ whose origins are in $\bay$ and whose target
points are on $\partial$, their target points can be computed in
$O(n'+m)$ time.
\end{lemma}
\lemmaspace

For the ray shootings in Observation \ref{obser:50}, it is easy to see
that these rays are target-sorted. Thus,
by Lemma \ref{lem:rayshooting}, their target points
can be computed in $O(n'+k)$ time
(of course, these ray shootings can be done by using the visibility
algorithms in \cite{ref:AtallahAn91,ref:JoeCo87,ref:LeeVi83},
which do not compute a visibility map).
We present the above visibility map based technique because
our main algorithm in Section \ref{sec:main} will need it.

In addition, as part of the preprocessing for our main algorithm, we also
compute the horizontal visibility map $\HM(\bay)$ and the vertical
visibility map $\VM(\bay)$ of $\bay$. Further, for each $1\leq i\leq
k-1$, we compute the trapezoid of the horizontal visibility
map $\HM(\bay)$ that contains
the origin $or(\rho_i)$ of the ray $\rho_i$, in totally $O(n'+k)$
time, in the following way.

Recall that $or(\rho_i)$ is either $v_i$ or in the interior of $\bay$.
In the latter case, $or(\rho_i)$ is an endpoint of the line segment
$B'_M(r_i,r_{i+1}) = \overline{v_ior(\rho_i)}$ whose slope is $-1$, and
the position of $or(\rho_i)$
has been determined earlier by the $-1$-sloped ray shooting
operations. By Lemma \ref{lem:90}, all origins
$or(\rho_1),or(\rho_2),\ldots,or(\rho_{k-1})$ are ordered
from northeast to southwest. Further, $or(\rho_i)$'s
are all visible from $\overline{cd}$ along the
direction of slope $-1$.  Thus,
it is not difficult to show that if we visit the trapezoids of $\HM(\bay)$
by scanning the edges of $\partial$ from $c$ to $d$ and looking at
the trapezoids bounded by each edge, then the trapezoids containing such
$or(\rho_i)$'s are encountered in the same order as
$or(\rho_1),or(\rho_2),\ldots,or(\rho_{k-1})$.  This implies that we can use
a similar algorithm as for computing the target points of target-sorted
parallel rays on $\partial$ (i.e., scanning $\partial$ from $c$ to $d$ and
checking the trapezoids of $\HM(\bay)$ thus visited along $\partial$) to
find all the sought trapezoids, in $O(n'+k)$ time.

The above discussion leads to the following lemma.

\lemmaspace
\begin{lemma}\label{lem:preprocess}
The preprocessing on $\bay$ takes $O(n'+k)$ time.
\end{lemma}
\lemmaspace

In the main algorithm, the horizontal visibility map
$\HM(\bay)$ will be used to guide the main process. More
specifically, during the algorithm, we traverse inside
$\bay$ following certain rays, and use $\HM(\bay)$ to keep
track of where we are (i.e.,
which trapezoid of $\HM(\bay)$ contains our current position). The vertical
visibility map $\VM(\bay)$ will be used to compute the target points
of some target-sorted vertical rays using the above visibility map based
approach.

For any two points $a$ and $b$ on $\partial$ with $a$ lying on the portion
of $\partial$ from $c$ clockwise to $b$, we denote by $\partial(a,b)$
the portion of $\partial$ between $a$ and $b$ and say that
$a$ is {\em before} $b$ or $b$ is {\em after} $a$.

\subsubsection{The Main Algorithm}
\label{sec:main}

After the preprocessing, the problem of computing
$\Vor(bay(\overline{cd}))$ with the root set $R$ may be divided into
multiple subproblems and we need to solve each subproblem. For
convenience of the forthcoming discussion, we assume that the
original problem on $bay(\overline{cd})$ with $R$ is merely one such
subproblem, i.e., for any two consecutive roots $r_i$ and $r_{i+1}$
in $R$, if $v_i\in B_M(r_i,r_{i+1})$, then $B'_M(r_i,r_{i+1})$
($=\overline{v_ior(\rho_i)}$) lies completely in
$\Vor(bay(\overline{cd}))$ and has been computed.
Recall that in the preprocessing, we have already computed the
trapezoid of the horizontal visibility map $\HM(\bay)$ that contains
the origin $or(\rho_i)$ of the ray $\rho_i$, for each $1\leq i\leq
k-1$. Observation \ref{obser:afterPre} below summarizes these facts.

\lemmaspace
\begin{observation}\label{obser:afterPre}
After the preprocessing,
\begin{itemize}
\item
for any two consecutive roots $r_i$ and
$r_{i+1}$ in $R$, if $v_i\in B_M(r_i,r_{i+1})$, then their bisector
portion
$B'_M(r_i,r_{i+1})$ ($=\overline{v_ior(\rho_i)}$) has been computed;
\item
for each $1\leq i\leq k-1$, the trapezoid of $\HM(\bay)$ that
contains the origin $or(\rho_i)$ of the ray $\rho_i$ is known.
\end{itemize}
\end{observation}
\lemmaspace

As discussed
before, our task is to handle the
interactions among the rays $\rho_i$ for all $i=1,2,\ldots,k-1$.


In the algorithm, we need to compute the target points for $O(k)$
horizontal and vertical rays. The main procedure is guided by the
horizontal visibility map $\HM(\bay)$ so that the target point of
each horizontal ray can be determined in constant time. For the
vertical ray shootings, we use the visibility map based approach
with the vertical visibility map $\VM(\bay)$. Note that the vertical
ray shootings will occur in an online fashion in the algorithm. We
will show that the vertical rays involved are target-sorted. To
compute the target points for these vertical rays, the algorithm
maintains a {\em reference point}, denoted by $p^*$. Initially,
$p^*=c$. Then during the algorithm, $p^*$ will be moved forward
along $\partial$ from $c$ to $d$, i.e., every time $p^*$ is moved on
$\partial$, its new position is always after its previous position.
In this way, the target points of all vertical rays are computed in
totally $O(n'+k)$ time (recall that $n'$ is the number of obstacle
vertices of $\bay$).

Let $\Psi=\{\rho_1,\rho_2,\ldots,\rho_{k-1}\}$. We process the rays
of $\Psi$ incrementally in the order of $\rho_1,\rho_2,\ldots,\rho_{k-1}$, whose
origins are ordered from northeast to southwest by Lemma 5 \ref{lem:90}. 
By Observation \ref{obser:40}, each ray in $\Psi$ is
either horizontally going east or vertically going south. We say
that initially all rays are {\em active} and the entire
$bay(\overline{cd})$ is {\em active}. In general, the active rays
are used to decompose the active region of $\bay$. During the
algorithm, some portion of $\bay$ will be implicitly set as {\em
inactive}, which means that each point of such a region is in the
Voronoi region of a root that has been determined. The active region
of $\bay$ at any moment of the algorithm always forms a connected
simple polygon, a fact that we will not explicitly argue in the
following algorithm description. Similarly, some rays will be set as
inactive, meaning that they will no longer be involved in the
further decomposition of the current active region of $\bay$. When
the algorithm terminates, the entire $bay(\overline{cd})$ is
inactive and all rays of $\Psi$ are inactive. Note that setting a
region or a ray as inactive is done implicitly and is used only for
our analysis. Since each ray in $\Psi$ lies on the bisector of two
roots in $R$, we say that the two roots {\em define} the ray.

\begin{figure}[t]
\begin{minipage}[t]{\linewidth}
\begin{center}
\includegraphics[totalheight=1.5in]{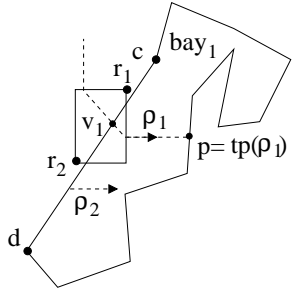}
\caption{\footnotesize Illustrating an example of $\rho_1$ being
horizontal.} \label{fig:rho1}
\end{center}
\end{minipage}
\vspace*{-0.15in}
\end{figure}

We start with the first ray $\rho_1$. If $\rho_1$ is horizontal
(going east), then since $or(\rho_1)$ is the most northeast origin,
no other ray in $\Psi$ can intersect it. Let $p$ be the target point
of $\rho_1$ on $\partial$ (see Fig.~\ref{fig:rho1}). Clearly, $p$
can be found in $O(1)$ time since we already know the trapezoid in
$\HM(\bay)$ that contains $or(\rho_1)$  by Observation
\ref{obser:afterPre}. Denote by $\alpha$ the
portion of $B(r_1,r_2)$ between $v_1$ and $p$. Note that $\alpha$ is
either the line segment $\overline{v_1p}$ (if $v_1=or(\rho_1)$), or
the concatenation of the two line segments
$\overline{v_1or(\rho_1)}$ and $\overline{or(\rho_1)p}$. In either
case, $\alpha$ partitions $bay(\overline{cd})$ into two simple
polygons. One of them contains $\overline{cv_1}$ as an edge and we
denote it by $bay_1$ (see Fig.~\ref{fig:rho1}). We claim that
$bay_1$ is the Voronoi region of $r_1$, i.e., $\Vor(r_1)=bay_1$.
Indeed, by the above analysis and Lemma \ref{lem:800}, $\alpha$ is
in $\Vor(\bay)$, implying that for any point $q\in bay_1$, there
is a shortest path from $s$ to $q$ via $r_1$. The claim thus
follows, and $\Vor(r_1)$ is determined. We then set the ray $\rho_1$
and the region of $bay_1$ as inactive. Hence, the active region of
$bay(\overline{cd})$ becomes $bay(\overline{cd})\setminus bay_1$,
which needs to be further decomposed. In addition, we move the
reference point $p^*$ from $c$ to $p$ ($=tp(\rho_1)$). We then
continue with the next ray $\rho_2$.

If $\rho_1$ is vertical (going south), then we push $\rho_1$ onto a
stack $S$ (initially, $S=\emptyset$), and let the reference point
$p^*$ stay at $c$.  We then continue with the next ray $\rho_2$.

We will show below that our algorithm maintains the following
general {\em invariants}, which are used to prove the correctness of
the algorithm. Suppose the current moment of the algorithm is right
before the next ray $\rho$
is considered, and assume
$\rho$ lying on the bisector $B(r_j,r_i)$ with $i>j$.  The stack $S$
may be non-empty; if $S=\emptyset$, then the invariants below
related to any rays in $S$ are not applicable. Let $\rho'$ be the
ray at the top of $S$, and suppose $\rho'$ lies on $B(r_t,r_{t'})$
with $t'>t$.

{\bf Invariant Properties:} (1) All rays in $S$ are active and
vertically going south. (2) The origins of all rays in $S$ from top
to bottom are ordered from southwest to northeast. (3) The origin of
the next ray to be considered by the algorithm (i.e., $\rho$) is to
the southwest of the origin of the ray at the top of $S$ (i.e.,
$\rho'$). (4) The two indices $j=t'$. (5) For each ray
$\rho''$ in $S\cup\{\rho\}$,
suppose $\rho''$ lies on the bisector $B(r_{j'},r_{i'})$
of two roots $r_{j'}$ and $r_{i'}$ with $i'>j'$; then the portion of
the boundary of the Voronoi region $\Vor(r_{i'})$ (resp.,
$\Vor(r_{j'})$) from $v_{i'-1}$ (resp., $v_{j'}$) to the origin
$or(\rho'')$ of $\rho''$ has already been computed. (6) For each ray
$\rho''$ in $S$, suppose it lies on the bisector $B(r_{j'},r_{i'})$
of two roots $r_{j'}$ and $r_{i'}$ with $i'>j'$; then $r_{j'}$ is to
the right of $\rho''$ and $r_{i'}$ is to the left of $\rho''$. (7)
The root $r_t$ is to the left of all rays in $S\setminus\{\rho'\}$
(recall that $\rho'\subset B(r_{t},r_{t'})$ with $t'>t$). (8) For
any two consecutive rays $\rho'_1$ and $\rho'_2$ in $S$ such that
$\rho'_1$ is closer to the top of $S$, suppose $\rho'_1$ is on
$B(r_{i_1},r_{i_2})$ for $i_2>i_1$ and $\rho'_2$ is on
$B(r_{j_1},r_{j_2})$ for $j_2>j_1$; then $i_1=j_2$. (9) The target
points of all rays in $S$ from bottom to top are ordered clockwise
on $\partial$ (i.e., from $c$ to $d$). (10) If $\rho$ is vertical,
then the target point $tp(\rho)$ of $\rho$ is after the target point
$tp(\rho')$ of $\rho'$ on $\partial$. (11) If the target point of
any ray in $S$ has not been computed yet, then the target point of
that ray is after the reference point $p^*$ (i.e., on
$\partial(p^*,d)$). (12) The target point $tp(\rho)$ is after $p^*$.
(13) Suppose $\rho''$ is the first {\em horizontal} ray in $\Psi$
that will be considered by the algorithm in a future time from now;
then its target point $tp(\rho'')$ is after $p^*$. (14) The
trapezoid in the horizontal visibility map $\HM(\bay)$ that contains the origin
$or(\rho)$ of the ray $\rho$ is known.

Now consider the moment that is right after we finish processing
$\rho_1$ and before we consider $\rho_2$. Based on the processing of
$\rho_1$ discussed above, either $\rho_1$ is horizontal and $S$ is
empty, or $\rho_1$ is vertical and $S= \{\rho_1\}$. Lemma
\ref{lem:invariants10} below shows that in either case, all
invariants of the algorithm hold. We intend to use the proof of
Lemma \ref{lem:invariants10} as a ``warm-up" for the analysis of the
more general situations later. Observation \ref{obser:new10} follows
from the definitions of the rays in $\Psi$ and Lemma \ref{lem:90}.

\lemmaspace
\begin{observation}\label{obser:new10}
The target points of all rays in $\Psi$ are on $\partial$. For any
two rays $r_j$ and $r_i$ in $\Psi$ with $i>j$, if $r_j$ is horizontal
or $r_i$ is vertical, then $tp(r_i)$ is after $tp(r_j)$ on $\partial$.
\end{observation}

\begin{lemma}\label{lem:invariants10}
At the moment after $\rho_1$ has been processed and before $\rho_2$ is
considered, all invariants of the algorithm hold.
\end{lemma}
\begin{proof}
Recall that $\rho_1$ is on $B(r_1,r_2)$ and $\rho_2$ is on
$B(r_2,r_3)$, and the reference point $p^*$ is at the target point
$tp(\rho_1)$ if $\rho_1$ is horizontal and at the vertex $c$ otherwise.

We first discuss the case when $\rho_1$ is horizontal, in which $S$
is empty and $p^*=tp(\rho_1)$.
Invariants (1) through (11) except
(5) simply follow since they are all related to some rays in $S$.
For Invariant (5), we only need to consider $\rho_2\subset B(r_2,r_3)$,
i.e., we need to show that the portion of the boundary of the
Voronoi region $\Vor(r_{3})$ from $v_{2}$ to the origin $or(\rho_2)$
of $\rho_2$, which is also the boundary portion of the Voronoi region
$\Vor(r_{2})$ from $v_{2}$ to $or(\rho_2)$,
has already been computed. Denote this boundary portion by
$\alpha$. Note that $\alpha$ is the portion
of $B(r_2,r_3)$ between $v_2$ and $or(\rho_2)$. Recall that
$or(\rho_2)$ is either $v_2$ or not. If $or(\rho_2)=v_2$, then we
are done since $\alpha$ is just a single point $v_2$.
Otherwise, $v_2$ must be on $B_M(r_2,r_3)$ and $\alpha$ is
$B'_M(r_2,r_3)$ ($=\overline{v_2or(\rho_2)}$), which has been computed in
our preprocessing by Observation \ref{obser:afterPre}.
Hence, Invariant (5) follows.

For Invariant (12), we need to show that $tp(\rho_2)$ is after
$p^*=tp(\rho_1)$, which is true due to Observation \ref{obser:new10}
and $\rho_1$ being horizontal. For Invariant (13), let
$i>1$ be the smallest index such that $\rho_i\in\Psi$ is horizontal.
If there is no such $i$, then Invariant (13) trivially holds;
otherwise, we need to prove that $tp(\rho_i)$ is after
$p^*=tp(\rho_1)$, which is true due to Observation \ref{obser:new10}
and $\rho_1$ being horizontal. For Invariant (14), we need to show
that the trapezoid of $\HM(\bay)$ containing $or(\rho_2)$ is known,
which is true by Observation \ref{obser:afterPre}.
Hence, when $\rho_1$ is horizontal, all invariants hold.

We then discuss the case when $\rho_1$ is vertical, in which
$S= \{\rho_1\}$ and the reference point $p^*=c$. Invariants
(1) and (2) simply follow. By Lemma \ref{lem:90}, $or(\rho_2)$ is to
the southwest of $or(\rho_1)$, and thus Invariant (3) holds.
Invariant (4) is obvious.
For Invariant (5), we need to consider both
$\rho_1$ and $\rho_2$. The proof is similar to that for the case
when $\rho_1$ is horizontal, and we omit it.
For Invariant (6), we need to show that $r_1$ is to the right of
$\rho_1$ and $r_2$ is to the left of $\rho_1$, which is true due to
Lemma \ref{lem:1000} and $\rho_1$ being vertical. Invariant (7)
simply follows since $\rho_1$ is the only ray in $S$. Invariants (8)
and (9) trivially hold since $S$ has only one ray. For
Invariant (10), we need to show that if $\rho_2$ is vertical,
then $tp(\rho_2)$ is after $tp(\rho_1)$, which is true by
Observation \ref{obser:new10}. For Invariant (11), note that the
target point $tp(\rho_1)$ has not been computed. Since $p^*=c$,
Invariant (11) trivially holds. Invariants (12) and (13) also
easily hold since $p^*=c$ and the target
points of all rays in $\Psi$ are on $\partial$.
For Invariant (14), we need to show that the trapezoid of $\HM(\bay)$
that contains $or(\rho_2)$ is known, which is true by Observation
\ref{obser:afterPre}.

We hence conclude that all invariants of the algorithm hold.
\end{proof}

As an implementation detail, although we view $S$ as a stack,
we represent $S$ as a doubly-linked list so that we can
access the rays in $S$ from both the top and the bottom of $S$. But,
we always pop a ray out of $S$ from its top and push a ray onto $S$
at its top. Next, we discuss the general situations of our algorithm.

Suppose our algorithm just starts to process a ray $\rho_{i}\in \Psi$, $i>1$,
which lies on the bisector $B(r_i,r_{i+1})$, and all invariants of the
algorithm hold right before $\rho_i$ is processed. There are a number of cases
and subcases to consider, depending on whether $\rho_i$ is vertical or horizontal,
whether $S$ is empty, and the intersecting consequences between $\rho_i$ and
the rays in $S$ (if $S\not=\emptyset$), etc.

\vspace*{0.05in}
\noindent
{\bf Case 1}: $\rho_i$ is
vertical (going south). Then we simply push $\rho_i$ onto the top of $S$ and
the reference point $p^*$ is not changed. The algorithm then continues with the next
ray $\rho_{i+1}\in \Psi$ in this situation.
Lemma \ref{lem:invariants30} below shows that all invariants of the algorithm hold.

\lemmaspace
\begin{lemma}\label{lem:invariants30}
If the ray $\rho_i\in \Psi$ is vertical, then at the moment after $\rho_i$ is
processed and before $\rho_{i+1}$ is considered, all invariants of
the algorithm hold.
\end{lemma}
\lemmaspace
\begin{proof}
Note that $\rho_{i+1}$ is on $B(r_{i+1},r_{i+2})$. Let $\xi$ be the
moment right after $\rho_i$ is processed and $\xi'$ be the
moment right before $\rho_i$ is considered. Thus, from $\xi'$ to
$\xi$, the only change to $S$ is that we push $\rho_i$ onto the top of
$S$. The proof below is based on the assumption that $S$
has at least two rays at the moment $\xi$ (i.e., $\rho_i$ and at
least another ray), since otherwise the invariants related to other
rays in $S$ than $\rho_i$ trivially hold. This also implies that $S$
is not empty at the moment $\xi'$. Let $\rho$ be the ray at the top
of $S$ at the moment $\xi'$, and assume $\rho$ lying on the bisector
$B(r_{j},r_{j'})$ with $j'>j$.

Invariant (1) holds since $\rho_i$ is vertical.

For Invariant (2), since all invariants of the algorithm hold at the
moment $\xi'$, it suffices to show that $or(\rho_i)$ is to the
southwest of $or(\rho)$. Note that at the moment $\xi'$, $\rho_i$ is
the next ray to be considered by the algorithm. Thus, by
Invariant (3) at the moment $\xi'$, $or(\rho_i)$ is to the southwest of
$or(\rho)$. Invariant (2) thus follows.

For Invariant (3), Lemma \ref{lem:90} implies that $or(\rho_{i+1})$ is to
the southwest of $or(\rho_i)$.

Invariant (4) trivially holds since $\rho_i\subset B(r_{i},r_{i+1})$ and
$\rho_{i+1}\subset B(r_{i+1},r_{i+2})$.

For Invariant (5), it suffices to consider the ray $\rho_{i+1}$, i.e.,
to show that the portion of the
boundary of $\Vor(r_{i+1})$ between $v_{i+1}$ and $or(\rho_{i+1})$, which is also
the portion of the boundary of $\Vor(r_{i+2})$ between $v_{i+1}$ and
$or(\rho_{i+1})$, has already been computed. (Note that the case for
the ray $\rho_i$ trivially holds due to Invariant (5) at the moment
$\xi'$ when $\rho_i$ is the next ray to be considered.)
Recall that the above boundary portion
is a single point $v_{i+1}$ if $v_{i+1}=or(\rho_{i+1})$ and is the line segment
$\overline{v_{i+1}or(\rho_{i+1})}$ otherwise. By Observation
\ref{obser:afterPre}, if $v_{i+1}\neq
or(\rho_{i+1})$, then $\overline{v_{i+1}or(\rho_{i+1})}$ has already
been computed in the preprocessing. Thus, Invariant (5) follows.

For Invariant (6), it suffices to prove that $r_{i}$ is to the
right of $\rho_i$ and $r_{i+1}$ is to the left of $\rho_i$, which follows from
Lemma \ref{lem:1000} since $\rho_i$ is vertical.

For Invariant (7), it suffices to show that $r_i$ is to the left
of $\rho$. Recall that $\rho$ is on $B(r_{j},r_{j'})$ with $j'>j$.
At the moment $\xi'$, by Invariant (4), $j'=i$; by Invariant (6),
$r_i$ ($=r_{j'}$) is to
the left of $\rho$. Thus, Invariant (7) follows.

For Invariant (8), it suffices to show $j'=i$, which has been
proved above for Invariant (7).

For Invariant (9), it suffices to show that $tp(\rho_i)$ is after
$tp(\rho)$ on $\partial$. Since $\rho_i$ is vertical, at the moment $\xi'$,
by Invariant (10), $tp(\rho_i)$ is
after $tp(\rho)$ on $\partial$. Invariant (9) thus follows.

For Invariant (10), we need to prove that if $\rho_{i+1}$ is
vertical, then $tp(\rho_{i+1})$ is after $tp(\rho_i)$ on $\partial$,
which follows from Observation \ref{obser:new10}.

For Invariant (11), note that the target point $tp(\rho_i)$ has not
been computed.  We need to show that $tp(\rho_i)$ is after
$p^*$ on $\partial$. At the moment $\xi'$,  by Invariant (12), $tp(\rho_i)$ is
after $p^*$. Further, $p^*$ has not been moved since the moment
$\xi'$. Invariant (11) thus follows.

For Invariant (12), we need to show that $tp(\rho_{i+1})$ is after
$p^*$. If $\rho_{i+1}$ is vertical, then by Observation
\ref{obser:new10}, $tp(\rho_{i+1})$ is after $tp(\rho_i)$ on $\partial$,
and we have also
shown above that $tp(\rho_i)$ is after $p^*$; thus $tp(\rho_{i+1})$
is after $p^*$. If $\rho_{i+1}$ is
horizontal, then at the moment $\xi'$, since $\rho_i$
is vertical, the first horizontal ray in $\Psi$ to be considered by the
algorithm in future is $\rho_{i+1}$; thus by Invariant (13),
$tp(\rho_{i+1})$ is after $p^*$.
Invariant (12) then follows.

Invariant (13) trivially holds since $\rho_i$ is vertical. More
specifically, suppose the first horizontal ray in $\Psi$ that will
be considered by the algorithm after the moment $\xi'$ is $\rho_j$.
Note that $j\geq i$. Since all invariants of the algorithm hold at
the moment $\xi'$, by Invariant (13), $tp(r_j)$ is after $p^*$
on $\partial$. Then at the moment $\xi$,
since $\rho_i$ is vertical, the first horizontal ray in $\Psi$
to be considered by the algorithm is still $\rho_j$. Proving that
Invariant (13) holds at the moment $\xi$ is to prove that
$tp(r_j)$ is after $p^*$, which has been proved above since $p^*$
has not been moved since the moment $\xi'$.

For Invariant (14), we need to show that the trapezoid of $\HM(\bay)$
that contains $or(\rho_{i+1})$ is known, which is true by Observation
\ref{obser:afterPre}.

We conclude that all invariants of the algorithm hold at the moment
$\xi$.
\end{proof}

\vspace*{0.05in}
\noindent
{\bf Case 2}:
$\rho_i$ is horizontal (going east).  Let $p=tp(\rho_i)$. We claim
that we can find $p$ in constant time. Indeed, since $\rho_i$ is the
next ray considered by the algorithm, by Invariant (14), the trapezoid of
$\HM(\bay)$ that contains $or(\rho_i)$ is known. The claim then follows
since $p$ is on the boundary of the above trapezoid.
Since $\rho_i$ is horizontal, by Invariant (13)
(at the moment right before processing $\rho_i$), $p$
is after the reference point $p^*$.
Depending on whether the stack $S$ is empty, there are two
subcases to consider.

\vspace*{0.05in}
\noindent
{\bf Subcase 2(a)}:
$S=\emptyset$. Then no ray in $S$ intersects $\rho_i$
before it hits $\partial$
(and thus no ray shooting for any ray $\rho_j\in \Psi$ with $j<i$
intersects $\rho_i$ before hitting $\partial$).
Also, for each ray $\rho_j\in \Psi$ with $j>i$,
since $or(\rho_j)$ is to the southwest of $or(\rho_i)$ and $\rho_i$ is
horizontal, $\rho_j$ cannot intersect $\rho_i$.
Hence, the portion $\overline{or(\rho_i)p}$ of the ray $\rho_i$
appears in $\Vor(\bay)$. Recall that $\rho_i$ is on
$B(r_i,r_{i+1})$. The portion of $B(r_i,r_{i+1})$ between
$v_{i}$ and $p$ divides the current active region of
$bay(\overline{cd})$ into two simple polygons; one of them contains
$\overline{v_{i-1}v_i}$ and we denote it by $bay_i$.
Further, each point in $bay_i$ has a shortest path to $s$ via $r_i$.
Thus, $bay_i$ is the Voronoi region $\Vor(r_i)$.
We then set $\rho_{i}$ and the region
$bay_i$ as inactive.  In addition, we move $p^*$ to $p$. We then consider the
next ray $\rho_{i+1}$. We prove below that all invariants of the
algorithm hold right after processing $\rho_i$.

Since $S$ is empty,
Invariants (1) to (11) except (5) simply hold since
they are all related to some rays in $S$. For Invariant (5), we only
need to consider $\rho_{i+1}$, which also holds by Observation
\ref{obser:afterPre} (the analysis is similar as before).
For Invariant (12), we need
to show that $tp(\rho_{i+1})$ is after $p^*$. Since $\rho_i$ is
horizontal, by Observation \ref{obser:new10}, $tp(\rho_{i+1})$ is
after $tp(\rho_i)$ ($=p^*$). Thus, Invariant (12) follows. For
Invariant (13), suppose $\rho_j\in\Psi$ is the first horizontal ray
to be considered by the algorithm. Note that it must be
$j>i$.  We need to show that $tp(\rho_j)$ is after $p^*$
($=tp(\rho_i)$), which is true by Observation \ref{obser:new10}
since $\rho_i$ is horizontal.
For Invariant (14), we need to show that the trapezoid of $\HM(\bay)$
that contains $or(\rho_{i+1})$ is known, which is true by Observation
\ref{obser:afterPre}.  Therefore, all invariants of the
algorithm hold right after processing $\rho_i$.

\vspace*{0.05in}
\noindent
{\bf Subcase 2(b)}:
$S\not=\emptyset$. Then
for the rays in $S$ whose target points lie on $\partial(p^*,p)$,
we compute their target points
by scanning $\partial(p^*,p)$ from $p^*$ to $p$ ($=tp(\rho_i)$);
this scanning process uses
the visibility map based approach with $\VM(\bay)$,
as described in the preprocessing.
By Invariant (9), such vertical rays (from bottom to top in $S$)
are target-sorted. Thus, the scanning
procedure takes linear time in terms of the number of edges of
$\partial(p^*,p)$ and the number of target points
found in this process. Note that the scanning procedure stops when
we encounter the point $p$. This also implies that the target points
of some rays in $S$ (e.g., the ray at the top of $S$) are not yet found if
they are on $\partial$ after $p$.

Let $\rho$ be the ray at the top of
$S$ (e.g., if $\rho_{i-1}$ is vertical, then $\rho$ is $\rho_{i-1}$).
Suppose $\rho$ is on the bisector $B(r_{j},r_{j'})$ with
$j'>j$. Then right before $\rho_i$ is processed,
by Invariant (4), $i=j'$ since $\rho_i\subset B(r_{i},r_{i+1})$;
by Invariant (3), $or(\rho_i)$ is to the southwest of $or(\rho)$.
Depending on whether the target point $tp(\rho)$ is before $p$,
there are two subcases.

\begin{figure}[t]
\begin{minipage}[t]{\linewidth}
\begin{center}
\includegraphics[totalheight=1.5in]{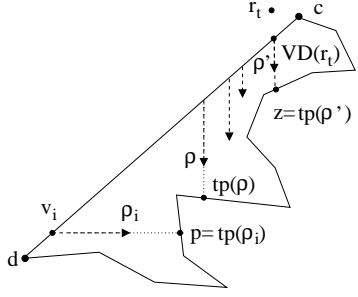}
\caption{\footnotesize Illustrating an example that the target
points of all rays in $S$ are before $p=tp(\rho_i)$. All vertical
rays are in $S$. The ray $\rho$ is at the top of $S$ and $\rho'$ is
at the bottom of $S$.} \label{fig:splitting}
\end{center}
\end{minipage}
\vspace*{-0.15in}
\end{figure}

\vspace*{0.05in}
\noindent
{\bf Subcase 2(b.1)}:
The target point $tp(\rho)$ is before $p=tp(\rho_i)$.
In this case, the scanning procedure has found $tp(\rho)$ on
$\partial(p^*,p)$.
Then by Invariants (9) and (11), the target points
of all rays in $S$ have been obtained and all such
target points
are before $p$ on $\partial$.  Since $p$ is the target point of
$\rho_i$, the above implies that
all rays in $S$ hit $\partial$ before they intersect $\rho_i$ (see
Fig.~\ref{fig:splitting}). Further, since all other active rays in $\Psi$, i.e.,
$\rho_{i+1},\rho_{i+2},\ldots,\rho_{k-1}$, have their origins to the
southwest of $or(\rho_i)$, no ray in $S$ can intersect these
active rays before it hits $\partial$. This means that for each
ray $\rho'$ in $S$, the portion of $\rho'$ between $or(\rho')$ and
$tp(\rho')$ appears in $\Vor(\bay)$. Based on the discussion
above, we perform a {\em splitting procedure} on $S$, as follows.

Let $\rho'$ be the ray at the {\em bottom} of $S$ and $z=tp(\rho')$
(see Fig.~\ref{fig:splitting}). Suppose $\rho'$ is on the bisector
$B(r_{t},r_{t'})$ with $t'>t$. By Invariant (5), the boundary portion
of $\Vor(r_t)$ between $v_t$ and $or(\rho')$ has been computed. The
concatenation of the segment $\overline{or(\rho')z}$ and this
boundary portion of $\Vor(r_t)$ splits the current active
region of $bay(\overline{cd})$ into two simple polygons.
One of them contains $\overline{v_{t-1}v_t}$ as an edge; further,
each point in this polygon has a shortest path to $s$ via $r_t$.
Thus, the polygon containing $\overline{v_{t-1}v_t}$ is the Voronoi
region $\Vor(r_t)$.
We also set the region $\Vor(r_t)$ as inactive.

We then continue to process the second bottom ray in $S$, in the similar fashion.
This splitting procedure stops once all rays in $S$ are
processed. In addition, we set all rays in $S$ as inactive and pop them
out of $S$ ($S$ then becomes empty). Finally, we move the reference
point $p^*$ to $p$ ($=tp(\rho_i)$), and consider the next ray
$\rho_{i+1}$. By the same analysis as that for the subcase 2(a) when $S$ is
empty, we can prove that all invariants of the algorithm hold. We
omit the details.

\vspace*{0.05in}
\noindent
{\bf Subcase 2(b.2)}:
The target point $tp(\rho)$ is not before $p=tp(\rho_i)$.
In this case, $tp(\rho)$ has not been found on $\partial(p^*,p)$ by the
scanning procedure. Then, it is easy to see that $\rho$ intersects $\rho_i$
before it hits $\partial$ (and so may some other rays in $S$).
We need to consider the consequences of the intersections of such rays in $S$
with $\rho_i$. Recall that $\rho_i$ is on
$B(r_i,r_{i+1})$ and $\rho$ is on $B(r_j,r_i)$ with $i>j$. Below
we show how to determine the Voronoi region $\Vor(r_i)$ and the portion
of the bisector $B(r_j,r_{i+1})$ in $\Vor(\bay)$. Let $p_1$
be the intersection point of $\rho_i$ and $\rho$ (see
Fig.~\ref{fig:rayintersection}).

First of all, we determine the Voronoi region $\Vor(r_i)$ (see
Fig.~\ref{fig:rayintersection}).
Since $\rho$ is the leftmost ray in $S$ by Invariant (2),
both the line segments $\overline{or(\rho_i)p_1}$ and
$\overline{or(\rho)p_1}$ appear in $\Vor(\bay)$.
Since the ray $\rho$ is in the stack $S$,
by Invariant (5), the boundary portion of $\Vor(r_i)$ between
$v_{i-1}$ and $or(\rho)$ has been computed, which we denote by $\alpha$.
At the moment right before $\rho_i$ is processed, since $\rho_i$ is
the next ray to be considered, also by Invariant (5), the
boundary portion of $\Vor(r_i)$ between
$v_{i}$ and $or(\rho_i)$ has been computed, which we denote by $\beta$
(i.e., $\beta=v_i$ if $v_i=or(\rho_i)$ and
$\beta=\overline{v_ior(\rho_i)}$ otherwise).
As argued similarly in the earlier analysis, $\Vor(r_i)$ is
the region bounded clockwise by $\alpha$, the segment
$\overline{or(\rho)p_1}$, the segment $\overline{or(\rho_i)p_1}$,
$\beta$, and $\overline{v_{i-1}v_i}$. This region of $\VD(r_i)$ is
then set as inactive.

\begin{figure}[t]
\begin{minipage}[t]{\linewidth}
\begin{center}
\includegraphics[totalheight=2.0in]{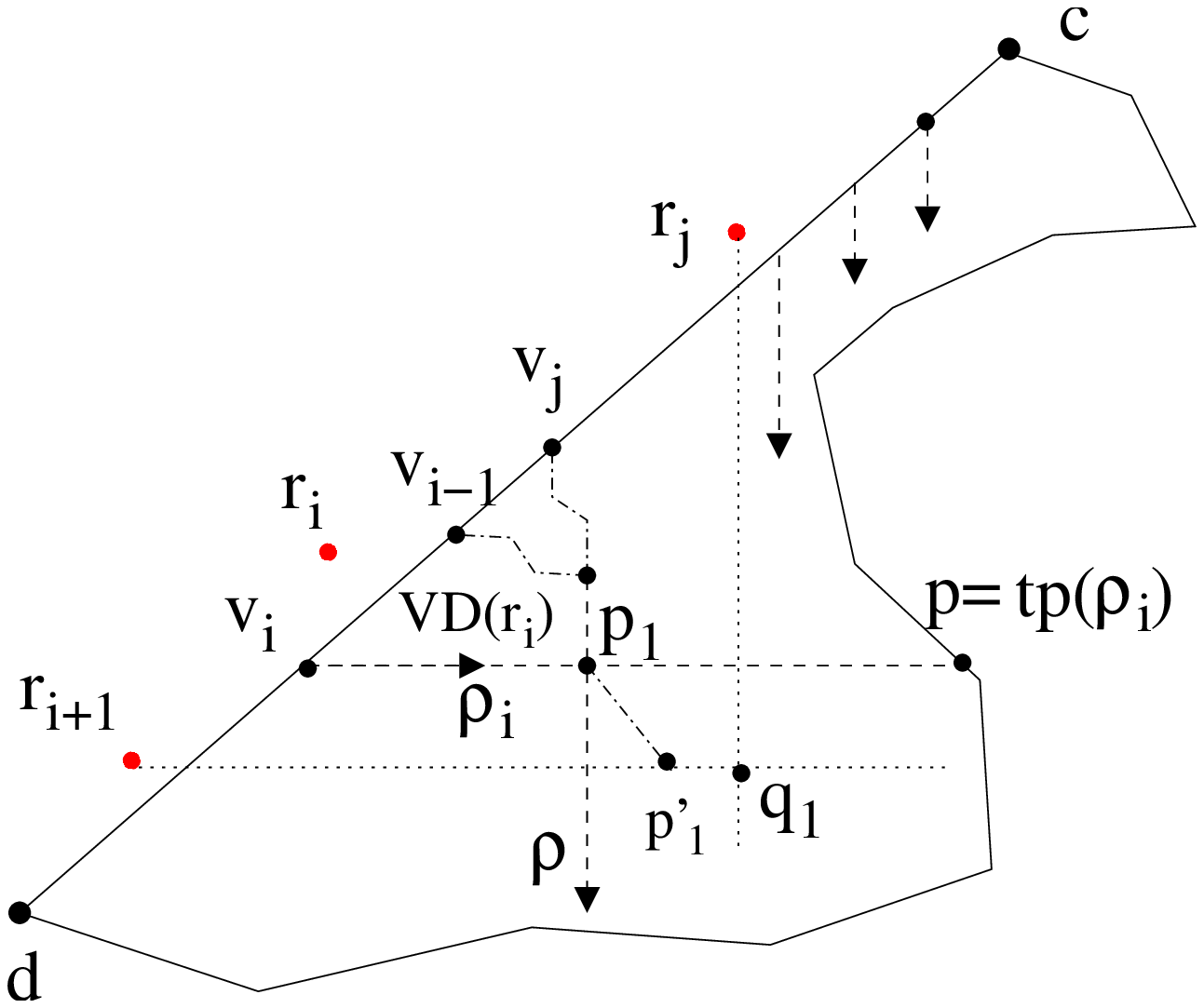}
\caption{\footnotesize Illustrating an example that
the ray $\rho$ at the top of $S$ intersects $\rho_i$ (at $p_1$) before $\rho$
hits $\partial$.}
\label{fig:rayintersection}
\end{center}
\end{minipage}
\vspace*{-0.15in}
\end{figure}

Second, we determine the portion of the bisector $B(r_j,r_{i+1})$ that
appears in $\Vor(\overline{cd})$. Since $\rho_i$ is horizontal,
by Lemma \ref{lem:1000}, the root $r_{i+1}$ is below $\rho_i$. Since $\rho\subset
B(r_j,r_i)$ with $i>j$, by Invariant (6), $r_j$ is to the right
of $\rho$. Therefore, the intersection point (denoted by $q_1$) of
the horizontal line through $r_{i+1}$ and the vertical line through
$r_j$ is to the southeast of $p_1$ (see
Fig.~\ref{fig:rayintersection}). We first discuss the portion of
$B(r_j,r_{i+1})$ contained in the rectangle $Rec(p_1,q_1)$.

Obviously, $Rec(p_1,q_1)$ is contained in the rectangle
$Rec(r_j,r_{i+1})$. Thus, the portion of $B(r_j,r_{i+1})$  in
$Rec(p_1,q_1)$ is a portion of the middle segment of
$B(r_j,r_{i+1})$. Further, since $p_1$ is the intersection of
$\rho_i$ and $\rho$, $p_1$ is at the intersection of $B(r_i,r_{i+1})$
and $B(r_j,r_i)$. Thus, $p_1$ is on $B(r_j,r_{i+1})$.

We claim that $r_{i+1}$ is to the southwest of $r_{j}$. This can be
proved by showing that $r_{i+1}$ is to the southwest of $p_1$ and
$p_1$ is to the southwest of $r_j$. Indeed, since $r_{i+1}$ is below
$\rho_i$ and $\overline{cd}$ is positive-sloped, $p_1$ must be to
the right of $r_{i+1}$, which also implies that $r_{i+1}$ is to the
southwest of $p_1$. Similarly, we can show that $p_1$ is to the
southwest of $r_j$.

Because $r_{i+1}$ is to the southwest of $r_{j}$, the middle segment
of $B(r_j,r_{i+1})$ is $-1$-sloped. Denote by $B_M'(r_j,r_{i+1})$
the portion of $B(r_j,r_{i+1})$ contained in $Rec(p_1,q_1)$. Based
on the above analysis, $B_M'(r_j,r_{i+1})$ is a $-1$-sloped line
segment with an endpoint at $p_1$ and the other endpoint on one of
the two edges of $Rec(p_1,q_1)$ incident to $q_1$ (see
Fig.~\ref{fig:rayintersection}). Below we prove that
$B_M'(r_j,r_{i+1})\cap \bay$ (i.e., the portion of
$B_M'(r_j,r_{i+1})$ contained in $\bay$) appears in $\Vor(\bay)$, implying that
we should keep this portion of
$B_M'(r_j,r_{i+1})$. The
proof is similar to that for Lemma \ref{lem:800} and hence we only sketch it
here. For convenience, we view $B_M'(r_j,r_{i+1})$ as the open
segment that does not contain its two endpoints.

It suffices to show that $B_M'(r_j,r_{i+1})$ does not intersect
any current active ray. Consider any current active ray $\rho'$,
$\rho'\not\in \{\rho,\rho_{i}\}$. Then $\rho'$ either
is in $S$ or is a ray $\rho_t\in\Psi$ with $t>i$.

\begin{itemize}
\item
If $\rho' \in S$, then $\rho'$ is vertical by Invariant (1). By Invariant
(7), $\rho'$ is to the right of the root $r_j$, and thus to the right of
the rectangle $Rec(p_1,q_1)$. Hence, $\rho_t$ does not intersect
$B_M'(r_j,r_{i+1})$ since $B_M'(r_j,r_{i+1})$ is strictly inside
$Rec(p_1,q_1)$.

\item
If $\rho'= \rho_t\in \Psi$ with $t>i$, then there are two subcases.

If $\rho_t$ is horizontal, then by Lemma \ref{lem:100}, $\rho_t$ is below
$r_{i+1}$, and is thus below the rectangle $Rec(p_1,q_1)$. Hence,
$\rho_t$ does not intersect $B_M'(r_j,r_{i+1})$.

If $\rho_t$ is vertical, then by Lemma \ref{lem:90}, the origin
$or(\rho_i)$ is to the northeast of $or(\rho_t)$. Clearly, $or(\rho_i)$
is to the left of $Rec(p_1,q_1)$ and thus $or(\rho_t)$ is to the left
of $Rec(p_1,q_1)$. Since $\rho_t$ is vertical, $\rho_t$ is also to
the left of $Rec(p_1,q_1)$. Hence, $\rho_t$ does not intersect
$B_M'(r_j,r_{i+1})$.
\end{itemize}

The above argument shows that all active rays cannot intersect
$B_M'(r_j,r_{i+1})$. Hence, the portion of $B_M'(r_j,r_{i+1})$
contained in $\bay$ must appear in $\Vor(\bay)$.

Then, we compute $B_M'(r_j,r_{i+1})$ in $O(1)$ time, and let
$B_M'(r_j,r_{i+1})=\overline{p_1p_1'}$ (see
Fig.~\ref{fig:rayintersection}). Note that $p_1'$ is either on the
right edge or the bottom edge of $Rec(p_1,q_1)$.

However, $\overline{p_1p_1'}$ may intersect $\partial$.
To determine whether such intersection occurs, we move among
the trapezoids in the horizontal
visibility map $\HM(\bay)$ from the endpoint $p_1$ of
$B_M'(r_j,r_{i+1})$ along the segment $\overline{p_1p_1'}$, as follows.

Note that the portion of the ray $\rho_i$ between its origin $or(\rho_i)$
and its target point $p=tp(\rho_i)$ is contained in a
single trapezoid of $\HM(\bay)$,
i.e., the trapezoid containing $or(\rho_i)$, which is already known according
to Invariant (14). Further, this trapezoid is the one that
contains $p_1$ since $p_1\in \overline{or(\rho_i)p}$.
Starting at $p_1$ in this trapezoid, we move along the
segment $\overline{p_1p_1'}$, and enter/exit trapezoids in $\HM(\bay)$ one after another,
until we encounter either $p_1'$ or
an edge of $\partial$ for the first time. In this way, we can
determine whether $\overline{p_1p_1'}$ intersects $\partial$.
Further, if $\overline{p_1p_1'}$ intersects $\partial$, then the first such
intersection point, denoted by $z$, is also found in this moving
process; if $B_M'(r_j,r_{i+1})$ does not intersect $\partial$,
then the trapezoid of $\HM(\bay)$ containing the point $p_1'$ is determined.
It is easy to see that the running time of the above moving procedure is
proportional to the number of trapezoids in $\HM(\bay)$ that we visit
when moving along $\overline{p_1p_1'}$. We will analyze the total running time
of the moving process in a global manner later.

Depending on whether $B_M'(r_j,r_{i+1})$ ($=\overline{p_1p_1'}$) intersects
$\partial$, there are two cases to consider.

\begin{figure}[t]
\begin{minipage}[t]{\linewidth}
\begin{center}
\includegraphics[totalheight=2.0in]{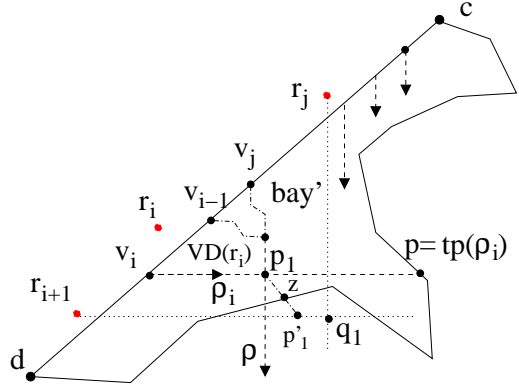}
\caption{\footnotesize Illustrating an example that
$B'_M(r_j,r_{i+1})$ ($=\overline{p_1p_1'}$) intersects $\partial$
(first at $z$).}
\label{fig:rayintersectioncase1}
\end{center}
\end{minipage}
\vspace*{-0.15in}
\end{figure}

If $B_M'(r_j,r_{i+1})$ intersects $\partial$, then we have found the first
intersection point $z$ of $B_M'(r_j,r_{i+1})$ and $\partial$ (see
Fig.~\ref{fig:rayintersectioncase1}).  Note that $z$ must be after $p$
on $\partial$.  Also, note that the Voronoi region $\Vor(r_i)$ has been
computed and set as inactive,
and thus $p_1$ lies on the boundary of the current active region of $\bay$ (see
Fig.~\ref{fig:rayintersectioncase1}).
Similarly as before, the line segment $\overline{p_1z}$ divides the current
active region of $\bay$ into two simple polygons;
one of them, say $bay'$, contains the point $p$.
Then, the Voronoi
regions of the roots that define the rays in $S$ form a decomposition of
$bay'$, and we use a procedure similar to the splitting procedure
discussed earlier to compute this decomposition of $bay'$, i.e., by
considering the rays in $S$
from bottom to top. However, it is possible that the target points
of some rays in $S$ have not been computed yet. Recall that all target
points of the rays in $S$ before $p$ ($=tp(\rho_i)$) have
been computed. But, if the target point of a ray in $S$ is on
$\partial(p,z)$, then it is not yet known.
To compute these target points, we simply scan
$\partial(p,z)$ from $p$ to $z$. Again, by Invariant (9), the
vertical rays in $S$ are target-sorted. Hence this computation can be done
in linear time in terms of the number of edges of $\partial(p,z)$
and the number of target points found during this process. In addition,
we set the region $bay'$ and all rays in $S$ as inactive, and pop all rays
out of $S$ ($S$ becomes empty). Finally, we move the reference point
$p^*$ to $z$, and continue with the next ray $\rho_{i+1}$.
Again, since $S$ is empty, similar to the analysis for the subcase 2(a)
when $S$ is empty, all invariants of the algorithm hold. We omit the details
of the proof.

If $B_M'(r_j,r_{i+1})$ ($=\overline{p_1p_1'}$) does not intersect $\partial$ (see
Fig.~\ref{fig:termination}), then as shown
above, $B_M'(r_j,r_{i+1})$ appears entirely in $\SPM(\calF)$ since it is contained
inside $\bay$. Again, the point $p_1'$ is
on either the right edge or the bottom edge of $Rec(p_1,q_1)$ (two cases).
We discuss these two cases below. Recall that the trapezoid of
$\HM(\bay)$ that contains $p_1'$ has been computed.
Recall that the ray at the top of $S$
is $\rho$, lying on $B(r_j,r_i)$ with $i>j$.

We first discuss the case when $p_1'$ is on the bottom edge of $Rec(p_1,q_1)$
(see Fig.~\ref{fig:termination}). Let $\rho_i^*$ be
the vertical ray originating at $p_1'$ and going south, which is on
$B(r_j,r_{i+1})$ by Observation \ref{obser:20}.
We pop $\rho$ out of $S$ and push $\rho_i^*$ onto the top of $S$,
and set $\rho$ as inactive and $\rho_i^*$ as active.
We move $p^*$ to $p$ ($=tp(\rho_i)$). We then continue to
consider the next ray $\rho_{i+1}\in \Psi$. Lemma \ref{lem:invariants40}
below shows that all invariants of the algorithm hold.

\begin{figure}[t]
\begin{minipage}[t]{\linewidth}
\begin{center}
\includegraphics[totalheight=2.0in]{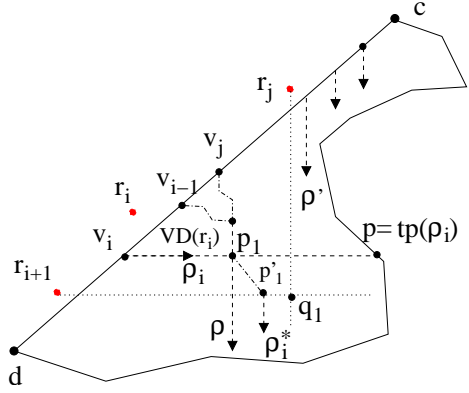}
\caption{\footnotesize Illustrating an example that
the point $p_1'$ ($=or(\rho_i^*)$) is on the bottom edge of
$Rec(p_1,q_1)$.}
\label{fig:termination}
\end{center}
\end{minipage}
\vspace*{-0.15in}
\end{figure}

\lemmaspace
\begin{lemma}\label{lem:invariants40}
At the moment right before the next ray $\rho_{i+1}$ is considered,
all invariants of the algorithm hold.
\end{lemma}
\lemmaspace
\begin{proof}
Let $\xi$ be the moment right before the next ray $\rho_{i+1}$ is
considered, and $\xi'$ be the moment right before the ray $\rho_i$ is
considered. Thus, the change to $S$ from the time $\xi'$ to $\xi$ is that $\rho$
is popped out and $\rho_i^*$ is pushed in. Recall that at the moment $\xi'$, all
invariants of the algorithm hold. Our goal is to prove that all
invariants still hold at the moment $\xi$. We assume that $S$ has
at least two rays at the moment $\xi$ (otherwise,
all invariants related to any other rays in $S$ than $\rho^*_i$
hold trivially). Let $\rho'$ be the second ray from the top of $S$ (i.e., right below
$\rho_i^*$ in $S$) at the moment $\xi$. Then $\rho'$ is also the second
ray from the top of $S$ at the moment $\xi'$ (i.e., right below
$\rho$ in $S$). See Fig.~\ref{fig:termination} for
an example.

Invariant (1) simply follows since $\rho_i^*$ is vertically going
south.

For Invariant (2), it suffices to show that $or(\rho^*_i)$
is to the southwest of $or(\rho')$. At the moment $\xi'$, since
the ray $\rho$ at the top of $S$ is on $B(r_j,r_i)$ with $i>j$,
by Invariant (6), $r_j$ is to the left of $or(\rho')$. Since
$or(\rho^*_i)\in Rec(r_j,r_{i+1})$ is to the left of $r_j$, we obtain
that $or(\rho_i^*)$ is to the left of $or(\rho')$.  Further, at the moment $\xi'$,
by Invariant (2), $or(\rho)$
is to the southwest of $or(\rho')$. Since
$or(\rho_i^*)$ is below $or(\rho)$, $or(\rho_i^*)$ is below $or(\rho')$. Since
$or(\rho_i^*)$ is both below and to the left of $or(\rho')$, we obtain
that $or(\rho_i^*)$ is to the southwest of $or(\rho')$. Invariant (2) thus
follows.

For Invariant (3), we need to show
that $or(\rho_{i+1})$ is to the southwest of $or(\rho_i^*)$.
By Lemma \ref{lem:90}, $or(\rho_{i+1})$ is to the southwest of $or(\rho_i)$.
Since $or(\rho_i^*)$ is to the right of $or(\rho_i)$,
$or(\rho_i^*)$ is also to the right of $or(\rho_{i+1})$. By Lemma
\ref{lem:1000}(3), $or(\rho_{i+1})$ is below $r_{i+1}$. Hence,
$or(\rho_{i+1})$ is below the rectangle $Rec(p_1,q_1)$ and thus below
$or(\rho_i^*)$. Since $or(\rho_{i+1})$ is both below and to the left of
$or(\rho_i^*)$, we obtain that $or(\rho_{i+1})$ is  to the southwest of
$or(\rho_i^*)$. Thus, Invariant (3) follows.

Invariant (4) simply
follows since $\rho_i^*\subset B(r_j,r_{i+1})$ and
$\rho_{i+1}\subset B(r_{i+1},r_{i+2})$.

For Invariant (5), we need to consider both $\rho_i^*$ and
$\rho_{i+1}$. For $\rho_i^*$, since $\rho_i^*\subset B(r_j,r_{i+1})$,
we need to show that the boundary portion of the Voronoi region
$\Vor(r_{i+1})$ (resp., $\Vor(r_j)$) from $v_{i}$ (resp., $v_j$) to $or(\rho_i^*)$
has been computed. For this, recall
that the boundary portion of $\Vor(r_{i+1})$ between $v_i$ and $p_1$ is a
common boundary of $\Vor(r_{i+1})$ and $\Vor(r_i)$, which has been computed. We
denote this boundary portion by $\alpha$. Also, the boundary
portion of $\Vor(r_{j})$ between $v_j$ and $p_1$ has been computed;
we denote this boundary portion by $\beta$.
Further, after we find the point $p'_1$, the line
segment $\overline{p_1p_1'}$ has also been obtained. Since
$\overline{p_1p_1'}$ appears entirely in $\SPM(\calF)$, the
boundary portion of $\Vor(r_{i+1})$ between $v_i$ and
$or(\rho_i^*)$ ($=p_1'$) is the concatenation of $\alpha$ and
$\overline{p_1p_1'}$, which has been computed. Similarly, the boundary portion of
$\Vor(r_{j})$ between $v_j$ and
$or(\rho_i^*)$ is the concatenation of $\beta$ and $\overline{p_1p_1'}$,
which has been computed too.  Thus, the case for $\rho_i^*$ holds.

For the ray $\rho_{i+1}$, which is the ray to be considered next by
the algorithm, we need to show that the boundary portion of the Voronoi
region $\Vor(r_{i+1})$ from $v_{i+1}$ to $or(\rho_{i+1})$, which is
also the boundary portion of the Voronoi region $\Vor(r_{i+2})$ from $v_{i+1}$
to $or(\rho_{i+1})$, has been computed. This simply follows from
Observation \ref{obser:afterPre}.

In summary, Invariant (5) holds.

For Invariant (6), it suffices to show
that $or(\rho_i^*)$ is to the left of $r_j$ and to the right of $r_{i+1}$.
Recall that $or(\rho_i^*)$ is on the rectangle $Rec(p_1,q_1)$, $q_1$ is to
the southeast of $p_1$, and $q_1$ is the intersection of the
vertical line through $r_j$ and the horizontal line through $r_{i+1}$.
As shown above, $r_{i+1}$ is to the left of $p_1$.
Since $or(\rho_i^*)$ is to the right of $p_1$, $or(\rho_i^*)$ is to
the right of $r_{i+1}$. Since $Rec(p_1,q_1)$ is to the
left of $r_{j}$, $or(\rho_i^*)$ is to the left of $r_j$.
Invariant (6) thus holds.

For Invariant (7), we need to show that $r_j$ is to the left of
all rays in $S\setminus \{\rho_i^*\}$. At the moment $\xi'$,
the ray $\rho\subset B(r_j,r_i)$ is at the top of $S$ with $i>j$; thus
by Invariant (7), $r_j$ is to the left of all rays in $S\setminus
\{\rho\}$. Since $S\setminus \{\rho\}=S\setminus \{\rho_i^*\}$,
Invariant (7) still holds at the moment $\xi$.

For Invariant (8), recall that $\rho'$ is the second ray from the top of
$S$ at both the moments $\xi$ and $\xi'$. We assume $\rho'$ lying on $B(r_t,r_{t'})$ with
$t'>t$. To prove Invariant (8) held at $\xi$, it suffices to
show $t'=j$ since $\rho_i^*\subset B(r_j,r_{i+1})$.
At the moment $\xi'$, since $\rho\subset B(r_j,r_i)$ is the ray at the top
of $S$, by Invariant (8), we have $j=t'$. Thus, Invariant (8) still
holds at the moment $\xi$.

For Invariant (9), it suffices to show that $tp(\rho_i^*)$ is
after $tp(\rho')$ on $\partial$. Intuitively this is true due to the
following facts: There is a path inside $\bay$ from $v_i$ to
$or(\rho_i^*)$ (i.e., the concatenation of $\overline{v_ior(\rho_i)}$,
$\overline{or(\rho_i)p_1}$, and $\overline{p_1p_1'}$),
and both $\rho_i^*$ and
$\rho'$ are vertical, and $\rho_i^*$ is to the left of $\rho'$.
A detailed analysis is given below.

First, it is easy to see that
$tp(\rho_i^*)$ must be after the point $p$ ($=tp(\rho_i)$). The
target point $tp(\rho')$ may be after $p$ or before $p$. If
$tp(\rho')$ is before $p$, then we are done. Thus, we consider the case of
$tp(\rho')$ being after $p$.
By Invariant (2) (at the moment $\xi$) proved above, $or(\rho')$
is to the northeast of $or(\rho^*_i)$.  Thus, the ray $\rho'$ must
cross $\rho_i$ before it hits $\partial$ at $tp(\rho')$; in other
words, the two line segments $\overline{or(\rho')tp(\rho')}$ and
$\overline{or(\rho_i)tp(\rho_i)}$ intersect inside $\bay$.  Further, since
$\rho'$ is to the right of $\rho^*_i$ and $p_1$ is to the left of
$\rho_i^*$, the intersection point of $\overline{or(\rho')tp(\rho')}$ and
$\overline{or(\rho_i)tp(\rho_i)}$ is on $\overline{p_1tp(\rho_i)}$.

Recall that
$\overline{v_ior(\rho_i)}$ is either a single point or a line segment
that is in $\Vor(\bay)$ and does not intersect
$\partial$. Consider the region in $\bay$ bounded by
$\overline{v_ior(\rho_i)}$, $\overline{or(\rho_i)p}$, and $\partial(p,d)$,
which we denote by $Z$. It is easy to see that $Z$ is a simple polygon. Let
$\alpha$ be the concatenation of $\overline{p_1p'_1}$ and
$\overline{p_1'tp(p_i^*)}$. Note that $\alpha$ is entirely inside $Z$
except that its two endpoints are on the boundary of $Z$, i.e.,
$p_1\in \overline{or(\rho_i)p}$ and $tp(p_i^*)\in \partial(p,d)$.
Thus, $\alpha$ divides $Z$ into two simple polygons; one of them
contains $\overline{p_1p}$ as an edge, which is denoted by $Z'$.
Since the intersection of $\overline{or(\rho')tp(\rho')}$ and
$\overline{or(\rho_i)tp(\rho_i)}$ is on $\overline{p_1tp(\rho_i)}$,
the ray $\rho'$ intersects $Z'$.  By Invariant (7)
(at the moment $\xi$) proved above, the root $r_j$ is to the left of $\rho'$.
Thus, $\rho'$ cannot intersect the curve $\alpha$. Hence,
the target point $tp(\rho')$ must be on the boundary of $Z'\cap \partial(p,d)$,
which is on $\partial(p,tp(\rho_i^*))$. Thus,
$tp(\rho_i^*)$ is after $tp(\rho')$, and Invariant (9) follows.

For Invariant (10), we need to show that if $\rho_{i+1}$ is
vertical, then the target point $tp(\rho_{i+1})$ is after $tp(\rho_i^*)$
on $\partial$.  By Invariant (3)
(at the moment $\xi$) proved above, $or(\rho_{i+1})$ is to the southwest of
$or(\rho_i^*)$. Let $Z$ be the simple polygonal region in $\bay$ bounded
by $\overline{v_ior(\rho_i)}$, $\overline{or(\rho_i)p_1}$,
$\overline{p_1p_1'}$, $\overline{p'_1tp(\rho_i^*)}$,
$\partial(tp(\rho_i^*),d)$, and $\overline{dv_i}$.
Regardless of whether $or(\rho_{i+1})= v_{i+1}$, the origin
$or(\rho_{i+1})$ of $\rho_{i+1}$ is in $Z$ since $or(\rho_{i+1})$
is to the southwest of $or(\rho_i^*)=p_1'$.
Further, since both $\rho^*_i$ and $\rho_{i+1}$ are vertical, $tp(\rho_{i+1})$
must be on $\partial(tp(\rho_i^*),d)$. Invariant (10) thus follows.

For Invariant (11), it suffices to show that $tp(\rho_i^*)$ is
after $p^*$ since $tp(\rho_i^*)$ has not been computed.  Since
$tp(\rho_i^*)$ is after $p=tp(\rho_i)$ ($=p^*$), Invariant (11) simply follows.

For Invariant (12), we need to show that the target point
$tp(\rho_{i+1})$ is after $p^*$ ($=p=tp(\rho_i)$).  Let $Z$ be the simple
polygonal region in $\bay$ bounded by $\overline{v_ior(\rho_i)}$,
$\overline{or(\rho_i)p}$, $\partial(p,d)$, and $\overline{dv_i}$.
Clearly, $or(\rho_{i+1})$ is in $Z$. Further, since $or(\rho_{i+1})$
is to the southwest of $or(\rho_i)$, regardless of whether
$\rho_{i+1}$ is vertical or horizontal, $tp(\rho_{i+1})$ must be on
$\partial(p^*,d)$. Thus,
Invariant (12) holds.

For Invariant (13), suppose $l$ is the smallest index with $l>i$
such that $\rho_l\in \Psi$ and $\rho_l$ is horizontal. We need
to prove that $tp(\rho_l)$ is after $p^*$ ($=p=tp(\rho_i)$).
Consider the simple polygon $Z$ defined above for proving Invariant (12).
Since $\rho_l$ is horizontal, by Lemma \ref{lem:100}, $\rho_l$ is below
$r_{i+1}$. Thus, it is easy to see
that $or(\rho_{l})$ is in $Z$ and $tp(\rho_{l})$ is on
$\partial(p^*,d)$. Hence, $tp(\rho_{l})$ is after $p^*$, and Invariant (13) holds.

For Invariant (14), we need to show that the trapezoid of $\HM(\bay)$
that contains $or(\rho_{i+1})$ is known, which is true by Observation
\ref{obser:afterPre}.

We conclude that all invariants of the algorithm still hold at the moment $\xi$.
\end{proof}

For the purpose of discussing the analysis of the running time of our
algorithm later, we call the ray $\rho_i^*$ the {\em termination
vertical ray} of the (horizontal) ray $\rho_i$.

We have finished the discussion for the case when $p_1'$ is on the
bottom edge of $Rec(p_1,q_1)$.

We then discuss the case when the point
$p_1'$ is on the right edge of $Rec(p_1,q_1)$ (see
Fig.~\ref{fig:successor}). Denote by
$\rho_{i1}$ the horizontal ray originating at $p'_1$ and going
east, which is on $B(r_j,r_{i+1})$ by Observation \ref{obser:20}.
Then, we pop $\rho$ out of $S$
and set $\rho$ as inactive. Also, we set $\rho_{i1}$ as active and move
the reference point $p^*$ to $p$ ($=tp(\rho_i)$).
Finally, we let $\rho_{i1}$ be the
next ray to be considered by the algorithm (note that $\rho_{i1}$ is
not in $\Psi$). Lemma \ref{lem:invariants50} below shows that all
invariants of the algorithm hold. Recall that the trapezoid of
$\HM(\bay)$ that contains $p_1'$ has been computed.

\begin{figure}[t]
\begin{minipage}[t]{\linewidth}
\begin{center}
\includegraphics[totalheight=2.0in]{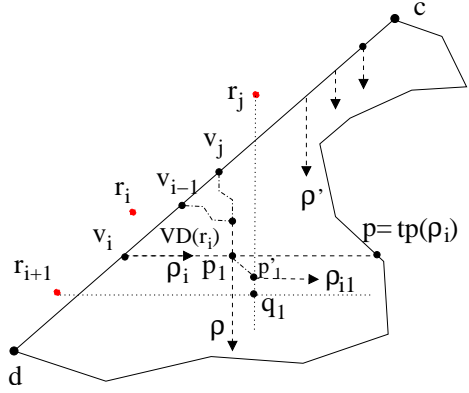}
\caption{\footnotesize Illustrating an example that
the point $p_1'$ ($=or(\rho_{i1})$) is on the right edge of
$Rec(p_1,q_1)$.}
\label{fig:successor}
\end{center}
\end{minipage}
\vspace*{-0.15in}
\end{figure}

\lemmaspace
\begin{lemma}\label{lem:invariants50}
At the moment right before the next ray $\rho_{i1}$ is considered,
all invariants of the algorithm hold.
\end{lemma}
\lemmaspace
\begin{proof}
Let $\xi$ be the moment right before the next ray $\rho_{i1}$ is
considered, and $\xi'$ be the moment right before the ray $\rho_i$ is
considered. Thus, the only change to $S$ from the
time $\xi'$ to $\xi$ is that $\rho$
is popped out. At the moment $\xi'$, all
invariants of the algorithm hold. Our goal is to prove that all
invariants still hold at the moment $\xi$.
We assume $S\not=\emptyset$ at the moment $\xi$ (otherwise,
all invariants related to any rays in $S$ hold trivially).
Let $\rho'$ be the ray at the top of $S$
at the moment $\xi$. Then $\rho'$ is the second ray from the top
of $S$ (i.e., right below the ray $\rho$ in $S$)
at the moment $\xi'$. Refer to Fig.~\ref{fig:successor} for an example.

Invariants (1) and (2) simply hold.

For Invariant (3), we need to show that $or(\rho_{i1})$ is to the
southwest of $or(\rho')$.
By Invariant (2) at the moment $\xi'$, $or(\rho)$ of the ray $\rho$ at the top of $S$
is to the southwest of $or(\rho')$.  Since $or(\rho_{i1})=p_1'$ is below $or(\rho)$,
$or(\rho_{i1})$ is below $or(\rho')$. Also, by Invariant (7) at the moment $\xi'$,
since the top ray $\rho$ in $S$ is on $B(r_j,r_i)$ (with $i>j$), $r_j$ is to the left of
$\rho'$ (which is vertical).
Since $p_1'$ is on the vertical line through $r_j$, $p_1'=or(\rho_{i1})$ is
to the left of $or(\rho')$. Hence, $or(\rho_{i1})$ is to the
southwest of $or(\rho')$, and Invariant (3) follows.

For Invariant (4), suppose $\rho'$ is on $B(r_{t},r_{t'})$ with $t'>t$;
we need to show $j=t'$ since $\rho_{i1}\subset B(r_j,r_{i+1})$ is
the next ray to be considered by the algorithm.
At the moment $\xi'$, $\rho\subset B(r_j,r_i)$ (with $i>j$) is at the top
of $S$ and $\rho'$ is the second ray from the top of $S$; thus,
by Invariant (8) at the moment $\xi'$, $j=t'$.
Invariant (4) hence follows.

For Invariant (5), since no new ray is pushed onto $S$, we only
need to consider the ray $\rho_{i1}$.
The proof is the same as that for Invariant (5) (for the ray $\rho_i^*$) in
the proof of Lemma \ref{lem:invariants40}, and we omit it.

Invariants (6), (7), (8), (9), and (11) trivially hold since
no new ray is pushed into $S$.

Invariant (10) simply follows since $\rho_{i1}$ is the next ray to
be considered by the algorithm and $\rho_{i1}$ is not vertical.

For Invariant (12), we need to show that the target point
$tp(\rho_{i1})$ is after $p^*$ ($=p=tp(\rho_i)$).
Consider the simple polygonal region $Z$ in $\bay$ bounded by $\overline{v_ior(\rho_i)}$,
$\overline{or(\rho_i)p}$, $\partial(p,d)$, and $\overline{dv_i}$.
It is easy to see that
$or(\rho_{i1})$ is in $Z$ and $tp(\rho_{i1})$ is on
$\partial(p^*,d)$. Thus, $tp(\rho_{i1})$ is after $p^*$.

For Invariant (13), suppose $l$ is the smallest index with $l>i$
such that $\rho_l\in \Psi$ and $\rho_l$ is horizontal. We need
to prove that $tp(\rho_l)$ is after $p^*$ ($=p=tp(\rho_i)$).
Consider the simple polygon $Z$ defined above for proving Invariant (12).
Since $\rho_l$ is horizontal, by Lemma \ref{lem:100}, $\rho_l$ is below
$r_{i+1}$. Thus, it is easy to see that
$or(\rho_{l})$ is in $Z$ and $tp(\rho_{l})$ is on
$\partial(p^*,d)$. Hence, $tp(\rho_{l})$ is after $p^*$,
and Invariant (13) holds.

For Invariant (14), recall that the trapezoid of $\HM(\bay)$
that contains $p_1'$ ($=or(\rho_{i1})$) has been computed, and thus
Invariant (14) holds.

We conclude that all invariants of the algorithm hold at the moment $\xi$.
\end{proof}

%

For analysis, we refer to the ray $\rho_{i1}$ as a {\em successor horizontal ray}
of the (horizontal) ray $\rho_i$.

This finished the discussion for the case when $p_1'$ is on the
right edge of $Rec(p_1,q_1)$.

Again, $\rho_{i1}$ is the next ray to be considered by the algorithm.
Although our earlier discussion on the algorithm processing
the next ray is mostly on processing a ray $\rho_i\in \Psi$,
the processing for $\rho_{i1}$ ($\not\in \Psi$) is the same, and
the proof for all invariants is also very similar. In particular, there may
also be a termination vertical ray or a successor horizontal ray
generated at the end of processing $\rho_{i1}$,
which we still refer to as a termination vertical ray or a
successor horizontal ray of $\rho_i$. It is easy to see that
a horizontal ray $\rho_i$ may lead to multiple successor horizontal rays but at
most one termination vertical ray, i.e.,
a successor horizontal ray may generate another successor horizontal ray
(e.g., see Fig.~\ref{fig:sucrays}),
but a termination vertical ray does not generate another ray.

One might be curious about why the roles of horizontal rays and vertical rays
are quite different in our above algorithm, while the $L_1$ metric does not prefer
one of these two directions over the other.  The asymmetric roles of these
two directions are related to the order of $\rho_1,\rho_2,\ldots,\rho_{k-1}$ in which
we process these rays.  If one uses a reversed order
(i.e., $\rho_{k-1},\rho_{k-2},\ldots,\rho_1$) in the processing, then
the roles of these two types of rays will be reversed.

For the purpose of analyzing the running time of the algorithm later,
we discuss more details related to the successor horizontal rays of a horizontal
ray $\rho_i\in \Psi$.
We process the first successor horizontal ray $\rho_{i1}$ of $\rho_i$
in the same way as $\rho_i$.
After $\rho_{i1}$ is processed, we may obtain another
successor horizontal ray $\rho_{i2}$. In general, assume all
successor horizontal rays we obtain for $\rho_i$ are
$\rho_{i1},\rho_{i2},\ldots,\rho_{it}$, ordered by the time when they are
produced (see Fig.~\ref{fig:sucrays}).
Then, after the last ray $\rho_{it}$ is processed, we may
or may not obtain the termination vertical ray $\rho_i^*$. For example,
when processing $\rho_{it}$, if $S=\emptyset$, then no
termination vertical ray is generated. In either case, after
$\rho_{it}$ is processed, we continue to consider the next ray
$\rho_{i+1}\in \Psi$.

Let $\rho_{i0}=\rho_i$. For each $1\leq w\leq t$, we define the
points $p_{w+1}$, $q_{w+1}$, and $p'_{w+1}$ for the ray $\rho_{iw}$
similarly to the points $p_{1}$, $q_1$, and $p'_{1}$ for
$\rho_{i0}$ (see Fig.~\ref{fig:sucrays}).  Note that when processing
$\rho_{it}$, depending on the specific situations,
the points $p_{t+1}$, $q_{t+1}$, and $p'_{t+1}$ may
not exist (e.g., if $S=\emptyset$). In the following, we assume
they exist (otherwise, the analysis is actually simpler).

\begin{figure}[t]
\begin{minipage}[t]{\linewidth}
\begin{center}
\includegraphics[totalheight=2.5in]{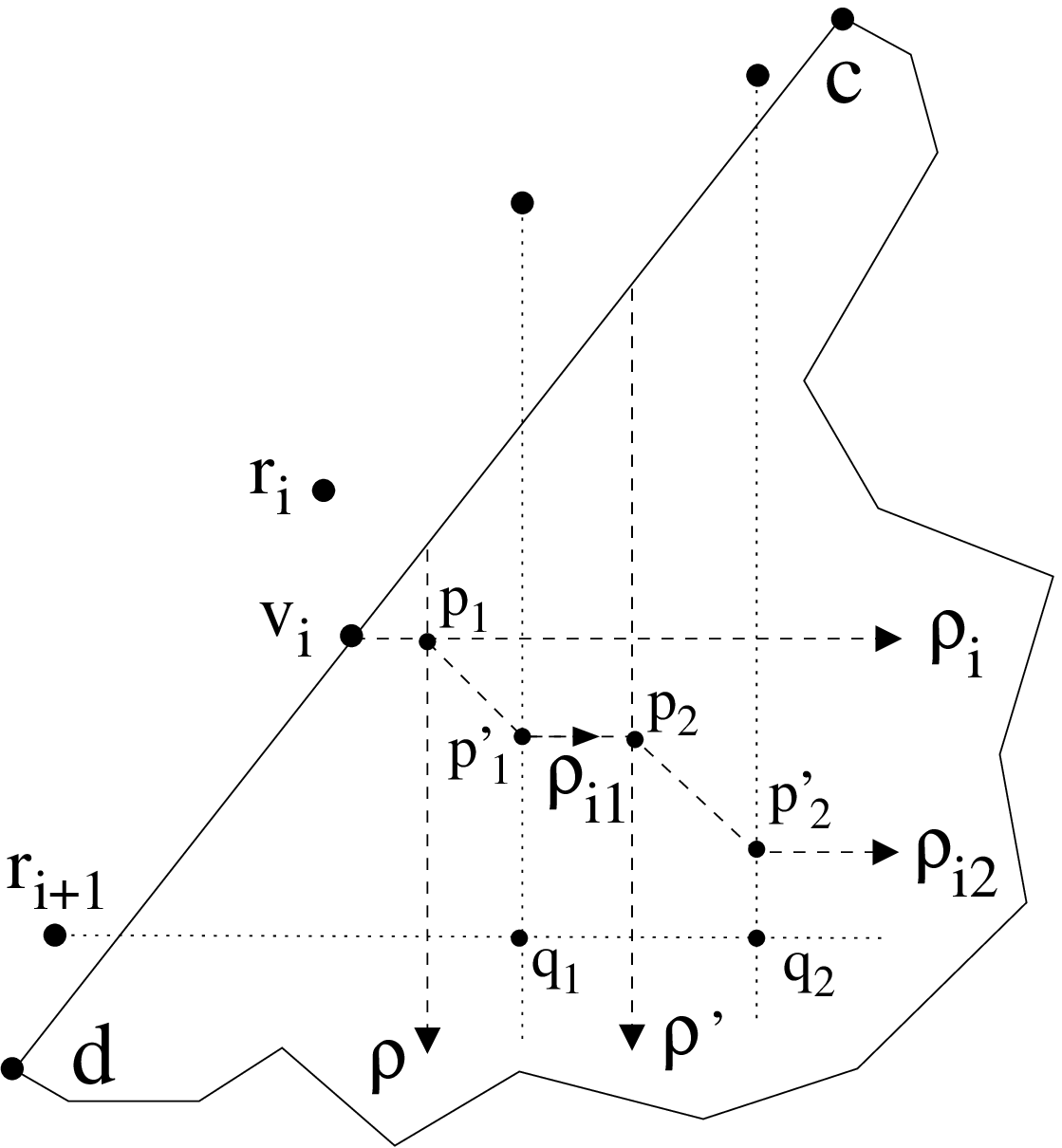}
\caption{\footnotesize Illustrating the first two successor horizontal rays
$\rho_{i1}$ and $\rho_{i2}$ of a horizontal ray $\rho_i\in \Psi$.}
\label{fig:sucrays}
\end{center}
\end{minipage}
\vspace*{-0.15in}
\end{figure}

It is easy to see that for each $1\leq w\leq t$, the ray
$\rho_{i,w-1}$ contains the top edge of the rectangle
$Rec(p_w,p_w')$ and the ray $\rho_{iw}$ touches the bottom edge of
$Rec(p_w,p_w')$. In addition, the ray $\rho_{it}$ contains the top
edge of $Rec(p_{t+1},p_{t+1}')$. In other words, $\rho_{i0}$ ($=\rho_i$),
$Rec(p_1,p_1')$, $\rho_{i1}$, $Rec(p_2,p_2')$, $\rho_{i2}$,
$\ldots$, $Rec(p_t,p_t')$, $\rho_{it}, Rec(p_{t+1},p_{t+1}')$ are
ordered from high to low and left to right (see Fig.~\ref{fig:sucrays}).
Thus, no two different rectangles in the sequence
above intersect in their interior. Actually, the rectangles
$Rec(p_1,p_1'), Rec(p_2,p_2'), \ldots, Rec(p_{t+1},p_{t+1}')$ are
ordered from northwest to southeast. Further, all successor
horizontal rays and rectangles involved are higher than
$r_{i+1}$. To see this fact, note that for each $1\leq w\leq t+1$, the
point $p_w'$ is higher than the point $q_w$ and $q_w$ is on the
horizontal line through $r_{i+1}$ (see Fig.~\ref{fig:sucrays}).
Thus, all these rectangles are
contained in the horizontal strip between the horizontal line
containing $\rho_i$ and the horizontal line through $r_{i+1}$;
we denote this strip by $\HStrip(\rho_i)$. Recall that during our
algorithm, the horizontal visibility map $\HM(\bay)$ is utilized as a
guide and we often move among its trapezoids.
When computing and processing these successor horizontal rays, we always
follow $\HM(\bay)$, e.g., for each $w=0,1,\ldots,t$,
we utilize $\HM(\bay)$ to compute the target
point $tp(\rho_{iw})$ of the ray $\rho_{iw}$, to determine whether
$\overline{p_{w+1}p'_{w+1}}$ intersects $\partial$, and to find the trapezoid
in $\HM(\bay)$ that contains $p'_{w+1}=or(\rho_{i,w+1})$. The discussion above
implies that the time for processing all successor horizontal rays of
$\rho_i$ is proportional to $O(t)$ plus the number of trapezoids in
$\HM(\bay)$ that intersect the horizontal strip $\HStrip(\rho_i)$ as
well as the time for computing the target points of some (vertical) rays in
$S$.

In addition, during this process, each of the $t$ successor horizontal rays
$\rho_{iw}$ of $\rho_i$ corresponds to a ray in $S$ that is popped out.
Thus, there are $t$ vertical rays popped out of $S$ for $\rho_i$.
But, at most one ray, i.e., the termination vertical ray
$\rho_i^*$, is pushed onto $S$ for $\rho_i$.

We have finished the description of our algorithm for computing
$\SPM(\bay)$, which is
summarized by the pseudo-code of Algorithm \ref{algo:10}.

\subsubsection{The Time Complexity}

It remains to analyze the running time of the algorithm. First,
we show the following lemma.

\lemmaspace
\begin{lemma}\label{lem:numrays}
The total number of rays ever contained in the stack $S$
throughout the entire algorithm is at most $k$. Once a ray is popped out
of $S$, it will never be pushed back in again.
\end{lemma}
\lemmaspace
\begin{proof}
When processing each ray $\rho_i\in \Psi$, if it is vertical, then we push it
onto $S$; if it is horizontal, then as shown above, although
there may be multiple successor horizontal rays of $\rho_i$, at most one
ray, i.e., the termination vertical ray, is put into $S$. Further,
according to our algorithm, once a ray in $S$ is popped out, it will
never be considered again, and thus never be put into $S$ again.
\end{proof}

We then discuss the total time for computing the target
points for all vertical ray shootings in the entire algorithm.
We use a reference point $p^*$ on $\partial$ and the
vertical visibility map $\VM(\bay)$ for this purpose. To conduct the
vertical ray shootings, because the rays involved are always target-sorted,
we simply scan the edges in a portion of $\partial$ between $p^*$
and another point $p$ that is after $p^*$ on $\partial$. Further,
when such a scanning is done, we always move $p^*$ to $p$. This
implies that any portion of $\partial$ is scanned at most once in
the entire algorithm. In addition, the number of all vertical ray
shootings is at most $k$. This is because each vertical ray involved is from
$S$, and by Lemma \ref{lem:numrays}, the number of rays ever contained in
$S$ is at most $k$. Therefore, the total time for computing the
target points of all vertical rays in the entire algorithm is $O(n'+k)$.

For each ray $\rho_i\in \Psi$, if it is vertical, then processing it takes
$O(1)$ time, i.e., pushing $\rho_i$ onto
$S$. If it is horizontal, then assume that $\rho_i$  has $t$ successor
horizontal rays. We have discussed that, besides the procedure for
computing their target points, the time for processing these $t$ successor
horizontal rays is proportional to $t$ plus the number of
trapezoids in $\HM(\bay)$ intersecting the horizontal strip
$\HStrip(\rho_i)$. We have also shown that each successor
horizontal ray corresponds to a ray in the stack $S$ that is popped
out. Since there are at most $k$ rays ever contained in $S$ by Lemma \ref{lem:numrays},
the total number of successor horizontal rays in the entire algorithm is
at most $k$. On the other hand, consider two different horizontal
rays $\rho_i$ and $\rho_j$ in $\Psi$. We claim that the two
horizontal strips $HStip(\rho_i)$ and $\HStrip(\rho_j)$ do not
intersect each other in their interior. WLOG, assume $i<j$.
Indeed, the strip $\HStrip(\rho_i)$ is above the horizontal line
through the root $r_{i+1}$ and $\HStrip(\rho_j)$ is below the ray
$\rho_j$. Since $\rho_j$  is horizontal and $j>i$, by Lemma
\ref{lem:100}, $\rho_j$ is below $r_{i+1}$. Our claim thus holds.
The above claim implies that, besides the time for computing
their target points, the time for processing all successor
horizontal rays in the entire algorithm is proportional to
the total number of trapezoids in $\HM(\bay)$ plus $k$, which is $O(n'+k)$.

The algorithm performs totally $O(k)$ horizontal ray shootings, for computing
the target points of the horizontal rays in $\Psi$ and their successor
horizontal rays.  Using $\HM(\bay)$ and based on the fact that we already know
(i.e., have computed) the trapezoid of $\HM(\bay)$ containing the origin of each
such horizontal ray, all such horizontal ray shootings can be done in $O(k)$ time.


In summary, the total running time of our algorithm for computing the
shortest path map for the bay $\bay$ is $O(n'+m')$ (where $m'=k-1$
is the number of $\SPM(\calM)$ vertices on $\overline{cd}$).
It is easy to see that the size of this SPM is $O(n'+m')$
(e.g., since the running time is $O(n'+m')$).

Theorem \ref{theo:baytime} thus follows.

\sectionspace
\section{Computing a Shortest Path Map for a Canal}
\label{sec:canal}

In this section, we show how to compute a shortest path map for
a canal, which uses our shortest path map algorithm for a bay
in Section \ref{sec:bay} as a main procedure.

Consider a canal $\canal$ with $x$ and $y$ as the corridor path
terminals and two gates $\overline{xd}$ and
$\overline{yz}$ (e.g., see Fig.~\ref{fig:baycanal}).  There may
be $\SPM(\calM)$ vertices on both gates. Let $m_1$
(resp., $m_2$) be the number of $\SPM(\calM)$ vertices on
$\overline{xd}$ (resp., $\overline{yz}$), and $n'$ be the number of
obstacle vertices of the canal. We
show that a shortest path map for the canal can be computed in
$O(m_1+m_2+n')$ time. Let $R_1$ (resp., $R_2$)
be the set of roots whose cells in $\SPM(\calM)$ intersect
$\overline{xd}$ (resp., $\overline{yz}$).

Recall that we have defined {\em wavefront incoming/outgoing} terminals in Section \ref{subsec:single}. Namely, consider the corridor path terminals $x$ and $y$ of $\canal$. It is possible that $y$ has a shortest path from
$s$ via $x$ (i.e., this path contains the corridor path of $\canal$), 
in which case there is
a ``pseudo-cell" in $\SPM(\calM)$ with $x$ as the root and $y$ being the only other point in this ``pseudo-cell"; then $x$ is a {\em wavefront incoming}
terminal and $y$ is the {\em wavefront-outgoing} terminal. If neither $y$
has a shortest path from $s$ via $x$ nor $x$ has a shortest path
from $s$ via $y$, then both $x$ and $y$ are {\em wavefront-incoming}
terminals. In this case, there is a point on the corridor path of
$\canal$ that has two shortest paths from $s$, one via $x$ and the
other via $y$ (we will use this property to compute an SPM for
$\canal$).

Note that for the two terminals $x$ and $y$, either both of them are
wavefront-incoming terminals, or only one of them is an wavefront-incoming
terminal and the other is an wavefront-outgoing terminal. Below, we first
discuss the former case; the algorithm for the latter case is very
similar.


\subsection{Both $x$ and $y$ are Wavefront-Incoming Terminals}

If both $x$ and $y$ are wavefront-incoming terminals, by the
properties of the corridor path, there is a point $p^*$ on the
corridor path of $\canal$ such that there exist two
shortest paths $\pi_1(s,p^*)$ and $\pi_2(s,p^*)$ from $s$ to $p^*$ with
$x\in \pi_1(s,p^*)$ and $y\in \pi_2(s,p^*)$.
The point $p^*$ can be found in $O(n')$ time
since we know the shortest path distances from $s$ to $x$ and to $y$.

Let $\Vor(\canal,R_1)$ be the (additively) weighted
Voronoi diagram of $\canal$ with respect to the root set $R_1$, i.e.,
we treat $\canal$ as a bay with the gate $\overline{xd}$. As
defined in Section \ref{sec:bay},
$\Vor(\canal,R_1)$ is the Voronoi decomposition of $\canal$ with
respect to the roots in $R_1$.  Similarly, let $\Vor(\canal,R_2)$ be
the weighted Voronoi diagram of $\canal$ with respect to the root set
$R_2$. Using our algorithm in Section \ref{sec:bay},
$\Vor(\canal,R_1)$ and $\Vor(\canal,R_2)$ can be computed in totally
$O(m_1+m_2+n')$ time. Denote by $\Vor(\canal,R_1,R_2)$ the weighted
Voronoi diagram of $\canal$ with respect to the roots in $R_1\cup R_2$.
As shown in Section \ref{sec:bay}, after $\Vor(\canal,R_1,R_2)$ is
computed, an SPM on $\canal$ with the source $s$ can be built in $O(m_1+m_2+n')$ time.
Thus, the key is to compute  $\Vor(\canal,R_1,R_2)$. Below, we show
how to compute $\Vor(\canal,R_1,R_2)$ in $O(m_1+m_2+n')$ time with
the help of the point $p^*$, $\Vor(\canal,R_1)$, and
$\Vor(\canal,R_2)$.

To compute $\Vor(\canal,R_1,R_2)$, our strategy is to find a
``dividing curve" in $\canal$ that divides $\canal$ into two simple
polygons $C_1$ and $C_2$, such that each point in $C_1$ has a
shortest path from $s$ via a root in $R_1$ and each point in $C_2$
has a shortest path from $s$ via a root in $R_2$. Further, each
point on the dividing curve has two shortest paths from $s$, one
path containing a root in $R_1$ and the other path containing a root
in $R_2$. After finding $C_1$ and $C_2$, we simply apply the
algorithm in Section \ref{sec:bay} on $C_1$ and $R_1$ to compute the
weighted Voronoi diagram of $C_1$ with respect to $R_1$, i.e.,
$\Vor(C_1,R_1)$. We similarly compute $\Vor(C_2,R_2)$. Then,
$\Vor(\canal,R_1,R_2)$ consists of $\Vor(C_1,R_1)$ and
$\Vor(C_2,R_2)$. Thus, our remaining task is to compute a dividing
curve in $\canal$, which we denote by $\PC$.

Note that the point $p^*\in \PC$. Computing $\PC$ can be done in
$O(n'+m_1+m_2)$ time by a procedure similar to the merge procedure
of the divide-and-conquer algorithm for computing the Voronoi
diagram of a set of points in the plane \cite{ref:ShamosCl75}. The details are
given below.

 To compute $\PC$, we start at the
point $p^*$ and trace $\PC$ out by traversing some corresponding
cells in $\Vor(\canal,R_1)$ and in $\Vor(\canal,R_2)$
simultaneously. Specifically, we first compute a triangulation of
$\Vor(\canal,R_1)$, denoted by $Tri_1$, and a triangulation of
$\Vor(\canal,R_2)$, denoted by $Tri_2$ (this can be done in linear
time \cite{ref:ChazelleTr91} since each cell of $\Vor(\canal,R_1)$
and $\Vor(\canal,R_2)$ is a simple polygon).
Since $p^*$ is in a triangle (say, $tri_1$) of $Tri_1$ and is in
a triangle (say, $tri_2$) of $Tri_2$, we find $tri_1$ in $Tri_1$ and
$tri_2$ in $Tri_2$.  From the cell
of $\Vor(\canal,R_1)$ (resp., $\Vor(\canal,R_2)$) that contains
$tri_1$ (resp., $tri_2$), we obtain the root $r_1$ (resp., $r_2$) of
that cell. We then move along the bisector $B(r_1,r_2)$ inside
$\canal$, starting at $p^*$ and going in each of the two directions
along $B(r_1,r_2)$. As following a line segment or a ray of
$B(r_1,r_2)$ in a direction, we determine, in $O(1)$ time, which of
$tri_1$ or $tri_2$ that we exit first. As we cross from one triangle
$tri$ (say, in $Tri_1$) to the next triangle $tri'$, we check which
of the following cases occurs: (i) The next triangle $tri'$ (in
$Tri_1$) is contained in the same cell of $\Vor(\canal,R_1)$ as that
containing $tri$; (ii) $tri'$ is contained in a different cell of
$\Vor(\canal,R_1)$ than that containing $tri$; (iii) the movement
touches the boundary of $\canal$ (thus $tri'$ does not exist).  In
Case (i), we continue to follow the same bisector (say,
$B(r_1,r_2)$). In Case (ii), we find the root (say, $r'_1$) of the
next cell of $\Vor(\canal,R_1)$; then we compute a new bisector
(say, $B(r'_1,r_2)$), and our movement continues along
$B(r'_1,r_2)$. In Case (iii), the movement reaches an end of $\PC$
(on the boundary of $\canal$). The dividing curve $\PC$ is the
concatenation of the portions of the bisectors thus traversed.

Due to the properties of the cells of $\Vor(\canal,R_1)$ and
$\Vor(\canal,R_2)$, our movement above can visit each triangle of $Tri_1$ and $Tri_2$
at most once, taking $O(1)$ time per triangle visited.
Thus, the partition curve $\PC$ is computed in $O(n'+m_1+m_2)$ time.

In summary, in this case, an SPM on $\canal$ can be computed in
$O(n'+m_1+m_2)$ time.

\subsection{Only One of $x$ and $y$ is a Wavefront-Incoming Terminal}

In this case, exactly one of $x$ and $y$ is a wavefront-incoming terminal.  The
algorithm is similar to that for the former case.  The only
difference is on how to find a point $p^*$ on the dividing curve
$\PC$ because in this case no such a point $p^*$ can be on the
corridor path of $\canal$.

WLOG, we assume that $x$ is a wavefront-incoming terminal and $y$ is
not. Then each point on the corridor path (including $y$) has a
shortest path from $s$ via $x$.  Further, the shortest path  
through $x$ passes $y$ and goes to the outside of $\canal$,
which means that $y$ is the root of a cell $C(y)$ in
$\SPM(\calM)$. If the canal gate $\overline{yz}$ is completely
contained in the cell $C(y)$, then it is easy to see that
$\Vor(\canal,R_1)$ is $\Vor(\canal,R_1,R_2)$. Otherwise, as in the
former case, we need to find a dividing curve $\PC$ to divide
$\canal$ into two polygons $C_1$ and $C_2$ such that each point in
$C_1$ has a shortest path from $s$ via a root in $R_1$ and each
point in $C_2$ has a shortest path from $s$ via a root in $R_2$.  To
obtain $\PC$, the key is to find a point $p^* \in \PC$. Since the
canal gate $\overline{yz}$ is not completely contained in $C(y)$,
there must be a point $q$ on $\overline{yz}$ that is on the common
boundary of $C(y)$ and another cell $C(r)$ in $\SPM(\calM)$. We
claim that $q$ is on $\PC$. Indeed, note that $r$ is in $R_2$. Hence
there is a shortest path $\pi_1(s,q)$ from $s$ to $q$ that contains
$x$, the corridor path in $\canal$, and the line segment
$\overline{yq}$, and there is another shortest path $\pi_2(s,q)$
from $s$ to $q$ via the root $r\in R_2$. In other words, $q$ has two
shortest paths from $s$, one via a root in $R_1$ and the other via a
root in $R_2$.
Therefore, $q$ is on the
dividing curve $\PC$. The rest of the algorithm is similar to that
for the former case.

In summary, in this case, an SPM on $\canal$ can also be built in
$O(n'+m_1+m_2)$ time.

Therefore, a shortest path map SPM on $\canal$ can be computed in
$O(n'+m_1+m_2)$ time. Similarly, the size of this SPM is
$O(n'+m_1+m_2)$.

Theorem \ref{theo:canaltime} thus follows.

\sectionspace
\section{Applications of Our Shortest Path Algorithms}
\label{sec:fixed}

In this section, we extend our techniques to solve some other 
problems.

\subsection{The $L_1$ Geodesic Voronoi Diagram}

Given a set $\calP$ of $h$ polygonal obstacles of totally $n$
vertices and a set of $m$ point sites, the \lgvd\ problem aims to 
construct the $L_1$ geodesic Voronoi diagram of for the $m$ point
sites. Denote by $\gvd(\calP)$ the Voronoi diagram that we want to
construct. 

Mitchell's algorithm \cite{ref:MitchellAn89,ref:MitchellL192} can be modified to
compute $\gvd(\calP)$ in $O((n+m)\log(n+m))$ time. Namely,
instead of initiating a wavelet at a single source, the modified
algorithm for $\gvd(\calP)$ initiates a wavelet at
each point site. The rest of the algorithm remains the same as before.

We can also extend our \SPM\ algorithm in a similar way to
compute $\gvd(\calP)$. Generally, since our algorithm makes use of
Mitchell's algorithm \cite{ref:MitchellAn89,ref:MitchellL192} as a
main procedure when computing the shortest path map $\SPM(\calM)$
for the ocean $\calM$, to
compute $\gvd(\calP)$, we can simply replace Mitchell's algorithm by
its modified version for computing $L_1$ geodesic Voronoi diagrams.
More specifically, our
algorithm for computing $\gvd(\calP)$ has the following steps. (1) Compute a
triangulation of the free space, in which the $m$ point sites are
treated as $m$ point obstacles. (2) Compute the corridor structure on $\calP$
and the $m$ point obstacles that consists of
$O(m+h)$ corridors, which partition the plane into a set $\calP'$ of
$O(m+h)$ convex polygons of totally $O(n+m)$ vertices. (3) Compute
the core set $\c(\calP')$ for the convex polygons in $\calP'$. (4)
Apply Mitchell's modified algorithm 
\cite{ref:MitchellAn89,ref:MitchellL192} to compute the $L_1$
geodesic Voronoi diagram $\gvd(\c(\calP'))$ on the core set
$\c(\calP')$. (5) Based on $\gvd(\c(\calP'))$, compute the $L_1$ geodesic
Voronoi diagram $\gvd(\calP')$ on the convex polygon set $\calP'$.
Although we have multiple sources, this step is the same as before
(i.e., as in Lemma \ref{lem:50}). (6) Based on $\gvd(\calP')$, compute the
Voronoi regions in all bays and canals, as in Sections \ref{sec:bay}
and \ref{sec:canal}. Again, the algorithms for this step are as before,
i.e., as the algorithms in Sections \ref{sec:bay} and \ref{sec:canal}.
We then obtain the final $L_1$ geodesic Voronoi diagram $\gvd(\calP)$.

To analyze the running time, Steps (1), 
(2), and (3) are the same as before except that the number of obstacles
becomes $m+h$. Specifically, the triangulation in 
Step (1) takes $O(n+(h+m)\log^{1+\epsilon}(h+m))$ time \cite{ref:Bar-YehudaTr94}. 
Steps (2) and
(3) together take $O(n+(h+m)\log(h+m))$ time. Step (4) takes
$O((m+h)\log(m+h))$ time since the core set $\c(\calP')$ has totally
$O(m+h)$ vertices. Steps (5) and (6) are also the same as before,
which take linear time, i.e., $O(n+m)$. Therefore, the entire
algorithm takes $O(n+(h+m)\log^{1+\epsilon}(h+m))$ time, which is
dominated by the time of the triangulation procedure in Step (1).

\lemmaspace
\begin{theorem}\label{theo:100}
The $L_1$ geodesic Voronoi diagram of $m$ point sites among
a set of $h$ pairwise disjoint polygonal obstacles of totally $n$
vertices in the plane can be computed in
$O(n+(h+m)\log^{1+\epsilon}(h+m))$ time (or $O(n+(h+m)\log(h+m))$
time if a triangulation is given).
\end{theorem}
\lemmaspace

If the $m$ point sites are all inside a simple polygon, then Theorem
\ref{theo:100} leads to the following result.

\lemmaspace
\begin{corollary}\label{corr:30}
The $L_1$ geodesic Voronoi diagram of a set of $m$ point sites in a
simple polygon can be computed in $O(n+m\log^{1+\epsilon} m)$ time
(or $O(n+m\log m)$ time if a triangulation is given).
\end{corollary}
\lemmaspace

Note that the currently fastest known \gvd\ algorithm for the
Euclidean version of the single simple polygon case runs in
$O((n+m)\log (n+m))$ time \cite{ref:PapadopoulouA98}.

{\bf Remark.} Since the given $m$ sites are points, there is an
alternative triangulation algorithm that may be faster (than simply
applying the algorithm in \cite{ref:Bar-YehudaTr94}) in some
situations. The algorithm works as follows: (1) Compute the
triangulation of the free space without considering the $m$ sites;
(2) find the triangles in the triangulation that contain those $m$
sites (e.g., by a point location data structure); (3) triangulate
those triangles that contain at least one point site by considering
the point sites as obstacles. It is easy to see that this algorithm
takes $O(n+m\log n)$ time in the single polygon case and
$O(n+h\log^{1+\epsilon}h+m\log n)$ time in the polygonal domain
case. Therefore, using this triangulation algorithm, the geodesic
Voronoi diagram can be constructed in $O(n+m(\log n+\log m))$ time
in the single polygon case and in $O(n+h\log^{1+\epsilon}h+m\log n +
(h+m)\log(h+m))$ time in the polygonal domain case.

\subsection{Shortest Paths with Fixed Orientations and Approximate
Euclidean Shortest Paths}

As in \cite{ref:MitchellAn89,ref:MitchellL192}, our algorithms can be
generalized to solving the $C$-oriented shortest path problem
\cite{ref:WidmayerOn87}. {\em A $C$-oriented path} is a polygonal
path with each edge parallel to one of a given set $C$ of fixed
orientations. A shortest $C$-oriented path between two points is a
$C$-oriented path with the minimum Euclidean distance. Rectilinear
paths are a special case of this problem with two fixed orientations of $0$ and
$\pi/2$. Let $c=|C|$. Mitchell's algorithm
\cite{ref:MitchellAn89,ref:MitchellL192} can compute a shortest
$C$-oriented path in $O(cn\log n)$ time and $O(cn)$
space among $h$ pairwise disjoint
polygons of totally $n$ vertices in the plane.
Similarly, our algorithms also work for this problem, as follows.

We first consider the convex case (i.e., all polygons are
convex). We compute a core for each convex polygon based on the
orientations in $C$. Note that in this case, a core has $O(c)$ vertices. Thus, we
obtain a core set of totally $O(ch)$ vertices. We then apply
Mitchell's algorithm for the fixed orientations of $C$ on the core
set to compute a shortest path avoiding the cores in $O(c^2h\log
ch)$ time and $O(c^2h)$ space, after which we find a shortest path
avoiding the input polygons in additional $O(n)$ time as in Lemma
\ref{lem:40}. Thus, a shortest path can be found in totally
$O(n+c^2h\log ch)$
time and $O(n+c^2h)$ space. For the general case when the polygons
need not be convex, the algorithm scheme is similar to our
$L_1$ algorithm in Section \ref{sec:general}. In summary, we have the following result.

\lemmaspace
\begin{theorem}\label{theo:70}
Given a set $C$ of orientations and a set of $h$
pairwise disjoint polygonal obstacles of totally $n$ vertices in the plane,
we can compute a
$C$-oriented shortest $s$-$t$ path in the free space in
$O(n+h\log^{1+\epsilon} h+c^2h\log ch)$ time (or $O(n+c^2
h\log ch)$ time if a triangulation is given) and $O(n+c^2h)$
space, where $c=|C|$.
\end{theorem}
\lemmaspace

This also yields an approximation algorithm for computing a
Euclidean shortest path between two points among polygonal obstacles.
Since the Euclidean metric can be approximated within an accuracy
of $O(1/c^2)$ if we use $c$ equally spaced orientations, as in
\cite{ref:MitchellAn89,ref:MitchellL192}, Theorem \ref{theo:70}
leads to an algorithm that computes a path guaranteed to have a length
within a factor $(1+\delta)$ of the Euclidean shortest path length, where
$c$ is chosen such that $\delta=O(1/c^2)$.

\begin{corollary}\label{corr:20}
A $\delta$-optimal Euclidean shortest path between two points among
$h$ pairwise disjoint polygons of totally $n$ vertices in the plane
can be computed in
$O(n+h\log^{1+\epsilon} h+(1/\delta)h\log \frac{h}{\sqrt{\delta}})$ time (or
$O(n+(1/\delta) h\log \frac{h}{\sqrt{\delta}})$ time if a triangulation is given)
and $O(n+(1/\delta)h)$ space.
\end{corollary}

\section{Conclusions}
\label{sec:con}

We present new algorithms for solving $L_1$ shortest path problems
in polygonal domains. Our algorithms are optimal if the triangulation
for the free space can be done optimally (i.e., $T=O(n+h\log h)$).
In fact, our results show that building an $L_1$ shortest path map is equivalent to
the triangulation in terms of the running time.  

Some of our techniques may be helpful on solving the Euclidean version
of the problem. For the Euclidean version, as the $L_1$ version, 
a long-standing open
problem is to compute a shortest path in $O(n+h\log h)$ time and $O(n)$
space. Hershberger and Suri \cite{ref:HershbergerAn99} built an SPM of
size $O(n)$ in $O(n\log n)$ time and $O(n\log n)$ space.
Recently, Inkulu {\em et al.} announced an algorithm that can find 
an Euclidean shortest path in $O(n+h\log h\log n)$ time
\cite{ref:InkuluA10}; we give an algorithm for the problem that
runs in $O(n+h\log^{1+\epsilon}h+k)$ time \cite{ref:ChenCo11Curved}, where
$k$ is a parameter sensitive to the input and is bounded by $O(h^2)$.
Note that our algorithm \cite{ref:ChenCo11Curved} particularly works
for obstacles that have curved boundaries. To generalize the techniques
given in this paper to the Euclidean version, some difficulty appears.
For example, the idea of using cores does not seem to work (e.g., 
Lemma \ref{lem:40} is not applicable to the Euclidean
version). A possible direction for solving the Euclidean version is
to first solve the convex case with the performance desired by the
open problem. If this is possible, then by generalizing the techniques
given in this paper, it is very likely that 
the general case may also be solved accordingly, 
and thus the open problem can be settled although we may still have to 
suffer the $O(n+h\log^{1+\epsilon}h)$ triangulation time.


\footnotesize \baselineskip=11.0pt
\bibliographystyle{plain}
\bibliography{reference}

\newpage
\normalsize
\appendix



\end{document}